\documentclass[a4paper,runningheads,envcountsame,orivec]{llncs}
\usepackage[headings]{fullpage}
\usepackage[T1]{fontenc}
\usepackage[english]{babel}
\usepackage{amssymb,mathtools}
\usepackage{tikz-cd}
\usepackage{microtype}
\usepackage{needspace}
\usepackage{cite}
\usepackage{xurl}
\usepackage[unicode,hidelinks]{hyperref}
\title{Higher Tripos Theory}
\titlerunning{Higher Tripos Theory}
\author{Yoshihiro Maruyama}
\authorrunning{Yoshihiro Maruyama}
\institute{Kyoto University\\
\email{maruyama@i.h.kyoto-u.ac.jp}}
\newcommand{\C}{\mathcal C}
\newcommand{\B}{\mathcal B}
\newcommand{\D}{\mathcal D}
\newcommand{\E}{\mathcal E}
\newcommand{\HH}{\mathcal H}
\newcommand{\M}{\mathcal M}
\newcommand{\PP}{\mathbb P}
\newcommand{\QQ}{\mathbb Q}
\newcommand{\Spc}{\mathcal S}
\newcommand{\Cat}{\operatorname{Cat}_{\infty}}
\newcommand{\Map}{\operatorname{Map}}
\newcommand{\Fun}{\operatorname{Fun}}
\newcommand{\Sub}{\operatorname{Sub}}
\newcommand{\Eq}{\operatorname{Eq}}
\newcommand{\Iso}{\operatorname{Iso}}
\newcommand{\Aut}{\operatorname{Aut}}
\newcommand{\im}{\operatorname{im}}
\newcommand{\id}{\operatorname{id}}
\newcommand{\Loc}{\operatorname{Loc}}
\newcommand{\Ex}{\operatorname{Ex}_{\infty}}
\newcommand{\Exone}{\operatorname{Ex}_{1,\mathrm{reg}}}
\newcommand{\Sh}{\operatorname{Sh}_{\mathrm{reg}}}
\newcommand{\Eff}{\operatorname{Eff}}
\newcommand{\Asm}{\operatorname{Asm}}
\newcommand{\Ab}{\operatorname{Ab}}
\newcommand{\Cod}{\operatorname{Cod}}
\newcommand{\Tot}{\operatorname{TOT}}
\newcommand{\Ftwo}{\mathbb F_2}
\newcommand{\NN}{\mathbb N}
\newcommand{\RegD}{\mathsf{Doc}_{\mathrm{reg}}}
\newcommand{\Trip}{\mathsf{Trip}_{\infty}}
\newcommand{\Top}{\mathsf{ETop}_{\infty}}
\newcommand{\TT}{\mathsf T_{\infty}}
\newcommand{\core}{^{\simeq}}
\newcommand{\onto}{\twoheadrightarrow}
\newcommand{\into}{\hookrightarrow}
\newcommand{\codes}{\mathcal X}
\newcommand{\Pj}{\mathcal P}
\newcommand{\Qj}{\mathcal Q}
\newcommand{\PAsm}{\operatorname{PAsm}}
\newcommand{\sk}{\operatorname{sk}}
\newcommand{\cosk}{\operatorname{cosk}}
\newcommand{\Uni}{\operatorname{Uni}}

\begin{document}
\maketitle
\begin{abstract}
Tripos theory constructs categorical models of higher-order logic from indexed logical data. We develop higher tripos theory and higher tripos-to-topos construction recovering the classical construction on zero-types through a covering criterion. 
The higher tripos-to-topos construction yields elementary $\infty$-toposes with closed bounded hierarchies of univalent universes, thus establishing a logical adjunction between $\infty$-triposes and elementary $\infty$-toposes with univalent universes. 
We also construct a logically conservative higher completion of regular $1$-categories with weak section families and a finite-limit-closed projective cover. This completion preserves original maps, predicates, and existing dependent products, and gives criteria for extending interpretations while preserving logical structure. 
With strongly inaccessible size bounds, Kleene assemblies instantiate the construction, producing a cumulative hierarchy of univalent universes that classifies every family in the completion and recovering the effective topos as the category of zero-types. 
We then establish an arithmetical completeness result: for a uniformly computably presented family, the numerical index set for identity with a fixed classifier code is $\Pi^0_2$-complete. 
We also construct two inequivalent univalent classifiers with the same global homotopy type: every internal map from one to the natural-number object is constant, whereas the other admits a nonconstant such map. These reconstruction, preservation, and computability results contribute to the higher effective-topos program of Awodey-Emmenegger and Anel-Awodey-Barton.
\end{abstract}

\section{Introduction}\label{sec:intro}
Categorical semantics provides a mathematical framework for interpreting type theories and relating their logical structure to computation. A basic construction problem is to recover function types from representations of functions equipped with an evaluation operation. These representations, called function codes, can describe the same function in different ways. Reconstruction identifies such codes and, in dependent type theory, carries these identifications through substitution in every context. Higher semantics also retains identities between functions and their iterated identities~\cite{HoTT13}. We reconstruct dependent function types with their full higher identity structure through a higher tripos-to-topos construction. Applied to Kleene assemblies, the construction preserves their maps, predicates, and existing dependent products, and recovers the effective topos as the category of zero-types. These results contribute to the higher effective-topos program initiated by Awodey-Emmenegger~\cite{AwodeyEmmenegger25} and Anel-Awodey-Barton~\cite{Awodey25}.

\paragraph*{Background.}
Classical tripos theory by Hyland-Johnstone-Pitts~\cite{HJP80} constructs toposes from indexed logical structure and PER (partial equivalence relation) quotients (see also Pitts~\cite{Pitts02}). Frey develops its higher-dimensional universal property~\cite{Frey11}. Pasquali decomposes its logical reflection into comprehension, extensionality, quotients, and functional completeness~\cite{Pasquali16}. In Hyland's effective topos, every endomorphism of the natural-number object is tracked by a total recursive function~\cite{Hyland82}. Categorical realizability retains this uniform tracking through the passage from represented data to quotient types (see van Oosten~\cite{vanOosten08}). 
We extend this connection to type-valued doctrines, reconstructing entire section spaces while preserving the maps and predicates of $1$-categorical sources. The reconstruction of these section spaces builds upon weak function objects and exact quotients. 

Carboni and Rosolini studied local cartesian closure of exact completions~\cite{CarboniRosolini00}; Emmenegger gives criteria using extensional dependent products, together with closure criteria relating weak dependent products in a projective cover to local cartesian closure~\cite{EmmeneggerLocal20}. His fullness formulation also applies to exact completions of homotopy categories~\cite{EmmeneggerFull20}. For path categories, van den Berg and Moerdijk derive local cartesian closure of the homotopy exact completion from weak homotopy $\Pi$-types~\cite{BergMoerdijkPath18}. Den Besten studies strong homotopy exponentials and $\Pi$-types using groupoid enrichment~\cite{denBesten20}. Related criteria for local cartesian closure hold for elementary quotient completions~\cite{MaiettiPasqualiRosolini26}. In Frey's homotopical account of PER categories, mapping spaces are discrete~\cite{Frey23}. We develop reconstruction through every finite identity level: coherent fullness gives dependent products from higher section codes, and projective transfer derives those codes from weak section families in the source category. Higher exact completion~\cite{Stefanich23} provides the required effective groupoid quotients.
Our main results are as follows. 

\paragraph*{Dependent-Product Reconstruction.}
We reconstruct dependent function types from codes that locally represent every section. For any regular $\infty$-category and finite $n\ge-1$, coherent fullness at height $n$ is equivalent to the existence of all dependent products of $n$-truncated families in its higher exact completion  (Theorem~\ref{thm:relative-fullness}). The central idea is that identities between represented sections can themselves be described as sections of an identity family. Repeating this observation lowers truncation one level at a time, allowing us to recover the full higher identity structure in finitely many stages. For an $n$-truncated family, a propositional image followed by at most $n+1$ nested effective groupoid quotients suffices, starting from the given code tower and its input objects.

\paragraph*{Univalent Universe Reconstruction and Logical Adjunction.}
The reconstruction extends to universes and logical interpretations. Using coherent fullness, we reconstruct univalent universes from an exhaustive directed hierarchy of weak universe presentations (Theorem~\ref{thm:universes}). Five conditions characterize closure of each universe level under the logical operations (Theorem~\ref{thm:generic-tests}). Each level has its own finite truncation bound. Groupoid quotients complete the presented families into classifiers retaining all fibrewise equivalences. Among regular interpretations preserving level classes, preservation of weak section families is equivalent to dependent-product preservation; equality of chosen local classes detects universe preservation. These results establish a logical adjunction on full mapping spaces between $\infty$-triposes and elementary $\infty$-toposes with closed bounded hierarchies, with a fully faithful right adjoint. For an $\infty$-tripos over $\operatorname{Set}$ whose propositional shadow is a tripos, the zero-types of its completion recover the classical tripos-to-topos construction whenever every zero-type admits a cover from a subobject of a represented context (Theorem~\ref{thm:shadow}).

\paragraph*{Logically Conservative Higher Completion.}
We next construct higher logical structure from regular $1$-categories. Such a category with weak section families and a finite-limit-closed covering class of regular projectives admits higher section codes and a locally cartesian closed higher exact completion. Its embedding is fully faithful, identifies original subobject lattices, and preserves existing dependent products (Theorem~\ref{thm:pt-completion}). For regular interpretations carrying the source projectives into a finite-limit-closed covering class of zero-truncated target projectives, preservation of weak section families in the source characterizes extensions preserving dependent products, with an equivalence of full mapping spaces (Theorem~\ref{thm:pt-morphisms}). For arbitrary exact locally cartesian closed targets, the criterion uses preservation of chosen higher section codes (Theorem~\ref{thm:lc-morphisms}).

\paragraph*{Realizability Hierarchy.}
Kleene assemblies provide a concrete instance of this completion, yielding an exhaustive bounded realizability hierarchy (Theorem~\ref{thm:real-hierarchy}). 
The resulting cumulative univalent hierarchy classifies every map at some level, including nonmodest families. Zero-types recover the effective topos, and original numerical maps retain their trackers.

\paragraph*{Arithmetical Completeness and Rice-type Classifier Separation.}
The realizability construction also yields computational applications. For each $n\ge2$, a uniformly computably presented realizability family $X_{n,e}$, indexed by programs $\varphi_e$, has fibres
\begin{equation}\label{eq:intro-fibres}
 X_{n,e}\simeq
 \begin{cases}1,&\varphi_e\text{ is total},\\ K(\Ftwo,n),&\varphi_e\text{ is not total}.
 \end{cases}
\end{equation}
Its univalent classifier $U_n$ satisfies, with $\Gamma=\Map(1,-)$,
\begin{equation}\label{eq:intro-rice}
 \Gamma U_n\simeq1\amalg K(\Ftwo,n+1),
 \qquad \Map(U_n,2)\simeq2.
\end{equation}
Identity with one fixed code has a $\Pi^0_2$-complete numerical index set, and every map $U_n\to N$ to the natural-number object is constant (Theorem~\ref{thm:rice}). The two fixed representatives yield an inequivalent classifier with the same global homotopy type and a nonconstant numerical observation. The uniform family remains distinct from every pullback of the fixed two-fibre family after each cover of the parameter object. This connects Rice phenomena for universes and injective types~\cite{EscardoStreicher16,deJongEscardo26} to a uniform family of simply connected types.

\paragraph*{Function Types and Exact Quotient Depth.}
A second application determines how many quotient stages are needed to construct certain higher function types. For every $n\ge1$, we give a uniformly computably presented family requiring exactly $n$ nested quotient stages. Here construction starts from zero-types and allows finite limits, images, and effective groupoid quotients. An explicit product formula describes the resulting function types (Theorem~\ref{thm:cochain-products} in Appendix~\ref{app:cochains}). The fibre indexed by a program is contractible exactly when that program halts on every input. Thus the required quotient depth grows with $n$, while the problem of identifying contractible fibres remains $\Pi^0_2$-complete (Subsection~\ref{sec:cochains}).

\paragraph*{Related Work.}
Several constructions have been developed for univalent and higher realizability: van den Berg's univalent polymorphism~\cite{Berg20}, Uemura's cubical assemblies~\cite{Uemura19}, Swan and Uemura's analysis of Church's thesis~\cite{SwanUemura21}, and Speight's groupoidal realizability~\cite{Speight24}. Agwu constructs models in groupoids internal to assemblies, with univalent classifiers for assemblies and modest assemblies~\cite{Agwu25}; Francese outlines a homotopy-tripos program based on higher Turing categories and coherently tracked higher assemblies~\cite{Francese19}. Awodey and Emmenegger recover the effective topos as the category of zero-types of an effective $2$-topos~\cite{AwodeyEmmenegger25}. 
Awodey's presentation of joint work with Anel and Barton announces a higher effective topos with dependent products and univalent universes~\cite{Awodey25}. 

The impredicative univalent universes $D_n$ described
in~\cite{Awodey25} classify discrete $n$-types in the realizability
sense, whose homotopy objects, including $\pi_0$, lie in $D_0$,
the subquotients of $N$.
We contribute to the higher effective-topos program by constructing an
explicit cumulative hierarchy of univalent universes in which every family
in the completion, including nonmodest families, is classified at some level.
Each level bounds the truncation degree of the families it classifies and
is closed under dependent products of composable maps belonging to that level.
To construct this hierarchy, we verify the required source hypotheses using
finite diagrams of assemblies with uniformly tracked structure maps.
The transfer and preservation theorems then connect the resulting hierarchy
to assembly data and characterize how interpretations transport its logical
structure.

\paragraph*{Organization.}
Section~\ref{sec:reconstruction} develops regular type doctrines and their completions. Section~\ref{sec:triposes} introduces $\infty$-triposes and establishes the logical adjunction. Section~\ref{sec:completion} gives source conditions for conservative logical completion and univalent universe reconstruction, and Section~\ref{sec:realization} constructs the realizability instance. Section~\ref{sec:applications} establishes Rice-type classifier separation and reconstructs higher function types from tracked cochains, with an exact quotient-depth calculation. Full proofs are provided in the appendix.

\section{Regular Type Doctrines and Their Completions}\label{sec:reconstruction}
We introduce regular type doctrines and construct their higher exact completions. Local section codes then reconstruct dependent products at each finite relative truncation height.

$\Spc$ denotes the $\infty$-category of spaces, and $(-)\core$ denotes the maximal subgroupoid of an $\infty$-category. An ordinary category is a $1$-category regarded as an $\infty$-category through its nerve. A category is bounded if every object has a finite truncation bound, which may depend on the object. We work externally in classical ZF with a fixed hierarchy of nested Grothendieck universes and choice in each working universe. All sites, diagram categories, functor categories, and ambient classifiers are taken in specified members of this hierarchy; enlargements stay within it. 
These assumptions concern the external metatheory; neither excluded middle nor choice is assumed in the internal logic of the constructed categories.

We use the universe-based definition of elementary $\infty$-toposes: finite limits, local cartesian closure, and enough closed univalent universal families~\cite{ApolPetrowitsch26}.

\subsection{Regular Doctrines and Exact Completion}\label{sec:doctrine}\label{sec:presentations}

A regular $\infty$-category has finite limits and pullback-stable realizations of \v{C}ech nerves. It is exact if all internal Segal groupoids have effective quotients. A cover is an effective epimorphism; a regular functor preserves finite limits and covers. Higher exactness includes effective quotients of internal group nerves. For a small regular $\C$, put
\begin{equation}\label{eq:ambient}
 \HH_\C=\Sh(\C,\Spc),\qquad
 h:\C\into\HH_\C,\qquad
 j:\C\into\E:=\Ex(\C)\subseteq\HH_\C.
\end{equation}
Here $\Sh(\C,\Spc)$ denotes the $\infty$-category of space-valued sheaves for the regular topology generated by singleton covers, $h$ is the Yoneda embedding $A\mapsto\Map_{\C}(-,A)$, and $j$ is the same functor with codomain restricted to $\E$.
Higher exact completion is the full closure of the representables under groupoid realizations. It is finite-limit closed and exact, its unit is regular and fully faithful, and every object has an original atlas $jA\onto X$. Restriction along $j$ is an equivalence of regular functor categories into any exact target~\cite{Stefanich23}. Section codes and weak universe presentations extend this regular construction to dependent products and univalent classifiers.

\begin{definition}[Regular type doctrine]\label{def:doctrine}
A regular type doctrine is a functor $\PP:\B^{op}\to\Cat$, with $\B$ finitely complete, such that every fibre and reindexing is regular; every reindexing has a left adjoint $\Sigma_f\dashv f^*$ satisfying Beck--Chevalley; and each $\Sigma_f$ preserves pullbacks and has cartesian unit and counit. All structures and comparisons are coherent. Write $\C_\PP=\PP(1)$ and $\delta_\PP X=\Sigma_{!_X}1_X$. The doctrine is exact when $\C_\PP$ is exact.
\end{definition}

\begin{proposition}[Comprehension]\label{prop:moens}
The functor $\delta_\PP$ preserves finite limits, and there are natural equivalences
\begin{equation}\label{eq:moens}
 \PP(X)\simeq(\C_\PP)_{/\delta_\PP X}.
\end{equation}
They send $A$ to $(\Sigma_{!_X}A\to\delta_\PP X)$, identify $\Sigma_f$ with composition by $\delta_\PP f$, and identify $f^*$ with pullback.
\end{proposition}
This is higher Moens reconstruction~\cite{Weinberger24}; Appendix~\ref{app:moens} gives the explicit proof. Its propositional shadow is
\begin{equation}\label{eq:shadow}
 \operatorname{Prop}_\PP(X)=\Sub_{\PP(X)}(1_X)
       \simeq\Sub_{\C_\PP}(\delta_\PP X).
\end{equation}
Existential quantification is the image of a dependent sum.

A regular doctrine map consists of a finite-limit-preserving base functor and a coherent indexed transformation preserving finite limits, covers, and dependent-sum adjunctions. Write $\RegD$ for the resulting $\infty$-category. Under comprehension, such a map $\PP\to\QQ$ is a triple $(v,F,\theta)$ with $F:\C_\PP\to\C_\QQ$ regular and $\theta:F\delta_\PP\simeq\delta_\QQ v$; see Lemma~\ref{lem:doctrine-maps}.

\begin{theorem}[Free exact completion of a regular doctrine]\label{thm:regular-completion}
For every regular type doctrine $\PP$,
\begin{equation}\label{eq:regular-completion}
 (\mathsf Q\PP)(X)=\Ex(\PP(1))_{/j\delta_\PP X}
\end{equation}
is an exact regular type doctrine. For every exact regular type doctrine $\mathbb D$, restriction along the unit $\PP\to\mathsf Q\PP$ induces
\begin{equation}\label{eq:q-universal}
 \Map_{\RegD}(\mathsf Q\PP,\mathbb D)
       \simeq\Map_{\RegD}(\PP,\mathbb D).
\end{equation}
The unit is the identity on the base and is fully faithful in every fibre.
\end{theorem}

\begin{theorem}[Slicing and fibrewise exact completion]\label{thm:slice-completion}
For every regular $\infty$-category $\C$ and $T\in\C$, the regular slice embedding induces an equivalence
\begin{equation}\label{eq:slice-completion}
 \Ex(\C_{/T})\simeq(\Ex\C)_{/jT},
\end{equation}
natural under pullback in $T$. Consequently the doctrine of Theorem~\ref{thm:regular-completion} is the fibrewise higher exact completion:
\begin{equation}\label{eq:fibrewise-completion}
 (\mathsf Q\PP)(X)\simeq\Ex(\PP(X)).
\end{equation}
\end{theorem}
The proofs in Appendix~\ref{app:regular-completion} use the regular universal property and closure under groupoid realizations.

A type presentation over $\C$ is an object of $\Ex(\C)$, with witnesses of its generation from original objects by groupoid realizations. Maps between type presentations are maps between their realizations in $\HH_\C$, with all higher homotopies.

A relative atlas presentation of finite height describes a truncated family
by a tower of original atlases and their relative kernels.
Here finiteness refers to the height of the atlas recursion.
Every finitely truncated
presented family over an original base has such a description. If the
source is bounded, its completion is bounded and finite-atlas objects
exhaust it.
Appendix~\ref{app:lc-atlases} gives the recursive definition and
recognition lemma. The reconstruction below ranges over all presented types,
with its height parameter given by the relative truncation of the family whose dependent product is being reconstructed.

\subsection{Dependent Products and Univalent Completion}\label{sec:fullness}\label{sec:univalent-completion}
\begin{definition}[Weak section family]\label{def:weak}
For $Z\xrightarrow gY\xrightarrow fX$ in a regular $\infty$-category, a weak section family is a map $w:W\to X$ and an evaluation $e:W\times_XY\to Z$ over $Y$ such that every section $s:T\times_XY\to Z$ is locally represented: there are a cover $a:T'\onto T$, a map $v:T'\to W$ over $X$, and a homotopy over $Y$
\begin{equation}\label{eq:weak}
 e\circ(v\times_XY)\simeq s|_{T'\times_XY}.
\end{equation}
This holds for every $T\to X$.
\end{definition}
We call $W$ the code object and refer to the family together with its evaluation as a family of section codes. A map $v:T\to W$ over $X$ codes the section $e\circ(v\times_XY)$. We refer to this local representation property as weak adequacy.

\begin{definition}[Coherent fullness]\label{def:lc-codes}\label{def:coherent-fullness}
Let $\C$ be regular and $n\ge-1$ finite. The condition $\mathrm{CF}_n(\C)$ asserts that for every $T\in\C$ and every diagram of presented types $Z\xrightarrow gY\xrightarrow f jT$ with $g$ $n$-truncated, there are an original map $w:W\to T$ and an evaluation $e:jW\times_{jT}Y\to Z$ over $Y$ with the following property. For every original $t:S\to T$ and section $s:jS\times_{jT}Y\to Z$, there are a cover $a:S'\onto S$ in $\C$, an original $v:S'\to W$ over $T$, and a homotopy over $Y$
\begin{equation}\label{eq:lc-codes}
 e\circ(jv\times_{jT}Y)\simeq s|_{jS'\times_{jT}Y}.
\end{equation}
Write $\mathrm{CF}_{<\infty}$ for the condition that $\mathrm{CF}_n$ holds for every finite $n\ge-1$. Unless a height is specified, coherent fullness means $\mathrm{CF}_{<\infty}$. When interpretations are considered, fix one code family for each such problem.
\end{definition}
The original objects $W,S,S'$ and their maps are source data; $Y,Z$ may be higher presentations. For bounded sources, Lemma~\ref{lem:lc-atlas} identifies $\mathrm{CF}_{<\infty}$ with source-local coding of every finite-atlas section problem. Because an $n$-truncated map is also $(n+1)$-truncated, $\mathrm{CF}_n$ includes all the lower identity problems needed below.

\begin{theorem}[Relative finite-height reconstruction]\label{thm:relative-fullness}
For any regular $\infty$-category $\C$ and finite $n\ge-1$, the following are equivalent:
\begin{enumerate}
\item $\mathrm{CF}_n(\C)$ holds.
\item For every $Z\xrightarrow gY\xrightarrow fX$ in $\E=\Ex(\C)$ with $g$ $n$-truncated, the dependent product $\Pi_fg$ exists in $\E$.
\end{enumerate}
These products agree with those in $\HH_\C$ and satisfy Beck--Chevalley. For every $A\to X$ the full representing equivalence is
\begin{equation}\label{eq:relative-product}
 \Map_{\E/X}(A,\Pi_fg)
       \simeq\Map_{\E/Y}(A\times_XY,Z).
\end{equation}
Starting from the chosen code families for the section problem and its successive identity problems, together with their input objects, the reconstruction adds at most $n+1$ nested effective groupoid quotients over arbitrary completed bases.
\end{theorem}
The proof first extends weak adequacy to every completed context, then inducts simultaneously over all bases on the relative truncation of $g$. For propositional $g$, the image of the projection from the code object to the base represents the dependent product. At each higher stage, identities between section codes form a section problem of one lower height. For two evaluations $e_0,e_1:T\times_XY\to Z$ over $Y$, their identity family is the pullback of $\Delta_g:Z\to Z\times_YZ$ along $(e_0,e_1)$; its sections are homotopies over $Y$ between the evaluations. Reconstructing this kernel makes its entire \v{C}ech groupoid available for the next quotient. Appendix~\ref{app:relative-fullness} gives the full argument.

We next use these dependent products to construct a univalent classifier for an individual truncated family.
A local class is stable under pullback, and membership descends along covers of the codomain.
For a family $p:E\to V$ in an exact category $\E$, let $\Loc(p)$ consist of the maps becoming pullbacks of $p$ after a cover of their codomain. A family $u:\widetilde U\to U$ is univalent if pullback induces fully faithful maps of spaces
\begin{equation}\label{eq:univalence}
 \Map_\E(T,U)\longrightarrow(\E/T)\core.
\end{equation}
Classification therefore includes fibre equivalences and all their higher identities~\cite{GepnerKock17,Rasekh22}.

\begin{theorem}[Univalent completion of an individual family]\label{thm:univalent-completion}
Assume $\mathrm{CF}_n(\C)$, and let $p:E\to V$ be an $n$-truncated family in $\E=\Ex(\C)$. There is a univalent family $\Uni(p)=u:\widetilde U\to U$ and a cartesian square from $p$ to $u$ whose base map is a cover $V\onto U$. For every $T\in\E$,
\begin{equation}\label{eq:universe-classification}
 \Map_\E(T,U)\simeq(\E/T)_{\Loc(p)}\core.
\end{equation}
The family $u$ is $n$-truncated and $U$ is $(n+1)$-truncated. If $\Loc(p)$ is all monomorphisms, then $u$ is a subobject classifier $1\into\Omega$.
\end{theorem}
Coherent fullness at height $n$ gives the fibrewise mapping objects. Realizing the groupoid of fibre equivalences and its tautological action gives the univalent classifier~\cite{BergMoerdijk18,Stenzel23}. Appendix~\ref{app:universes} gives the construction.

\section{\texorpdfstring{$\infty$}{∞}-Triposes: Construction and Universal Property}\label{sec:triposes}
The input to the construction consists of regular types, coherent section
codes, and weak universe presentations. The logical closure of a universe
level is tested on finitely many generic families. Quotients and the types
represented by the codes belong to the completed category.

\subsection{Weak Universes and Logical Closure}\label{sec:generic-closure}
Let $\C$ be regular and put $\E=\Ex(\C)$. For a presented family $p:P\to jV$ with
$V\in\C$, write $q\preceq_\C p$ if, for every original test
$jT\to\operatorname{cod}(q)$, there are a source cover $T'\onto T$,
a source map $v:T'\to V$, and an equivalence
\begin{equation}\label{eq:lc-local-class}
 q|_{jT'}\simeq v^*p\quad\text{over }jT'.
\end{equation}
Original atlases identify this condition with membership in
$\Loc_\E(p)$, the full replete class of families locally pulled back
from $p$. This condition is local for covers of arbitrary completed bases.

\begin{proposition}[Higher families and source-generated classes]
\label{prop:source-generics-obstruction}
Let $\C$ be an ordinary regular category containing a nontrivial internal
group $G$. For every collection of source families $p_i:E_i\to V_i$,
the classes $\Loc_{\Ex(\C)}(jp_i)$ together form a proper subclass of
all arrows of $\Ex(\C)$, with $BG\to1$ witnessing the strict inclusion.
\end{proposition}
The quotient $BG\to1$ has loop object $jG$, and source maps and their
local pullbacks are $0$-truncated. Presented total spaces therefore extend
weak universe classification to higher families over original parameter objects.

\begin{definition}[Weak universe hierarchy]\label{def:weak-hierarchy}
A weak universe hierarchy consists of a directed index set $I$ with a
distinguished element $\mathrm{prop}$, finite integers $n_i\ge-1$, and
$n_i$-truncated presented families $p_i:P_i\to jV_i$, with $V_i\in\C$.
Writing $\M_i=\Loc_\E(p_i)$, require:
\begin{enumerate}
\item $\M_i\subseteq\M_k$ for $i\le k$;
\item every presented family over an original base belongs to some $\M_i$;
\item $n_{\mathrm{prop}}=-1$, and every monomorphism over an original
base belongs to $\M_{\mathrm{prop}}$.
\end{enumerate}
\end{definition}

For the generic constructions and closure conditions below, assume coherent fullness
$\mathrm{CF}_{<\infty}(\C)$. For each level,
Theorem~\ref{thm:univalent-completion} gives
$u_i:=\Uni(p_i):\widetilde U_i\to U_i$, with $U_i$
$(n_i+1)$-truncated. Form
\begin{equation}\label{eq:generic-container}
 d_i:D_i\longrightarrow U_i,
 \qquad d_i:=\Pi_{u_i}(\widetilde U_i\times U_i\to\widetilde U_i).
\end{equation}
The second map is $(n_i+1)$-truncated, so
Theorem~\ref{thm:relative-fullness} constructs this product from
$\mathrm{CF}_{n_i+1}$. Its evaluation is a map
$\epsilon_i:A_i:=D_i\times_{U_i}\widetilde U_i\to U_i$. Define
\begin{equation}\label{eq:generic-pair}
 B_i:=A_i\times_{\epsilon_i,U_i}\widetilde U_i
 \xrightarrow{g_i}A_i\xrightarrow{f_i}D_i,
 \qquad \sigma_i=f_ig_i,
 \qquad \rho_i=\Pi_{f_i}g_i.
\end{equation}
Both members of this composable pair lie in $\M_i$. Coherent fullness
$\mathrm{CF}_{n_i}$ constructs $\rho_i$ in the completion. The dependent-product
condition below characterizes its membership in $\M_i$.

\begin{theorem}[Five conditions for universe closure]\label{thm:generic-tests}
For a weak universe level, the following conditions are equivalent:
\begin{enumerate}
\item $\M_i$ contains identities and is closed under image inclusions,
relative diagonals, composition, and dependent products of composable
maps in $\M_i$;
\item the following five conditions for universe closure hold:
\begin{equation}\label{eq:generic-tests}
 \begin{gathered}
 \id_1\preceq_\C p_i,\qquad
 (\im(p_i)\into jV_i)\preceq_\C p_i,\qquad
 \Delta p_i\preceq_\C p_i,\\
 \sigma_i\preceq_\C p_i,\qquad \rho_i\preceq_\C p_i.
 \end{gathered}
\end{equation}
\end{enumerate}
Composition and dependent-product closure are separately equivalent to
the conditions on $\sigma_i$ and $\rho_i$, respectively. The comparison
includes the canonical substitution equivalences.
\end{theorem}
Every composable pair in $\M_i$ is classified by a map into $D_i$.
The last two conditions therefore imply closure by pullback and
Beck--Chevalley.

\begin{definition}[$\infty$-tripos]\label{def:tripos}
An $\infty$-tripos is a regular type doctrine $\PP$
whose closed-type category $\C_\PP=\PP(1)$ satisfies coherent fullness,
equipped with a weak universe hierarchy satisfying the five conditions for universe
closure~\eqref{eq:generic-tests} at every level.
\end{definition}
Exhaustiveness and the finite truncation bound at each level imply boundedness of both source and completion.

\begin{definition}[Closed bounded hierarchy]\label{def:closed-hierarchy}
On an exact locally cartesian closed $\infty$-category, a closed bounded
hierarchy is an exhaustive directed family of full replete local classes
$\M_i$ of $n_i$-truncated maps, with finite $n_i\ge-1$, each classified
by a selected univalent universal family. Each class contains identities
and is closed under image inclusions, relative diagonals, composition, and
same-level dependent products. The distinguished class
$\M_{\mathrm{prop}}$ consists of all monomorphisms and has bound $-1$.
\end{definition}
Section~\ref{sec:completion} develops universe levels closed under bounded groupoid quotients.

\begin{theorem}[$\infty$-tripos-to-topos construction and universe reconstruction]
\label{thm:tripos-construction}\label{thm:universes}
For an $\infty$-tripos $\PP$, put
\[
 \TT\PP:=\Ex(\PP(1)),\qquad \M_i=\Loc_{\TT\PP}(p_i).
\]
Then $\TT\PP$ is bounded, exact, and locally cartesian closed. The
univalent completions $u_i:\widetilde U_i\to U_i$ form a closed bounded
hierarchy, with equivalences
\begin{equation}\label{eq:classifying}
 \Map_{\TT\PP}(X,U_i)
 \simeq\bigl((\TT\PP)_{/X}\bigr)_{\M_i}\core.
\end{equation}
The family $u_i$ is $n_i$-truncated, $U_i$ is $(n_i+1)$-truncated,
and $u_{\mathrm{prop}}$ is $1\into\Omega$. Class inclusions induce
coherent cumulative maps and cartesian comparisons of universal families.
Consequently $\TT\PP$ is an exact elementary $\infty$-topos.
\end{theorem}

\begin{corollary}[Completed logical operations]\label{cor:lcc}
The completed doctrine
\[
 (\mathsf Q\PP)(X)=(\TT\PP)_{/j\delta_\PP X}
\]
has both adjoints to every reindexing. Its unit is fibrewise fully
faithful, identifies the subobject lattices of source fibre objects,
and preserves every dependent product already existing in the source.
Source predicate equality, finite conjunction, implication, and
existential and universal quantification are preserved and reflected.
\end{corollary}

\subsection{Logical Morphisms and the Fundamental Adjunction}
\label{sec:fundamental}
For regular $F:\C\to\D$ with exact target, exact completion gives
a regular extension $\bar F:\Ex(\C)\to\D$ that interprets presented
types independently of presentation, including all coherent comparisons.

Say that $F$ preserves selected coherent-fullness codes if their
interpreted evaluations are weak section families in $\D$, tested
against every object of $\D$.

\begin{theorem}[Criterion for preserving dependent products]\label{thm:lc-morphisms}
Let $\C$ be bounded regular with coherent fullness, and let $\D$ be
exact and locally cartesian closed. A regular $F:\C\to\D$ preserves
the selected coherent-fullness codes if and only if $\bar F$ preserves
all dependent products. Restriction gives
\begin{equation}\label{eq:lc-logical-up}
 \Fun_{\mathrm{reg},\Pi}(\Ex(\C),\D)\core
 \simeq\Fun_{\mathrm{reg},\mathrm{codes}}(\C,\D)\core.
\end{equation}
For each code-preserving $F$, the space of extensions with an identification
of their restriction with $F$ is contractible.

More generally, for any regular $\C$ satisfying $\mathrm{CF}_n$, with
finite $n\ge-1$, preservation of its selected $\mathrm{CF}_n$ codes is
equivalent to preservation by $\bar F$ of all products whose second map
is $n$-truncated. The corresponding restriction equivalence holds with
both sides at this finite height.
\end{theorem}
The product comparison is a cover by code preservation. Its diagonal is
the comparison for a lower-truncated identity family; finite induction
therefore makes it monic as well. The criterion is independent of the
choice of adequate code witnesses.

\begin{theorem}[Dependent-product preservation]\label{thm:preservation}
Let $F:\E\to\D$ be regular between exact categories carrying closed
bounded hierarchies, and suppose it preserves level classes. Then $F$
preserves every weak section family whose pair and code-family structure map belong to
one level if and only if it preserves all dependent products.
\end{theorem}

\begin{definition}[Logical morphism]\label{def:logical}
Fix a common hierarchy index set. A regular doctrine map between $\infty$-triposes induces $\widehat F:\Ex(\C_\PP)\to\Ex(\C_\QQ)$ by exact
completion. It is logical when its selected coherent-fullness
evaluations become weak section families in the target completion and
\begin{equation}\label{eq:class-equality}
 \Loc_{\Ex(\C_\QQ)}(\widehat Fp_i)=\M_i^{\QQ}
 \qquad(i\in I).
\end{equation}
\end{definition}
These are source-local conditions as well: adequacy can be tested using
objects and covers of the target source category $\C_\QQ$, and class equality is
equivalent to mutual local presentation of the two generating families
after such covers. Original atlas refinement proves these equivalences.
The preservation theorem identifies logicality with preservation by
$\widehat F$ of products and the chosen univalent families, so logical
maps are closed under identities and composition.

Let $\Trip$ be the $\infty$-category of $\infty$-triposes and logical
maps. Let $\Top$ have exact locally cartesian closed $\infty$-categories
with closed bounded hierarchies as objects, and regular functors
preserving dependent products and the chosen universal families as maps.
The specified hierarchy is part of the structure.

\begin{theorem}[Fundamental adjunction]\label{thm:fundamental}
There is an adjunction with fully faithful right adjoint
\begin{equation}\label{eq:adjunction}
 \TT:\Trip\rightleftarrows\Top:\Cod,
 \qquad\TT\PP=\Ex(\PP(1)),\qquad\Cod(\E)(X)=\E_{/X}.
\end{equation}
Its full mapping-space equivalence is
\begin{equation}\label{eq:mapping-adjunction}
 \Map_{\Top}(\Ex(\PP(1)),\E)
 \simeq\Map_{\Trip}(\PP,\Cod(\E)).
\end{equation}
The unit has base functor $j\delta_\PP$ and fully faithful fibre functors
$\PP(X)\simeq\C_\PP{}_{/\delta_\PP X}
\into\Ex(\C_\PP)_{/j\delta_\PP X}$.
They identify source subobject lattices and preserve existing products.
The counit is the equivalence $\Ex(\E)\simeq\E$.
\end{theorem}

The adjunction factors through the base-preserving reflection into exact
$\infty$-triposes and exhibits $\Top$ as the localization at logical maps
inducing equivalences of completed semantics; see
Corollaries~\ref{cor:exact-reflection} and~\ref{cor:morita} in
Appendix~\ref{app:adjunction}.

\section{Criteria for Conservative Logical Completion}\label{sec:completion}
Projective covers transfer ordinary section codes to
higher presentations, and universe presentations satisfying (H1)--(H3) keep section
codes and their quotients in one level.

For a bounded regular source $\C$, the condition
$\mathrm{CF}_{<\infty}(\C)$ holds exactly when
$\Ex(\C)$ is locally cartesian closed. For every regular source, the
unit is fully faithful, identifies source subobject lattices, and
preserves each existing dependent product.

\begin{definition}[A source with a projective cover]\label{def:pt-source}
A \emph{source with a projective cover} is a pair $(\C,\Pj)$, where
$\C$ is an ordinary regular category in which every composable pair has
a weak section family as in Definition~\ref{def:weak}, and
$\Pj\subseteq\C$ is a full replete subcategory, closed under finite
limits, consisting of regular projective objects, such that every
$X\in\C$ admits a regular epimorphism $P\onto X$ with $P\in\Pj$.
An object $P$ is \emph{regular projective} if every map $P\to B$ lifts
along every regular epimorphism $A\onto B$.
\end{definition}

\begin{theorem}[Logical completion from projective covers]\label{thm:pt-completion}
For a source $(\C,\Pj)$ with a projective cover as in Definition~\ref{def:pt-source}:
\begin{enumerate}
\item[(i)] $\C$ satisfies $\mathrm{CF}_{<\infty}$. Consequently
$\E=\Ex(\C)$ is exact, bounded, and locally cartesian closed; its
dependent products agree with the ambient ones and satisfy
Beck--Chevalley.
\item[(ii)] The unit $j:\C\to\E$ is fully faithful, induces
$\Sub_{\C}(X)\cong\Sub_{\E}(jX)$ naturally in $X$, and preserves
every existing dependent product. On predicates in original contexts,
it preserves and reflects equality, finite conjunction, implication,
and existential and universal quantification.
\end{enumerate}
\end{theorem}

For an ordinary regular $\C$ with a finite-limit-closed covering class
of regular projectives, write $\Exone(\C)$ for its ordinary exact/regular
completion. Then $\Exone(\C)$ and $\Ex(\C)$ are locally cartesian closed
exactly when every composable pair in $\C$ has a
weak section family (Corollary~\ref{cor:classical-lifting},
Appendix~\ref{app:projective-codes}). Universe presentations satisfying (H1)--(H3)
(Definition~\ref{def:lc-universes}) equip $\Ex(\C)$ with a closed
bounded hierarchy.

In a regular $\infty$-category, an object $P$ is projective if every cover $A\onto B$ induces a surjection $\pi_0\Map(P,A)\to\pi_0\Map(P,B)$.

For an ordinary source as in Definition~\ref{def:pt-source}, let an
exact locally cartesian closed target have a full replete,
finite-limit-closed covering class of $0$-truncated projectives.
A regular interpretation carrying the source projectives into that
class extends to preserve dependent products if and only if it
preserves source weak section families. Adequacy in the target is
tested against all target objects.
Corollary~\ref{cor:lc-doctrine} gives the corresponding completion
of regular type doctrines.

\subsection{Universe Presentations and Quotient Stability}

\begin{definition}[Quotient-stable level]\label{def:quotient-stable}
A local class $\M$ of $N$-truncated maps in an exact category is
quotient-stable at bound $N$ if, for every internal groupoid
$R_\bullet$ over $X$, the conditions $R_k\to X\in\M$ for all $k$ and
$|R_\bullet|\to X$ $N$-truncated imply $|R_\bullet|\to X\in\M$.
A hierarchy is quotient-stable if this holds at each specified bound.
\end{definition}

Let $\C$ be a bounded regular $\infty$-source. Fix a directed set $I$ and
finite-atlas families $p_i:E_i\to hV_i$, with $V_i\in\C$ and finite relative
bounds $n_i$. At the distinguished propositional level use the bound
$n_{\mathrm{prop}}=-1$.

In this bounded setting, a finite-atlas arrow means a morphism between finite-atlas objects; by Lemma~\ref{lem:lc-atlas}, these are precisely the arrows of $\E=\Ex(\C)$.
\begin{definition}[Conditions on universe presentations]\label{def:lc-universes}
Let $\M_i$ consist of the finite-atlas arrows $q$ with
$q\preceq_\C p_i$, as in~\eqref{eq:lc-local-class}; thus
$\M_i=\Loc_\E(p_i)$.
We require these presentations to satisfy conditions (H1)--(H3).
\begin{enumerate}
\item[(H1)] The classes increase with $I$. Every finite-atlas family over an original base belongs to $\M_i$ for some $i\in I$. The restriction of $\M_{\mathrm{prop}}$ to arrows over original bases consists exactly of the monomorphisms over those bases.
\item[(H2)] Each $\M_i$ contains identities and is closed under composition, relative diagonals, and image inclusions. If a groupoid of finite-atlas objects over a finite-atlas base $X$ has every term in $\M_i$ over $X$ and its ambient quotient is $n_i$-truncated over $X$, its quotient family belongs to $\M_i$.
\item[(H3)] For every diagram $Z\xrightarrow{g}Y\xrightarrow{f}hT$ with $f,g\in\M_i$, a section-code witness as in Definition~\ref{def:coherent-fullness} can be chosen with the structure map $hW\to hT$ of its code family in $\M_i$.
\end{enumerate}
\end{definition}

The quotient in (H2) is formed in the ambient regular-sheaf
$\infty$-topos $\HH_\C$.
Conditions (H1)--(H3) provide source section codes:
put $Y\to hT$ and $Z\to hT$ in a common level, use the graph of
$Z\to Y$ and relative-diagonal closure to put that map in the same level,
and then apply (H3).

\begin{theorem}[Sufficient conditions for a quotient-stable univalent hierarchy]
\label{thm:controlled-source}
Let $\C$ be bounded regular and let $(p_i)$ satisfy
(H1)--(H3) of Definition~\ref{def:lc-universes}.
Then $\mathrm{CF}_{<\infty}(\C)$ holds and these presentations
satisfy the five conditions for universe closure. Hence $\Cod(\C)$ is an $\infty$-tripos, and $\Ex(\C)$ carries the closed bounded hierarchy
$\M_i=\Loc_{\Ex(\C)}(p_i)$. Each level is quotient-stable
at its specified bound, and the hierarchy has the structure
specified in Definition~\ref{def:classes}.
\end{theorem}
Appendix~\ref{app:lc-universes} proves this criterion and the classification
and interpretation results of Corollary~\ref{cor:lc-universes}.
It also gives the corresponding result for completing regular type doctrines and the restriction
of the adjunction to quotient-stable hierarchies.
Appendix~\ref{app:fullness} gives the same-level reconstruction
theorem for exact categories.

\subsection{Recovery of the Classical Tripos-to-Topos Construction}\label{sec:classical}
\begin{theorem}[Recovery of the classical construction on zero-types]\label{thm:shadow}
Let $\PP$ be an $\infty$-tripos over $\operatorname{Set}$ and put
$\E=\Ex(\C_\PP)$. Suppose $H(I)=\Sub_{\C_\PP}(\delta_\PP I)$ is a
propositional tripos. Suppose every $0$-truncated object of $\E$ has a cover
from a subobject of some $j\delta_\PP I$. Then
\begin{equation}\label{eq:classical}
 \E_{\le0}\simeq\operatorname{Set}[H],
\end{equation}
where $\operatorname{Set}[H]$ denotes the topos obtained from $H$ by the classical tripos-to-topos construction.
\end{theorem}
Subobject conservativity makes the propositional shadow identical before and
after completion. Partial equivalence relations have effective
quotients that are zero-types in $\E$; the covering condition gives essential surjectivity of the
classical comparison. For an ordinary source in which every object has a cover from a subobject
of some $\delta_\PP I$, original atlases imply this condition.

\begin{example}[Bounded spaces]\label{ex:spaces}
Let $\kappa_0<\kappa_1<\cdots$ be strongly inaccessible cardinals and let
$\C_{\mathrm{set}}$ be a skeleton of sets of cardinality below some $\kappa_i$.
Then $\Ex(\C_{\mathrm{set}})$ is the category of finitely truncated spaces with
a simplicial presentation of size below some $\kappa_i$. At level $(i,n)$,
take maps with $\kappa_i$-small $n$-truncated fibres, and take all monomorphisms
as the propositional level. These classes have universe presentations
satisfying (H1)--(H3). Thus $\Cod(\C_{\mathrm{set}})$ is an $\infty$-tripos and its completed type
category has the corresponding closed bounded hierarchy. The quotient
$B(\mathbb Z/2)$ is a $1$-type that is not $0$-truncated.
Proposition~\ref{prop:space-model} proves all the size and presentation claims.
\end{example}

Proposition~\ref{prop:higher-nonexact} gives a bounded regular higher
source satisfying (H1)--(H3) whose exact completion is a strict enlargement.

\begin{corollary}[Separation of logical and quotient closure]
\label{cor:quotient-separation}
The bounded-space model of Example~\ref{ex:spaces} admits an exhaustive
closed bounded hierarchy containing the level of maps whose fibres
are $1$-types with finitely many components and fundamental groups
of finite $2$-power order. The bar quotient $BC_3$ witnesses a strict enlargement
of this level under $1$-truncated groupoid quotients.
\end{corollary}
Every bar-groupoid term $R_k=C_3^k$ is a finite set in this level,
but its $1$-truncated quotient $BC_3$ has fundamental group $C_3$
of order $3$.

\section{Realizability as Logical Completion}\label{sec:realization}
The projective transfer theorem applies to assemblies and gives
the higher presentation codes required in Section~\ref{sec:completion}.
We now verify an exhaustive bounded universe hierarchy in that completion.

Fix an acceptable numbering $(\varphi_e)$ of partial recursive functions. Kleene's first partial combinatory algebra is $K_1=(\NN,\cdot)$, with partial application $e\cdot r=\varphi_e(r)$.
An assembly consists of a carrier set and a nonempty set of numerical realizers for each element. A morphism of assemblies is a carrier function tracked by one partial recursive function defined on every input realizer and sending it to a realizer of the corresponding output. An assembly is modest if distinct carrier elements have disjoint sets of realizers.

Fix an increasing sequence of strongly inaccessible cardinals
\(\kappa_0<\kappa_1<\cdots\) contained, as a sequence, in a working
external universe.
Throughout this section, \(\Asm=\Asm(K_1)\) consists of Kleene assemblies
whose carriers have cardinality below some \(\kappa_i\).
Write \(\PAsm\) for its full subcategory of partitioned assemblies.
An object of \(\PAsm\) is a carrier
\(A\) with a name function \(\nu_A:A\to\NN\); the singleton
\(\{\nu_A(a)\}\) realizes \(a\). Distinct carrier elements may share the same numerical name. Carrier functions specify the maps, and their numerical trackers realize them.
Write \(\Eff=\Exone(\Asm)\) for its ordinary exact/regular completion.

Put \(\E=\Ex(\Asm)\). Lemma~\ref{real:completion-comparison}
gives \(\E\simeq\Ex(\Eff)\), using regular exact
completion~\cite{Stefanich23,vanOosten08}.
When viewing assemblies as a source with a projective cover, the selected projective
class is the replete closure of \(\PAsm\). We retain actual
partitioned representatives for all constructions below that use numerical names.

\begin{lemma}[Assembly source]\label{real:ordinary-source}
\(\Asm\) is an ordinary regular category, and \(\PAsm\) is a
finite-limit-closed covering subcategory of regular projectives.  Every
composable pair in \(\Asm\) has a dependent product, hence
a weak section family as in Definition~\ref{def:weak}.
\end{lemma}

Thus Theorem~\ref{thm:pt-completion} derives higher section codes and
dependent products from the assembly source. The construction of
Example~\ref{ex:spaces} also applies to the numerically named two-element
group assembly \(C_2\), giving \(BC_2\) with
\(1\times_{BC_2}1\simeq C_2\). Proposition~\ref{real:bc2} shows that this is a new higher type.

\begin{definition}[Small projective atlas classes]\label{real:def-small}
For a selected cardinal \(\kappa\), we set as follows:
\(\mathcal S_{\kappa,-1}=\{\text{monomorphisms in }\E\}\).
For \(n\ge0\), an \(n\)-truncated map \(f:Y\to X\) belongs to
\(\mathcal S_{\kappa,n}\) if its relative diagonal belongs to
\(\mathcal S_{\kappa,n-1}\), and for every \(P\in\PAsm\) and
\(P\to X\) there is a cover \(A\onto Y\times_XP\), with
\(A\in\PAsm\), such that every carrier fibre of \(A\to P\)
has cardinality less than \(\kappa\).  Such a cover is a
\(\kappa\)-small projective atlas over \(P\).
\end{definition}
These carrier bounds measure the size of projective presentations and
their successive identity atlases. They provide the size control used to
construct function objects and universe classifiers.

\begin{theorem}[Conservative realizability hierarchy]\label{thm:real-hierarchy}
The completion \(\E\) has the following properties.
\begin{enumerate}
\item The classes \(\mathcal S_{\kappa_i,n}\), for \(n\ge0\),
together with the distinguished class of all monomorphisms, form an
exhaustive bounded hierarchy with the structure specified in
Definition~\ref{def:classes}.  They are increasing in both indices,
local for covers, and closed under same-level dependent products.  More
precisely, if \(f,g\in\mathcal S_{\kappa_i,N}\) and \(g\) is
\(n\)-truncated with \(-1\le n\le N\), then
\(\Pi_fg\in\mathcal S_{\kappa_i,n}\), with Beck--Chevalley.

\item For every \((i,n)\), \(n\ge0\), there is a tracked weak
universe \(p_{\kappa_i,n}:E_{\kappa_i,n}\to V_{\kappa_i,n}\)
with
\begin{equation}\label{eq:real-weak-universes}
 \Loc(p_{\kappa_i,n})=\mathcal S_{\kappa_i,n}.
\end{equation}
Its base is an original partitioned assembly of cardinality at most
\(\kappa_i\), and its finite presentation levels have the same total
cardinal bound.  There is also an original-base weak universe for all
monomorphisms.  These presentations satisfy (H1)--(H3) of
Definition~\ref{def:lc-universes}.  In particular, over every original assembly base, each composable pair of
maps in one universe level admits a code family whose structure map lies in the source category at that
level. Its evaluation represents every section with an original parameter
after a source cover of that parameter.

\item Their univalent completions satisfy
\begin{equation}\label{eq:real-universe-classification}
 \Map_{\E}(X,U_{\kappa_i,n})
 \simeq (\E_{/X})_{\mathcal S_{\kappa_i,n}}\core.
\end{equation}
Both \(U_{\kappa_i,n}\to1\) and its universal total-space map
\(\widetilde U_{\kappa_i,n}\to1\) belong to
\(\mathcal S_{\kappa_{i+1},n+1}\).  Class inclusions induce coherent
cumulative maps of classifiers and pullback squares of universal families.
The propositional classifier is the embedded effective-topos truth family
\(1\into\Omega\).

\item The codomain doctrine of \(\Asm\), equipped with these presented
weak universes, is an \(\infty\)-tripos, and
\[
 \TT\Cod(\Asm)=\Ex(\Asm)\simeq\E.
\]
The codomain doctrine of \(\E\) is the corresponding exact
\(\infty\)-tripos. The category of zero-types in \(\E\) is equivalent to
\(\Eff\).
The unit from assemblies is fully faithful, identifies the subobject lattices
of original objects with those of their images, and preserves original dependent products.
In particular, all original numerical maps retain their trackers.
\end{enumerate}
\end{theorem}
Here ``tracked'' means that the family's finite simplicial presentation carries uniform numerical trackers for its face, degeneracy, and chosen horn-filler operations.

Logical interpretations of the realizability completion are determined
by suitable interpretations of assemblies. Let \(\D\) be an exact
locally cartesian closed target with a finite-limit-closed covering
class of \(0\)-truncated projectives. For a regular interpretation
\(F:\Asm\to\D\) sending partitioned assemblies into this class,
the extension \(\overline F:\E\to\D\) preserves dependent products
exactly when \(F\) preserves weak section families. Such an extension
also preserves the chosen universes exactly when
\(\Loc(\overline Fp_i)=\M_i^{\D}\) at every level. Restriction gives an equivalence of the
corresponding interpretation spaces, so these extensions are unique
up to coherent equivalence. This is the assembly case of
Theorem~\ref{thm:pt-morphisms} in
Appendix~\ref{app:projective-morphisms}.

\section{Applications of the Realizability Completion}\label{sec:applications}
We give two applications of the realizability completion. The first distinguishes classifiers by their computable observations; the second reconstructs higher function types from tracked cochains and computes their quotient depth.

\subsection{Arithmetical Completeness and Classifier Separation}\label{sec:application}
We continue in the realizability completion $\E$ and write $\Gamma X=\Map_\E(1,X)$. Every object of $\E$ has an assembly cover.
For global points $x,y:1\to U$, write $\operatorname{Id}_U(x,y)=1\times_{U\times U}U$, formed using $(x,y)$ and $\Delta_U$.

Write $N$ for the standard numerical assembly with carrier $\NN$ and realizers $\rho_N(k)=\{k\}$, and $2\into N$ for the subobject on $\{0,1\}$. We regard $\Ftwo$ as the additive group on $2$, with realizers $\rho_{\Ftwo}(b)=\{b\}$.

A uniform computable presentation consists of numerical representatives for the coefficient groups, a fibrewise equivalence relation computably enumerable uniformly in $e$, and uniformly recursive operations on representatives; finite products and iterated bar quotients induce the higher presentations.

For a $0$-truncated abelian group object $A$ in $\E_{/B}$, write $K_B(A,n)$ for its $n$-fold bar delooping, with $K_B(A,0)=A$. We omit the base when it is terminal and use the same notation for ordinary Eilenberg--Mac Lane spaces in formulas for global sections.

\begin{theorem}[Rice-type classifier separation]\label{thm:rice}
For each $n\ge2$ there is a uniformly computably presented family $p_n:\codes_n\to N$ and a univalent completion $u_n:\widetilde U_n\to U_n$ with code cover $c_n:N\onto U_n$ satisfying $c_n^*u_n\simeq p_n$, such that:
\begin{enumerate}
\item\label{rice:fibres} The fibre $X_{n,e}$ is $1$ when $e\in\Tot$ and $K(\Ftwo,n)$ otherwise, where $\Tot=\{e:\varphi_e\text{ is total}\}$. In particular all fibres are simply connected.
\item\label{rice:complete} For a fixed total index $e_*$ and $u_*=c_n(e_*)$, the set
\begin{equation}\label{eq:pi2}
 \{e:\Gamma\operatorname{Id}_{U_n}(c_n(e),u_*)\ne\varnothing\}
\end{equation}
is $\Pi^0_2$-complete under computable many-one reductions.
\item\label{rice:global} $\Gamma U_n\simeq1\amalg K(\Ftwo,n+1)$; the internal object $U_n$ has truncation height exactly $n+1$.
\item\label{rice:maps} All maps $U_n\to N$ are constant. The constant maps induce equivalences $\Map_\E(U_n,N)\simeq\NN$ and $\Map_\E(U_n,2)\simeq2$.
\item\label{rice:local} If $q_n:\mathcal Y_n\to2$ has the fixed fibres $1$ and $K(\Ftwo,n)$, then for every cover $a:V\onto N$ and every map $r:V\to2$, the families $a^*p_n$ and $r^*q_n$ lie in distinct equivalence classes over $V$.
\end{enumerate}
\end{theorem}
The resulting fibres have finite homotopy groups, and the theorem identifies the additional uniform information retained by their realizability presentations.

A computable reset system constructs a uniformly tracked modest abelian
group family $\mathcal M\to N$, with fibres $M_e$ obtained from copies
of $\Ftwo$ by uniformly recursive transitions that are zero or the
identity. Appendix~\ref{app:reset} gives the construction.

For each fixed index, the tracked fibre isomorphisms are $M_e\cong0$ for $e\in\Tot$ and $M_e\cong\Ftwo$ for $e\notin\Tot$.
Iterated bar constructions, formed by effective groupoid quotients in $\E_{/N}$, define $p_n=K_N(\mathcal M,n)\to N$.

Lemma~\ref{lem:em-eq} represents the fibre equivalences of $p_n$ on
all higher test contexts. Their groupoid and its tautological action
yield the classifier $U_n$ and its universal family by effective quotients.

For a map $F:U_n\to N$, $\Gamma F$ is constant on each of the two
components of $\Gamma U_n$. If these values differed, the total
recursive carrier function of $Fc_n:N\to N$ would decide $\Tot$.
Thus $Fc_n$ is constant, and descent along the code cover $c_n$
makes $F$ constant.

\begin{corollary}\label{cor:two-classifiers}
Let $W_n$ be the base of the univalent completion of the fixed two-fibre family $q_n$. Then $\Gamma W_n\simeq\Gamma U_n$, and a nonconstant map $W_n\to2$ distinguishes $W_n$ from $U_n$, giving $W_n\not\simeq U_n$.
\end{corollary}

\subsection{From Tracked Cochains to Higher Function Types}\label{sec:cochains}
The tracked reset family $\mathcal M\to N$ from Section~\ref{sec:application} also gives higher function types of exact quotient depth at every positive finite height. Put $G_N=N\times C_2$, where $C_2$ is the standard finite decidable group of order two, and write $B_NG_N$ for its bar quotient over $N$.

Normalized cochains for $C_2$, acting trivially on $\mathcal M$, provide tracked codes for sections and their successive identities; coherent fullness reconstructs the resulting function types.

Quotient depth is the least nesting depth of effective Segal groupoid quotients in a construction from the embedded zero-types $\Eff_{/N}$, with finite limits and images preserving the maximum input depth (Definition~\ref{def:cochain-depth}).

\begin{corollary}[Exact Quotient Depth of Function Types]\label{cor:cochain-depth}
For every $n\ge1$, the uniformly computably presented function-type family
\[
 \mathcal Q_n=K_N(\mathcal M,n)^{B_NG_N}\longrightarrow N
\]
formed in $\E_{/N}$ has fibres
\begin{equation}\label{eq:cochain-reset}
 (\mathcal Q_n)_e\simeq
 \begin{cases}
 1,&e\in\Tot,\\
 \displaystyle\prod_{j=0}^{n}K(\Ftwo,j),&e\notin\Tot.
 \end{cases}
\end{equation}
Its contractible-fibre index set is $\Pi^0_2$-complete, and its quotient depth in $\E_{/N}$ from objects of the embedded slice $\Eff_{/N}$ is exactly $n$.
\end{corollary}

Thus quotient depth grows while the arithmetical complexity of recognizing contractible fibres remains fixed. Appendix~\ref{app:cochains} proves the general cochain reconstruction theorem (Theorem~\ref{thm:cochain-products}) and proves this corollary.

\pdfbookmark[0]{References}{references}
\phantomsection

\appendix
\renewcommand{\theHsection}{appendix.\Alph{section}}
\renewcommand{\theHsubsection}{\theHsection.\arabic{subsection}}
\pdfbookmark[0]{Complete Proofs}{proof-appendix}
\phantomsection
\section*{Overview of the Appendices}
The appendices give the full proofs and develop the constructions and auxiliary results in detail.
\begin{itemize}
\item \textit{Regular foundations and logical reconstruction.}
Appendix~\ref{app:foundations} proves comprehension, exact completion, and its fibrewise
description. It defines relative atlas presentations of finite height and proves their recognition
lemma and relative finite-height reconstruction
(Theorem~\ref{thm:relative-fullness}). It then proves individual univalent
completion, the five conditions for universe closure, dependent-product-preserving extensions, and the fundamental
adjunction, including the factorization through exact doctrines.
\item \textit{Conditions for logical completion and universe reconstruction.}
Appendix~\ref{app:completion} establishes conservativity and the bounded-source
criterion, then proves the theorem on dependent-product reconstruction within a universe level
(Theorem~\ref{thm:fullness}) and the specialized projective interpretation theorem.
It also develops logical completion from projective covers and a sufficient criterion for universe
presentations, together with refinements concerning classification, doctrines,
and interpretations preserving the chosen universes. Further results include classical recovery,
bounded-space and higher-source examples, and a finite-group example separating
logical closure from quotient closure.
\item \textit{Realizability (Section~\ref{sec:realization}).}
Appendix~\ref{app:realizability-hierarchy} starts from the assembly source and
its projective cover, then verifies locality, closure, and exhaustiveness of
the small projective atlas classes. Tracked Kan presentations provide weak
universes and small section codes, from which it derives same-level dependent
products. The proof of Theorem~\ref{thm:real-hierarchy} is completed in
Subsection~\ref{app:realization-classification}, where univalent classification,
size bounds, and cumulativity are assembled into the hierarchy and the logical
interpretation criterion is applied.
\item \textit{Rice-type classifier separation (Subsection~\ref{sec:application}).}
Appendix~\ref{app:rice} constructs the reset family and represents fibrewise
equivalences to obtain its univalent classifier. It proves that the numerical
index set for identity with a fixed classifier code is $\Pi^0_2$-complete and
computes the classifier's global homotopy type. The proof of
Theorem~\ref{thm:rice} concludes in
Subsection~\ref{app:numerical-separation} with the constancy of numerical maps
and the obstruction to local classification by two fixed fibres. The same
subsection proves Corollary~\ref{cor:two-classifiers}, separating two classifiers
with equivalent global homotopy types.
\item \textit{Higher function types and quotient depth
(Subsection~\ref{sec:cochains}).}
Appendix~\ref{app:cochains} derives section codes and their successive identity
codes from tracked coefficient operations. Reconstruction from a finite tower of section codes
yields the resulting function types and the explicit product
formula of Theorem~\ref{thm:cochain-products} in
Subsection~\ref{app:cochain-reconstruction}. The calculation at height two makes
the quotient stages explicit. Subsection~\ref{app:cochain-depth} then applies
this formula to the reset family, computes the contractible-fibre index set,
and proves the exact quotient depth of Corollary~\ref{cor:cochain-depth} by a
matching truncation lower bound.
\end{itemize}

\section{Completion, Reconstruction, and Logical Adjunction}\label{app:foundations}
This appendix gives proofs for Sections~\ref{sec:reconstruction} and~\ref{sec:triposes}. We use mapping spaces and homotopy-coherent diagrams and adjunctions throughout.

\subsection{Exactness, Embeddings, and Descent}\label{app:descent}
We use the regular exact-completion theorem~\cite{Stefanich23}. For small regular $\C$, the regular topology is generated by singleton covers, and $h:\C\into\HH_\C$ is fully faithful and regular. The completion $\E=\Ex(\C)$ is the full closure of $h\C$ under groupoid realizations; it is finite-limit closed and exact. Restriction gives
\begin{equation}\label{eq:import-universal}
 \Fun_{\mathrm{reg}}(\Ex\C,\D)
       \simeq\Fun_{\mathrm{reg}}(\C,\D)
       \qquad(\D\text{ exact}).
\end{equation}
Every object has an original atlas: the full subcategory of objects admitting an original cover contains $h\C$ and is closed under groupoid realizations by composing an atlas of $R_0$ with $R_0\onto|R_\bullet|$. A cover between original objects is reflected by $h$, since $h$ preserves its image factorization and reflects equivalences. If $\C$ is exact, $\Ex\C\simeq\C$ compatibly with the unit. For an ordinary regular source, the completion is bounded and
\begin{equation}\label{eq:import-shadow}
 (\Ex\C)_{\le0}\simeq\Exone\C.
\end{equation}
This identifies the zero-types of the higher completion with the ordinary exact/regular completion~\cite{Stefanich23}.

\begin{lemma}[Quotient and truncation preservation]\label{lem:quotient-truncation}
A regular functor between exact $\infty$-categories preserves effective groupoid quotients. A finite-limit-preserving functor preserves $n$-truncated objects and maps for every finite $n\ge-2$; a fully faithful such functor also reflects them. The full embedding $\E\into\HH_\C$ preserves finite limits, covers, and groupoid quotients and preserves and reflects finite truncation. The same assertions hold for the regular Yoneda embedding of any exact category.
\end{lemma}
\begin{proof}
An effective groupoid is the \v{C}ech nerve of its quotient cover. A regular functor preserves the cover and its finite-limit nerve, so the image quotient is the realization of the image groupoid. This applies in particular to the exact regular-sheaf embedding. Preservation of finite truncation follows by induction from the diagonal criterion, starting with equivalences at degree $-2$; a fully faithful functor reflects the successive diagonal equivalences. The assertions for $\E\subseteq\HH_\C$ also follow directly from its exact-completion construction.
\qed\end{proof}

\begin{lemma}[Exact descent]\label{lem:exact-descent}
Let $\D$ be exact. For every cover $V\onto X$ in $\D$,
\begin{equation}\label{eq:slice-descent}
 \D_{/X}\simeq\lim_{[k]\in\Delta}\D_{/V^{\times_X(k+1)}}.
\end{equation}
If $\D$ is fully embedded in an $\infty$-topos preserving finite limits and effective groupoid quotients, an ambient family over $X$ that becomes represented after pullback to $V$ is represented. Pullback along covers is conservative. These statements apply to $\E\subseteq\HH_\C$.
\end{lemma}
\begin{proof}
Use the regular-sheaf embedding of $\D$, or the indicated full exact embedding. Ambient effective descent holds~\cite{Lurie09}. If $S_V=S\times_XV$ is represented, the \v{C}ech nerve of $S_V\to S$ has terms
\[
 S_V\times_V V^{\times_X(k+1)},
\]
which are represented by finite limits. Fullness makes this an internal groupoid in $\D$; exactness and preservation of its quotient identify its quotient with $S$. Thus ambient descent is essentially surjective on the represented objects. It is fully faithful on their mapping spaces, since the ambient slice equivalence and all embeddings are fully faithful. This proves~\eqref{eq:slice-descent}. An equivalence after pullback to the cover is an equivalence on the entire \v{C}ech nerve and hence, by descent, an equivalence, proving conservativity.
\qed\end{proof}

\begin{lemma}[Cover criterion for weak section families]\label{lem:weak-cover}
For a pair in a regular $\infty$-category $\D$, a code family is weakly adequate if and only if its transpose to the ambient dependent product is a cover in $\Sh(\D,\Spc)$. For $\C\subseteq\E\subseteq\HH_\C$, the same conclusion holds for a code family in $\E$ whose adequacy is tested either on original parameters using original covers, or on all completed parameters using completed covers.
\end{lemma}
\begin{proof}
A morphism of sheaves is a cover precisely when it is locally surjective on components over every representable. The regular topology is generated by singleton regular covers and their pullbacks and composites. The homotopy in~\eqref{eq:weak} is precisely a local lift of a point of the section space through the transposed evaluation map. This proves the first assertion and the original-test version of the second. An original atlas refines any completed test object; pulling back a completed cover and refining its domain by an original atlas yields a cover between original objects, which regular Yoneda reflects. Conversely, suppose the transpose is an ambient cover and take a section with completed parameter $A$. Choose an original atlas $hT\onto A$. Over $hT$, the regular-site lifting criterion gives an original cover $hT'\onto hT$ and a lift to the code object, with its evaluation homotopic to the restricted section. The composite $hT'\onto A$ is a cover in $\E$ and witnesses completed-context adequacy. This proves equivalence of these tests.
\qed\end{proof}

\subsection{Comprehension and Doctrine Maps}\label{app:moens}
\begin{proof}[of Proposition~\ref{prop:moens}]
Write $\C=\PP(1)$. The dependent-sum comparisons give $\delta X=\Sigma_{!_X}1_X$ functorially: for $f:X\to Y$, the map $\delta f$ is
\[
 \Sigma_{!_X}1_X\simeq\Sigma_{!_Y}\Sigma_f1_X
 \longrightarrow\Sigma_{!_Y}1_Y,
\]
where the last arrow is induced by the counit at $1_Y$, using $f^*1_Y\simeq1_X$. Composition and identity comparisons are those of the adjunctions.

Fix $u=!_X$. Define
\[
 K_X(A)=(\Sigma_uA\longrightarrow\delta X),\qquad
 L_X(W\xrightarrow{q}\delta X)
 =1_X\times_{u^*\delta X}u^*W.
\]
The map $1_X\to u^*\delta X$ in this pullback is the unit $\eta_{1_X}$. The naturality square of the unit for $A\to1_X$ is cartesian, giving
\[
 A\simeq1_X\times_{u^*\delta X}u^*\Sigma_u A=L_XK_X(A).
\]
For the other composite, preservation of pullbacks by $\Sigma_u$ gives
\begin{equation}\label{eq:moens-counit}
 \Sigma_uL_X(q)\simeq
 \delta X\times_{\Sigma_u u^*\delta X}\Sigma_u u^*W.
\end{equation}
Cartesianness of the counit at $q$ identifies
\[
 \Sigma_u u^*W\simeq
 W\times_{\delta X}\Sigma_u u^*\delta X.
\]
The composite $\delta X\xrightarrow{\Sigma_u\eta_{1_X}}
\Sigma_u u^*\delta X\xrightarrow{\varepsilon_{\delta X}}\delta X$ is the identity by the triangle identity. Substitution in~\eqref{eq:moens-counit} identifies $\Sigma_uL_X(q)$ with $W$ over $\delta X$. These comparisons are natural and are maps of slice objects, so $K_X$ and $L_X$ give an equivalence of $\infty$-categories.

Coherent composition of sums identifies $K_Y\Sigma_f$ with postcomposition of $K_X$ by $\delta f$. Taking right adjoints identifies $f^*$ with pullback by $\delta f$. Now apply Beck--Chevalley to a pullback square in $\B$ and to the terminal object in the relevant fibre. In slices, the resulting comparison says precisely that the corresponding square of $\delta$-images is a pullback. Also $\delta1\simeq1_\C$. Hence $\delta$ preserves finite limits and the equivalences are natural in the base.
\qed\end{proof}

\begin{lemma}[Maps of regular doctrines]\label{lem:doctrine-maps}
Up to a natural equivalence of spaces of coherent data, a regular doctrine map $\PP\to\QQ$ is a triple $(v,F,\theta)$, where $v:\B_\PP\to\B_\QQ$ preserves finite limits, $F:\C_\PP\to\C_\QQ$ is regular, and
\begin{equation}\label{eq:base-comparison}
 \theta:F\delta_\PP\simeq\delta_\QQ v
\end{equation}
is a natural equivalence. Its fibre functors are the induced slice functors under Proposition~\ref{prop:moens}.
\end{lemma}
\begin{proof}
A regular doctrine map determines $F$ on the terminal fibre; the terminal comparison for $v$ is included. Preservation of terminal objects in each fibre and of $\Sigma_{!_X}$ then gives~\eqref{eq:base-comparison}. The same comparison on a general object $A\in\PP(X)$ identifies its image with the arrow obtained by applying $F$ to $K_XA$ and transporting its codomain by $\theta_X$. Therefore the fibre transformation is recovered, including its naturality and sum comparisons.

Conversely, the slice functors associated with $(F,\theta)$ preserve finite limits and covers because $F$ does. They commute with reindexing by finite-limit preservation, and with sums because sums in a codomain doctrine are composition. This constructs a regular doctrine map. The two constructions are inverse on spaces of diagrams: the functors $K_X$ are equivalences, and the comparisons are their natural unit and counit equivalences, retaining their full coherent data. This proves the assertion also for all higher transformations.
\qed\end{proof}

\subsection{Exact Completion and Slicing}\label{app:regular-completion}
\begin{proof}[of Theorem~\ref{thm:regular-completion}]
Put $\C=\PP(1)$ and $\E=\Ex\C$. Since $j\delta_\PP$ preserves finite limits, pulling back $\Cod(\E)$ along it gives~\eqref{eq:regular-completion}. Its fibres are regular slices; reindexing is regular pullback, and sums are given by composition, with cartesian units and counits and Beck--Chevalley comparisons. Its closed-type fibre is $\E/j\delta_\PP1\simeq\E$, so the doctrine is exact. Comprehension and $j$ give the regular unit, with identity base and fully faithful fibre functors.

By Lemma~\ref{lem:doctrine-maps}, a regular map $\PP\to\mathbb D$ consists of a regular $F:\C\to\C_{\mathbb D}$, a finite-limit-preserving base functor $v$, and an equivalence $F\delta_\PP\simeq\delta_{\mathbb D}v$. Because $\C_{\mathbb D}$ is exact,~\eqref{eq:import-universal} extends $F$ to a regular $\overline F:\E\to\C_{\mathbb D}$, with a contractible space of choices relative to its restriction. Transporting the base comparison gives a regular map $\mathsf Q\PP\to\mathbb D$. Restriction is its inverse. More precisely, the interpretation spaces are homotopy pullbacks of the spaces of base functors, closed-type functors, and natural-equivalence comparisons. Taking this homotopy pullback preserves the equivalence~\eqref{eq:import-universal}, proving~\eqref{eq:q-universal} on full mapping spaces.
\qed\end{proof}

We next identify slices over original objects with the exact completions of the corresponding source slices.

\begin{proof}[of Theorem~\ref{thm:slice-completion}]
The category-of-elements description of a representable gives an equivalence
\[
 \Sh(\C/T,\Spc)\simeq\HH_\C/hT.
\]
Indeed presheaves over the representable $hT$ are presheaves on $\C/T$, and the induced topology is the regular topology: a slice arrow is a cover precisely when its underlying arrow in $\C$ is a cover. Under this equivalence, the representable associated to $A\to T$ is $hA\to hT$.

Let $\mathcal K$ be the full subcategory of $\HH_\C/hT$ generated by these original slice objects under groupoid realizations. By exact completion on the slice site, it is the image of $\Ex(\C/T)$. Since $\E/hT$ is closed under groupoid realizations and contains all original slice objects, $\mathcal K\subseteq\E/hT$.

For the reverse inclusion, let $\mathcal F\subseteq\E$ be the full subcategory of those $X$ for which every map $X\to hT$ defines an object of $\mathcal K$. It contains the originals, by full faithfulness of $h$. Suppose $R_\bullet$ is an internal groupoid whose terms belong to $\mathcal F$, and put $Q=|R_\bullet|$. For any map $Q\to hT$, its composites with the augmentation make $R_\bullet$ a groupoid in the slice. Every term belongs to $\mathcal K$ by the definition of $\mathcal F$, and its quotient in the slice is exactly $Q\to hT$. Closure of $\mathcal K$ gives $Q\in\mathcal F$. Thus $\mathcal F$ contains the full groupoid-realization closure of the originals, namely all of $\E$. This proves $\E/hT\subseteq\mathcal K$ and hence~\eqref{eq:slice-completion} as an equivalence of full $\infty$-categories.

Pullback in the ambient topos preserves groupoid realizations and takes an original slice object to the corresponding original pullback. Therefore the equivalence commutes with pullback, with its coherent comparisons. Finally, apply it to $T=\delta_\PP X$ and use comprehension~\eqref{eq:moens}; this gives~\eqref{eq:fibrewise-completion}, naturally in $X$.
\qed\end{proof}

\subsection{Relative Atlases and Bounded Sources}\label{app:lc-atlases}
The following relative recursion makes the finitely truncated part explicit, including when the original base has higher homotopy.
\begin{definition}[Relative atlas presentation of finite height]\label{def:lc-atlas}
For $T\in\C$ and finite $n\ge-1$, define $\mathcal A_n(T)\subseteq\HH_\C/hT$ recursively.
\begin{enumerate}
\item $X\to hT$ lies in $\mathcal A_{-1}(T)$ if it is monic and has an original atlas $hA\onto X$ over $hT$.
\item For $n\ge0$, it lies in $\mathcal A_n(T)$ if it is $n$-truncated and has an original atlas $hA\onto X$ whose relative kernel
\begin{equation}\label{eq:lc-atlas-kernel}
 hA\times_XhA\longrightarrow h(A\times_TA)
\end{equation}
lies in $\mathcal A_{n-1}(A\times_TA)$.
\end{enumerate}
A presentation includes these atlases and their recursive kernel presentations. A finite-atlas object is an object in some $\mathcal A_n(1)$. The mapping spaces between finite-atlas objects are those of $\HH_\C$.
\end{definition}
Finiteness measures the height of the recursion. The bound in~\eqref{eq:lc-atlas-kernel} is relative to its original base.

\begin{lemma}[Relative recognition and bounded sources]\label{lem:lc-atlas}
For any regular $\C$, $T\in\C$, and finite $n\ge-1$,
\[
 X\to hT\in\mathcal A_n(T)
 \quad\Longleftrightarrow\quad
 X\in\Ex(\C)\text{ and }X\to hT\text{ is }n\text{-truncated}.
\]
If $\C$ is bounded, then $\Ex(\C)$ is bounded and its objects are exactly the finite-atlas objects. In particular this applies to every ordinary regular source.
\end{lemma}
\begin{proof}
First prove the relative assertion for arbitrary regular $\C$ by induction on $n$. For $n=-1$, an atlas followed by the given monomorphism,
\[
 hA\onto X\into hT,
\]
identifies $X$ with the image of the original map $A\to T$, because regular Yoneda preserves image factorizations. Hence $X$ belongs to $\E$. Conversely a monomorphism over $hT$ with domain in $\E$ has an original atlas and is in $\mathcal A_{-1}(T)$.

For $n\ge0$, a presentation has relative kernel $K=hA\times_XhA$ in $\E$ by the induction hypothesis over $h(A\times_TA)$. By the Segal identities, every term of the \v{C}ech nerve of $hA\onto X$ is a finite fibre product of $hA$ and $K$, and therefore lies in $\E$. Fullness retains its structure maps and coherences, and exactness represents its ambient quotient $X$. Conversely, if $X\in\E$ and $X\to hT$ is $n$-truncated, choose an original atlas $hA\onto X$. Its kernel belongs to $\E$ by finite-limit closure. The map~\eqref{eq:lc-atlas-kernel} is a pullback of $\Delta_{X/hT}$ and is $(n-1)$-truncated, so the induction hypothesis gives its required presentation. This completes the relative induction.

Now assume $\C$ bounded. Its original objects are finitely truncated in $\HH_\C$ because finite-limit functors preserve finite truncation. Bounded objects of an $\infty$-topos are closed under groupoid realizations. To see this, let $R_\bullet$ have bounded terms. The map $R_1\to R_0\times R_0$ between bounded objects has some finite relative truncation bound $N$. Effectiveness identifies it with the pullback of $\Delta_{|R_\bullet|}$ along the cover $R_0^2\onto|R_\bullet|^2$. Locality of truncation makes that diagonal $N$-truncated, so the quotient is $(N+1)$-truncated. Thus the groupoid-realization closure of the originals is bounded. Applying the relative recognition to terminal maps now gives $\bigcup_n\mathcal A_n(1)=\E$. In an ordinary category every object is $0$-truncated, so the final assertion follows as a special case.
\qed\end{proof}

\subsection{Relative Finite-Height Reconstruction}\label{app:relative-fullness}
\begin{lemma}[Adequacy in all completed contexts]\label{lem:all-contexts}
Assume $\mathrm{CF}_n(\C)$. Every diagram $Z\xrightarrow gY\to X$ in $\E=\Ex\C$ with $g$ $n$-truncated has a weak section family in $\E$ whose total code object is original.
\end{lemma}
\begin{proof}
Choose an original atlas $jT\onto X$, pull the diagram back to $jT$, and take its source code $W\to T$. Its evaluation is also an evaluation $jW\times_XY\to Z$. For a section with arbitrary parameter $A\to X$, cover $A\times_XjT$ by an original $jS$. Full faithfulness makes the induced $jS\to jT$ original. Coherent fullness represents the restricted section after an original cover $S'\onto S$. The composite
\[
 jS'\onto jS\onto A\times_XjT\onto A
\]
is a cover in $\E$ and provides the local code and homotopy.
\qed\end{proof}

We use these section codes to reconstruct dependent products by induction on relative truncation.

\begin{proof}[of Theorem~\ref{thm:relative-fullness}]
Assume $\mathrm{CF}_n(\C)$. We induct, simultaneously over all completed bases,
on the relative truncation degree $r\le n$ of $g$.
For a pair $Z\xrightarrow gY\xrightarrow fX$, put
$S=\Pi_f^{\HH_\C}g$. This is $r$-truncated over $X$, since the ambient right
adjoint on slices preserves finite limits and hence finite truncation. By
Lemma~\ref{lem:all-contexts}, there is a code family $W\to X$ with locally
adequate evaluation. Lemma~\ref{lem:weak-cover} makes its transpose
$a:W\onto S$ a cover.

For $r=-1$, the factorization $W\onto S\into X$ is the ambient image of the
represented map $W\to X$. The regular embedding identifies this image with
one in $\E$. Thus $S$ is represented.

For $r\ge0$, let $B=W\times_XW$, $Y_B=Y\times_XB$, and $b:Y_B\to B$.
The two code projections give evaluations $e_0,e_1$ over $Y_B$. Let
$i:I\to Y_B$ be the pullback of $\Delta_g$ along their pair. Its relative
truncation degree is $r-1$. For every ambient test $T\to B$, a point of
$\Map_{/B}(T,W\times_SW)$ is a homotopy between the two induced maps to $S$.
Under the section adjunction this is a homotopy between their sections over
$T\times_BY_B$, which is exactly a section of $I$. This identification is an
equivalence of spaces, natural in $T$, so it preserves paths between homotopies
and all subsequent identities. Yoneda gives
\begin{equation}\label{eq:kernel-identity}
 W\times_SW\simeq\Pi_b^{\HH_\C}(I\to Y_B)
       \qquad\text{over }B.
\end{equation}

The induction hypothesis represents $W\times_SW$ in $\E$. Every term of the
\v{C}ech nerve of $a$ is then a finite fibre product of represented objects.
Fullness lifts all its structure maps and coherences. Exactness gives its
quotient in $\E$, and the embedding preserves that quotient, identifying it
with $S$. Fullness restricts the ambient section adjunction
 to~\eqref{eq:relative-product}. The representing object with its evaluation
has a contractible space of choices, so the construction is coherent and
inherits ambient Beck--Chevalley.

The base case uses an image and each inductive step adds
one quotient. Lemma~\ref{lem:all-contexts} gives a code in each context
before this induction. Thus at height $r$ the construction adds at most
$r+1$ nested quotients to the given input presentations and code objects.

Conversely, assume the products in the statement exist. For a pair over an
original $jT$, let $Q=\Pi_fg$ in $\E$. Testing its adjunction against original
$hS\to hT$ identifies $Q$ with the ambient product: these representables
generate the ambient slice under colimits, and the two represented mapping
functors send such colimits to limits. Choose an original atlas $jW\onto Q$
and restrict evaluation. Given a section with original parameter $jS$, pull
back the atlas along its transpose and cover the resulting object by $jS'$.
The composite $jS'\onto jS$ is a cover between original objects, hence is the
image of a cover $S'\onto S$ in $\C$. Full faithfulness makes the map to
$jW$ selecting a code original. This proves $\mathrm{CF}_n$.
\qed\end{proof}

\subsection{Represented Products and Univalent Completion}\label{app:universes}
\begin{lemma}[Preservation of represented dependent products]\label{lem:old-pi}
If an individual dependent product exists in a regular $\infty$-category, its regular Yoneda image is the corresponding ambient dependent product. In particular, the regular Yoneda embedding of a locally cartesian closed category preserves all dependent products in every slice.
\end{lemma}
\begin{proof}
Let $Q=\Pi_fg$ for $Z\xrightarrow gY\xrightarrow fX$ in $\C$. The natural equivalences
\[
 \Map_{\C/X}(T,Q)\simeq\Map_{\C/Y}(T\times_XY,Z)
\]
identify the image of $Q$ and the ambient section object on every original represented test $hT\to hX$. These tests generate the ambient slice under colimits, and the represented mapping functors carry these colimits to limits. Thus the canonical comparison induced by evaluation is an equivalence on every ambient test. This identifies the individual product together with its evaluation and substitution comparisons.
\qed\end{proof}

We now construct the univalent completion by realizing the groupoid of fibre equivalences and its action on the given family.

\begin{proof}[of Theorem~\ref{thm:univalent-completion}]
All constructions take place in $\E\subseteq\HH_\C$, and all dependent products
used agree with the ambient ones by Theorem~\ref{thm:relative-fullness}.
For $n$-truncated families $A,B$ over $S$, the fibrewise mapping object is
$\Pi_{A/S}(A\times_SB\to A)$. For a map $\varphi:A\to B$, write
$F_b=\sum_{a:A}(\varphi(a)=b)$. Define
\[
 \operatorname{isContr}(F)
     =\sum_{c:F}\prod_{x:F}(c=x),\qquad
 \operatorname{isEquiv}(\varphi)
     =\prod_{b:B}\operatorname{isContr}(F_b).
\]
The families $F_b$ are $n$-truncated. Their identities are $(n-1)$-truncated;
for $n=-1$ the identity maps are equivalences. Hence the displayed products
all have second maps of truncation degree at most $n$.
The ambient type theory identifies contractibility and equivalence as
propositions. The subobject of the mapping object cut out by
$\operatorname{isEquiv}$ therefore represents the full space of equivalences.
It is $n$-truncated over $S$.

Apply this construction to the two pullbacks of $p$ over $V\times V$. The
objects
\begin{equation}\label{eq:equiv-groupoid}
 R_0=V,\qquad R_1=\Eq_{V\times V}(\pi_1^*E,\pi_2^*E)
\end{equation}
form an internal groupoid by composition of fibre equivalences. All inverses,
identities, and higher comparisons are inherited coherently from the ambient
slices. Its tautological action on $E$ is also internal. Exactness gives
\[
 U=|R_\bullet|,\qquad
 \widetilde U=|E//R_\bullet|,\qquad u:\widetilde U\to U.
\]
Here $E//R_\bullet$ denotes the action groupoid of the tautological action.
Choose an ambient univalent classifier in an external universe large enough
to classify $p$~\cite{GepnerKock17,Lurie09}. If $c:V\to U^{\mathrm{amb}}$ is
its classifying map, univalence identifies $R_\bullet$ with its \v{C}ech
groupoid. Thus $V\onto U\into U^{\mathrm{amb}}$ is the image factorization of
$c$. Restricting the ambient universal family gives $u$; the action quotient
identifies its total object. The square from $p$ to $u$ is cartesian and
$V\to U$ is a cover.

Every pullback of $u$ becomes a pullback of $p$ on this cover. Conversely, an
ambient classifying map of a family locally pulled back from $p$ locally factors
through $U\into U^{\mathrm{amb}}$. Factorization through a subobject is local
and uniquely coherent, so it factors globally. By fullness, the resulting classifying map lies
in $\E$. This gives~\eqref{eq:universe-classification}, including its full
mapping spaces.

Truncation of $u$ follows locally from $p$. The pullback of $\Delta_U$ along
$V\times V\onto U\times U$ is $R_1\to V\times V$, which is $n$-truncated.
Truncation descends, so $U$ is $(n+1)$-truncated.

If the local class is all monomorphisms, classification identifies
$\Map(T,U)$ with the discrete space $\Sub(T)$. A map $T\to\widetilde U$ is a
subobject of $T$ with a section. Such a subobject is an equivalence, and the
space of these data is contractible. Therefore $\widetilde U$ is terminal and
$u$ is $1\into\Omega$.
\qed\end{proof}

\begin{remark}[Classification of an individual family]\label{rem:fragment}
For any family whose fibre-equivalence groupoid and tautological action are
represented in an exact full subcategory of a sheaf $\infty$-topos, the same
image-factorization proof constructs its univalent completion there. This
observation applies to the two classifier constructions in Appendix~\ref{app:rice}.
\end{remark}

\subsection{Local Classification and Universe Closure}\label{app:generic-tests}
\begin{proof}[of Proposition~\ref{prop:source-generics-obstruction}]
Every map of the ordinary source is $0$-truncated. Finite-limit functors
preserve finite truncation, and $0$-truncation is stable under pullback
and descent along covers. Thus every map in every $\Loc(jp_i)$ is
$0$-truncated. The nerve of $jG$ has an effective quotient $BG$ in the
exact completion, with
\[
 1\times_{BG}1\simeq jG.
\]
If $BG$ were $0$-truncated, its loop object would be $(-1)$-truncated.
The identity of $G$ gives this loop object a section, so it would be
terminal. Full faithfulness of $j$ would make $G$ trivial, contrary to
the assumption. Hence $BG\to1$ belongs to none of the asserted local
classes.
\qed\end{proof}

\begin{lemma}[Local classification over source objects]\label{lem:local-recoding}
For any regular source and a presented family $p:P\to jV$ with original base,
$q\preceq_\C p$ if and only if $q\in\Loc_\E(p)$. These classes are
pullback-stable and local for covers of arbitrary completed bases.
\end{lemma}
\begin{proof}
For the forward implication, choose an original atlas
$jT\onto\operatorname{cod}(q)$ and apply the source-tested condition.
The composite $jT'\onto jT\onto\operatorname{cod}(q)$ is a cover on
which $q$ is a pullback of $p$. Conversely, suppose $D\onto B$ and a
map $D\to jV$ classify $q|_D$. For a source test $jT\to B$, choose an
original atlas $jT'\onto D\times_BjT$. The composite $jT'\onto jT$
is a cover. Regular Yoneda reflects covers between original objects,
and full faithfulness makes $jT'\to D\to jV$ an original map.
This gives the required source recoding. Pullback stability follows
by pulling back the local presentation. Locality follows by composing
covers, and the converse implication by base change of a cover.
\qed\end{proof}

\begin{lemma}[The generic composable pair]\label{lem:generic-pair}
For every $X\in\E$, pullback of~\eqref{eq:generic-pair} identifies
$\Map_\E(X,D_i)$ with the full space of composable pairs
$Z\xrightarrow gY\xrightarrow fX$ in $\M_i$, including equivalences
of the intermediate and total objects and all higher comparisons.
\end{lemma}
\begin{proof}
A map $X\to D_i$ consists of a map $a:X\to U_i$ and a section of
the pullback of $d_i$. The dependent-product adjunction identifies
the latter with a map $a^*\widetilde U_i\to U_i$. By univalent
classification, the first datum is a family $f:Y\to X$ in $\M_i$,
and the second is a family $g:Z\to Y$ in $\M_i$. The classifying
equivalences are equivalences of spaces, so they include every
equivalence and higher homotopy of these data. The universal evaluations
identify the pair with the pullback of $(f_i,g_i)$.
\qed\end{proof}

We use local classification and the generic composable pair to prove the universe-closure criterion.

\begin{proof}[of Theorem~\ref{thm:generic-tests}]
Suppose the five conditions for universe closure hold. By
Lemma~\ref{lem:local-recoding}, these are membership conditions in
$\M_i$. Every identity is a pullback of $\id_1$. If $q\in\M_i$,
choose a cover on which it is a pullback of $p_i$. Images are stable
under base change, so the image inclusion of $q$ is locally a pullback
of the image inclusion of $p_i$. The same is true of relative
diagonals: on the induced cover of the fibre square of $q$, its
diagonal is a pullback of $\Delta p_i$. Locality gives both closures.

For composable $f,g\in\M_i$, Lemma~\ref{lem:generic-pair} gives a
map $c:X\to D_i$ with the given pair as its pullback. Consequently
\[
 fg\simeq c^*\sigma_i,\qquad \Pi_fg\simeq c^*\rho_i.
\]
The product exists by relative finite-height reconstruction, and the
displayed product comparison is the Beck--Chevalley equivalence. Pullback stability now gives composition
and dependent-product closure. These identifications are natural under
further base change and give the canonical substitution comparisons.

Conversely, $p_i$ belongs to its own local class. Image and diagonal
closure therefore give the second and third conditions; identity closure
gives the first. Both members of the generic pair belong to $\M_i$,
so composition and product closure give the remaining conditions.
Lemma~\ref{lem:local-recoding} turns these memberships into the
displayed source-local recodings.
\qed\end{proof}

The reconstructed products, univalent completions, and closure criterion now combine to give the $\infty$-tripos-to-topos construction.

\begin{proof}[of Theorem~\ref{thm:tripos-construction}]
Put $\C=\PP(1)$ and $\E=\Ex(\C)$. Exactness and finite limits
are provided by exact completion. For any arrow $q:A\to B$ in
$\E$, choose an original atlas $jT\onto B$. Exhaustiveness over source
bases places $q|_{jT}$ in some $\M_i$. By locality, $q$ belongs to
that same class. Since $p_i$ is $n_i$-truncated, so is every member
of its local class. Thus all arrows have finite truncation bounds.
Applying this to terminal maps proves that $\E$ is bounded; full
faithfulness and finite-limit preservation give boundedness of $\C$.

For each pair $Z\xrightarrow gY\to X$, choose a finite bound for
$g$. Theorem~\ref{thm:relative-fullness} constructs its dependent
product in $\E$, with its ambient value and Beck--Chevalley. Hence
$\E$ is locally cartesian closed. Theorem~\ref{thm:univalent-completion}
constructs $u_i$ and its classification equivalence for each level,
and gives the truncation bounds on $u_i$ and $U_i$.
Theorem~\ref{thm:generic-tests} gives all the stated logical closure
properties. The same atlas argument identifies the propositional class
with all monomorphisms: one implication follows from the bound $-1$,
and the other from the assumed exhaustiveness of this class on source
bases. Its classifier is therefore $1\into\Omega$.

If $i\le k$, the family $u_i$ belongs to $\M_k$. The classifying
equivalence for $u_k$ gives a map $U_i\to U_k$ and a cartesian
comparison of universal families, in a contractible space of choices.
Uniqueness identifies identity maps and successive comparisons with
the direct comparisons, including all higher coherence data. This is
the asserted cumulativity. Exhaustiveness, composition closure, and
same-level product closure give enough closed univalent families,
which establishes the elementary $\infty$-topos axioms.
\qed\end{proof}

\Needspace{6\baselineskip}
We next pass to the doctrine fibres and identify the logical operations preserved by the unit.

\begin{proof}[of Corollary~\ref{cor:lcc}]
Completed reindexing is pullback between slices of $\E$ and has both
dependent sum and dependent product as adjoints. Comprehension
identifies the unit on each fibre with the fully faithful slice
functor induced by $j$. The conservativity assertions of
Theorem~\ref{thm:lc-products}, which hold for every regular source,
show that $j$ identifies source subobject lattices and preserves
every existing source dependent product. These statements pass to
slices: a subobject of a slice object is a subobject of its domain
with the induced structural map, and iterated-slice equivalences
identify the corresponding dependent-product adjunctions.

Finite limits and images give preservation and reflection of predicate
equality, conjunction, and existential quantification. For a source
monomorphism $m:A\into B$ and source map $f:B\to T$, a
$\mathrm{CF}_{-1}$ code $W\to T$ has image equal to the completed
product $\Pi_{jf}(jm)$. This is an original subobject of $jT$.
Full faithfulness restricts its adjunction to source tests, making
it the source universal quantifier, so that operation exists and is
preserved. Implication between $A,B\into T$ is
$\Pi_{A\into T}(A\times_TB\into A)$ and is treated in the same way.
The subobject order equivalences reflect all these operations.
\qed\end{proof}

\subsection{Preservation of Dependent Products}\label{app:preservation}
\begin{proof}[of Theorem~\ref{thm:preservation}]
Assume $F$ preserves the indicated weak section families. Choose a level
containing composable $f,g$, and put $P=\Pi_fg$,
$P'=\Pi_{Ff}(Fg)$. Evaluation defines the comparison
\[
 c:FP\longrightarrow P'.
\]
The product $P$ and its evaluation are themselves a same-level weak
section family. Preservation of that family makes $c$ a cover.
We prove $c$ an equivalence by induction on a finite truncation bound
$n$ for $g$, simultaneously for all pairs and bases.

For $n=-1$, both products are monic over the base, and $FP$ is monic
over $FX$ by finite-limit preservation. Hence $c$ is a monic cover
and an equivalence. For $n\ge0$, set $B=P\times_XP$ and
$Y_B=Y\times_XB$. The two universal evaluations have an identity
family $I\to Y_B$, obtained by pulling back $\Delta_g$. It is
$(n-1)$-truncated. Pullback and relative-diagonal closure put both
$I\to Y_B$ and $Y_B\to B$ in the same level. Their section object
is the diagonal $P\to B$, since a section is precisely a homotopy
between the two universal evaluations, with all its higher homotopies.

Induction identifies its product comparison with an equivalence
\[
 FP\longrightarrow FP\times_{P'}FP.
\]
This is $\Delta_c$: both the section-object identification and the
product comparison are induced by the given evaluations. Thus $c$
is monic, and its already proved cover property makes it an equivalence.
Directed exhaustiveness puts every composable pair in a common level,
so $F$ preserves all dependent products.

Conversely a weak section family $W$ determines a cover
$W\onto\Pi_fg$ by the cover criterion for weak section families. Regularity
preserves this cover, and product preservation identifies its target
with $\Pi_{Ff}(Fg)$. The image evaluation is therefore weakly adequate,
which proves the converse.
\qed\end{proof}

\subsection{Chosen Codes and Dependent-Product Preservation}\label{app:lc-morphisms}
\begin{proof}[of Theorem~\ref{thm:lc-morphisms}]
Write $\E=\Ex(\C)$ and $G=\bar F$. We first prove the statement
at a fixed height $n$, for any regular source satisfying
$\mathrm{CF}_n$. All dependent products in $\E$ used below exist by
Theorem~\ref{thm:relative-fullness}; the target is locally cartesian
closed by assumption. Applying $G$ to evaluation gives, for every
pair whose second map has the given height, a comparison
\begin{equation}\label{eq:lc-app-mate}
 c:G(\Pi_fg)\longrightarrow\Pi_{Gf}(Gg).
\end{equation}

Suppose the selected $\mathrm{CF}_n$ codes remain weak section
families in $\D$. Over a source base $X=jT$, let $jW$ be a selected
code for the pair. Adequacy makes $jW\onto\Pi_fg$ a cover. Its image
under $G$ is a cover, and the composite
\[
 FW\onto G(\Pi_fg)\xrightarrow{c}\Pi_{Gf}(Gg)
\]
is a cover because its evaluation is the interpreted weak section family.
The image of $c$ contains the image of this composite, so $c$ is a
cover. For an arbitrary base $X$, choose an original atlas $jT\onto X$.
Pulling back~\eqref{eq:lc-app-mate} along $FT\onto GX$ gives the
comparison for the pulled-back pair, by finite-limit preservation
and Beck--Chevalley in both categories. It is a cover by the
source-base argument. Covers descend along covers, so $c$ is a
cover for every such pair.

We prove simultaneously for all bases by induction on $k\le n$
that $c$ is also monic whenever $g$ is $k$-truncated. For $k=-1$,
the two objects are subobjects of $GX$, hence $c$ is monic. Suppose
$k\ge0$, put $P=\Pi_fg$ and $B=P\times_XP$, and form the identity
family $I\to Y_B$ between the two universal evaluations over
$Y_B=Y\times_XB$. This family is $(k-1)$-truncated, and its product
over $B$ is $\Delta_{P/X}:P\to B$. This follows by testing on every
object over $B$: both sides give the full space of homotopies between
the two pulled-back sections. Induction says that $G$ preserves this
lower product. Its target product identifies, by the same section
calculation, with
\[
 GP\times_{\Pi_{Gf}(Gg)}GP\longrightarrow GB.
\]
The induced comparison from $GP$ is $\Delta_c$, since the comparison
maps are mates of evaluation. It is an equivalence. Hence $c$ is
monic and, being a cover, an equivalence. This proves preservation
of all products at height at most $n$.

Conversely, preservation of these products takes every selected
code cover to a cover of the target product. The image evaluation
is therefore a weak section family for every parameter in $\D$.
The argument applies to every choice of adequate code witnesses.

For the full assertion, boundedness of $\C$ implies boundedness of
$\E$ by the exact-completion recognition result. Each map in $\E$
then has a finite truncation bound, so the preceding argument applies
to all pairs under $\mathrm{CF}_{<\infty}$. The finite induction follows
the truncation bound of the source map, which a finite-limit functor
preserves.

The regular universal property gives an equivalence of categories of
coherent regular functors. Code preservation and product preservation
are invariant under equivalence and therefore select unions of
components in their maximal subgroupoids. The established equivalence
of preservation conditions restricts exact completion to the stated
mapping-space equivalences, both at a fixed height and for all products.
The homotopy fibre over a permitted $F$ is contractible, which is the
asserted uniqueness with a specified restriction comparison.
\qed\end{proof}

\subsection{The Logical Adjunction and Exact Reflection}\label{app:adjunction}
We first record the comparison of the chosen universal families used in the adjunction.

\begin{lemma}[Uniqueness of classifier comparisons]
\label{lem:classifier-comparisons}
For two univalent families classifying the same full replete class in
a category with finite limits, the space of cartesian equivalences
between them is contractible.
\end{lemma}
\begin{proof}
The first family is classified by the second. Its classifying map
together with the pullback equivalence forms the homotopy fibre over
that family of a fully faithful classifying map of spaces. This fibre
is nonempty and contractible. Reversing the roles gives an inverse:
each composite classifies the original universal family, so the same
uniqueness identifies it with the identity, with all coherent
comparisons. Thus every such classifying map is an equivalence and
the space of the asserted cartesian equivalences is contractible.
\qed\end{proof}

Let $G:\E\to\D$ be regular between exact locally cartesian closed
categories and preserve dependent products. Then finite limits and
dependent products identify the image of the fibre-equivalence
groupoid of $p$ with that of $Gp$, including the tautological action.
Regular functors between exact categories preserve effective groupoid
quotients. Consequently
\begin{equation}\label{eq:univalent-preservation}
 G\Uni(p)\simeq\Uni(Gp).
\end{equation}
The right side classifies $\Loc_\D(Gp)$. It is equivalent to a
specified target universal family exactly when that local class is
the chosen target class. Sufficiency and the contractible space of
compatible equivalences follow from
Lemma~\ref{lem:classifier-comparisons}; necessity follows from their
classifying properties.

For maps between $\infty$-triposes, source tests suffice for adequacy
in the completed target. Given a completed parameter, choose an
original atlas, apply the source-tested adequacy condition, and compose
covers. Conversely, weak adequacy in the completion on a source
parameter can be refined by an original atlas of the covering object;
reflection of source covers gives an original cover. Equality
of two local classes is equivalent to mutual membership of their
generators. Since both generator bases are original in the target,
Lemma~\ref{lem:local-recoding} expresses this membership by original
covers and original maps between the covered bases. These arguments
establish the source-local descriptions of logical maps in the text.

\begin{proof}[of Theorem~\ref{thm:fundamental}]
Theorem~\ref{thm:tripos-construction} gives the object
$\TT\PP=\Ex(\C_\PP)$ of $\Top$. A regular doctrine map has a
regular extension $\widehat F$, with a contractible space of choices
including its restriction comparison. Theorem~\ref{thm:lc-morphisms}
identifies preservation of its selected codes with preservation of
all products: the source closed-type category is bounded by the
construction theorem. Equation~\eqref{eq:univalent-preservation}
then identifies~\eqref{eq:class-equality} with preservation of the
chosen universal families. Logicality is therefore exactly the
condition that $\widehat F$ be a map of $\Top$. It is stable under
identities and composition. Coherent uniqueness in exact completion
gives the functoriality of $\TT$ on these logical maps.

For $\D\in\Top$, the codomain doctrine has regular slices and
regular reindexing. Composition defines dependent sums, with the
required cartesian units, counits, and Beck--Chevalley comparisons.
Its closed-type category is $\D_{/1}\simeq\D$ and is exact, so its
higher exact completion is itself. The actual dependent products give
coherent-fullness codes. Use the chosen universal families as weak
presentations. Their classes are exhaustive, have the specified
propositional level, and satisfy the five conditions for universe closure by
Theorem~\ref{thm:generic-tests}. Thus $\Cod(\D)$ is an $\infty$-tripos.
A map of $\Top$ induces a logical map of codomain doctrines by the
same preservation characterization, defining the functor $\Cod$.

Taking cores in the regular universal property and restricting by
the proved preservation conditions gives
\[
 \Fun_{\mathrm{reg},\Pi,U}(\Ex(\C_\PP),\D)\core
 \simeq
 \Fun_{\mathrm{reg},\mathrm{codes},U}(\C_\PP,\D)\core.
\]
Here $U$ on the right means equality of the interpreted local classes
with the chosen target classes; on the left it means preservation of
the chosen universal families. Both are invariant conditions, so
the equivalence retains all higher comparisons of functors.

By Lemma~\ref{lem:doctrine-maps}, a regular interpretation
$\PP\to\Cod(\D)$ is a triple $(v,F,\theta)$ with $F$ regular,
$v:\B_\PP\to\D$ finite-limit preserving, and
$\theta:F\delta_\PP\simeq v$. For each $F$, the space of $(v,\theta)$
is contractible. Hence the right side of the preceding display is
exactly $\Map_{\Trip}(\PP,\Cod(\D))$, and the left side is
$\Map_{\Top}(\TT\PP,\D)$. This proves the natural mapping-space
equivalence and therefore the adjunction.

The identity on $\Ex(\C_\PP)$ gives the unit, whose closed-type
functor is $j$ and whose base functor is $j\delta_\PP$. Comprehension
identifies its fibres with the indicated slice embeddings.
Corollary~\ref{cor:lcc} establishes their conservativity assertions.
For $\D\in\Top$, exactness gives
$\TT\Cod(\D)=\Ex(\D_{/1})\simeq\D$. This equivalence preserves
products and the chosen universal families by their representing
properties, so it is the counit equivalence in $\Top$. Thus the right
adjoint is fully faithful.
\qed\end{proof}

\begin{corollary}[Factorization through exact doctrines]
\label{cor:exact-reflection}
Let $\Trip^{\mathrm{ex}}$ be the full subcategory whose closed-type
categories are exact. The exact completion induces a
reflection $\mathsf Q:\Trip\to\Trip^{\mathrm{ex}}$, keeping each
doctrine's base unchanged, and
\[
 \TT\simeq\mathrm{ev}_1\mathsf Q,
 \qquad \mathrm{ev}_1:\Trip^{\mathrm{ex}}\rightleftarrows\Top:\Cod.
\]
Here $\mathrm{ev}_1(\PP)=\PP(1)$, with its reconstructed univalent hierarchy and induced functors on terminal fibres.
The second displayed pair is an adjunction. For exact $\PP$, the
base-preserving unit $\PP\to\mathsf Q\PP$ is an equivalence.
\end{corollary}

\begin{proof}
Theorem~\ref{thm:regular-completion} defines $\mathsf Q\PP$ over
the same base, with closed-type category $\E=\Ex(\C_\PP)$. For an
$\infty$-tripos $\PP$, this category is exact and locally cartesian
closed by Theorem~\ref{thm:tripos-construction}. Its actual products
give coherent-fullness witnesses. Regard each realized $p_i$ as a
family over the source object $jV_i\in\E$. The classes are unchanged,
so they still form a weak hierarchy and satisfy the five conditions for universe closure.
Thus $\mathsf Q\PP$ is an exact $\infty$-tripos. A logical map
$\PP\to\QQ$ gives $\widehat F$ preserving products and universal
families, and hence a logical map $\mathsf Q\PP\to\mathsf Q\QQ$
over the original base functor. This construction is coherent by
exact completion.

For exact $\QQ$, the regular doctrine reflection restricts to
\[
 \Map_{\Trip}(\mathsf Q\PP,\QQ)
 \simeq\Map_{\Trip}(\PP,\QQ).
\]
Indeed the completions of the two source closed-type categories are
canonically equivalent to $\E$, and on both sides logicality says
that the same regular functor preserves products and the same chosen
classes. The regular unit is logical because its extension is the
identity on $\E$ and its generating local classes are unchanged.
This proves the reflection. If $\PP$ is exact, $j$ is an equivalence,
as are its slice functors, and the base is unchanged. Its unit is
therefore an equivalence of $\infty$-triposes.

On exact $\infty$-triposes, the proof of the fundamental mapping-space
equivalence replaces $\Ex(\C_\PP)$ by $\C_\PP$. This gives
$\mathrm{ev}_1\dashv\Cod$. Finally,
$(\mathsf Q\PP)(1)\simeq\Ex(\PP(1))$, with the same local classes
and classifiers, proves $\TT\simeq\mathrm{ev}_1\mathsf Q$.
\qed\end{proof}

\begin{corollary}[Invariance under equivalence of completed semantics]
\label{cor:morita}
Let $W$ be the logical maps of $\infty$-triposes sent to equivalences by
$\TT$. Then $\TT$ induces an equivalence $\Trip[W^{-1}]\simeq\Top$.
\end{corollary}

\begin{proof}
The counit of $\TT\dashv\Cod$ is an equivalence, and applying
$\TT$ to every unit component gives an equivalence. Thus all unit
components belong to $W$. The functor $\TT$ factors through the
localization, and the composite of $\Cod$ with its localization
functor provides an inverse. The unit becomes an equivalence there,
and the counit is already an equivalence in $\Top$. These natural
comparisons identify the two composites with the respective identity
functors and retain their coherent adjunction data.
\qed\end{proof}

\section{Conditions for Logical Completion and Universe Reconstruction}\label{app:completion}
This appendix proves the source completion criterion, the reconstruction of dependent products
within a universe level, projective transfer from ordinary sources through
zero-truncated projective covers, and the examples separating logical closure
from quotient closure.

\subsection{Products and Conservativity}\label{app:lc-products}
For a bounded regular $\infty$-category $\C$, set
\begin{equation}\label{eq:lc-ambient}
 h:\C\into\HH_\C=\Sh(\C,\Spc),\qquad
 \E=\Ex(\C)\subseteq\HH_\C.
\end{equation}
By Lemma~\ref{lem:lc-atlas}, the condition $\mathrm{CF}_{<\infty}(\C)$ asserts that
every finite-atlas section problem over an original base admits section codes in $\C$
satisfying the weak adequacy requirement of Definition~\ref{def:lc-codes}.

\begin{theorem}[Conservative logical completion criterion]\label{thm:lc-products}
For a bounded regular $\infty$-category $\C$,
\begin{equation}\label{eq:lc-criterion}
 \mathrm{CF}_{<\infty}(\C)\quad\Longleftrightarrow\quad
 \Ex(\C)\text{ is locally cartesian closed}.
\end{equation}
These products agree with the ambient ones and satisfy Beck--Chevalley.
For an $n$-truncated second map, the construction uses an image and at most
$n+1$ additional nested groupoid quotients, relative to the given code
tower and input presentations, over every completed base.

For every regular source, the unit $j:\C\to\Ex(\C)$ is fully faithful
and induces natural equivalences
\begin{equation}\label{eq:lc-subobjects}
 \Sub_\C(T)\simeq\Sub_{\Ex(\C)}(jT),
\end{equation}
and preserves every individual dependent product already existing in $\C$.
\end{theorem}
\begin{proof}
Lemma~\ref{lem:lc-atlas} identifies finite-atlas objects with $\E=\Ex(\C)$
and shows that this category is bounded. Every arrow of $\E$ therefore has
some finite relative truncation bound. By Definition~\ref{def:lc-codes},
$\mathrm{CF}_{<\infty}$ is exactly the conjunction of the conditions
$\mathrm{CF}_n$ for finite $n$. Theorem~\ref{thm:relative-fullness} applies
to every pair, so $\mathrm{CF}_{<\infty}$ gives all dependent products, agreement with
the ambient ones, Beck--Chevalley, and the claimed relative construction
bound. Conversely, local cartesian closure gives the products in each
finite height, and the converse direction of that theorem gives every
$\mathrm{CF}_n$, hence $\mathrm{CF}_{<\infty}$.

We prove the conservativity assertions for an arbitrary regular $\C$.
Full faithfulness and regularity of $j$ are properties of regular exact
completion. If $m:X\into jT$ is a subobject, choose an original atlas
$jA\onto X$. The composite $jA\to jT$ is $jf$ for some $f:A\to T$.
Its image in $\E$ is $m$, since the first map is a cover. Regularity of $j$
identifies that image with $j(\operatorname{im}(f)\into T)$. Thus every
subobject of $jT$ is original. Full faithfulness identifies its order and
pullbacks with those in $\C$, proving
\begin{equation}\label{eq:lc-app-subobjects}
 \Sub_\C(T)\simeq\Sub_\E(jT).
\end{equation}
Finally, suppose an individual source pair has product $Q$ in $\C$.
Lemma~\ref{lem:old-pi} identifies $hQ$, with its evaluation, with the
ambient product. Its ambient adjunction restricts to the full subcategory
$\E$, which is closed under finite limits. Thus $jQ$ is also the product in
$\E$.
\qed\end{proof}

\subsection{Dependent-Product Reconstruction within a Universe Level}\label{app:fullness}
The structure specified in Definition~\ref{def:classes} keeps section codes and bounded groupoid
quotients within one chosen level.

\begin{definition}[Axioms for reconstruction within universe levels]\label{def:classes}
Let $\C$ be an exact $\infty$-category. We consider a structure on $\C$ consisting of a directed index set $I$, a distinguished element $\mathrm{prop}\in I$, and, for each $i\in I$, a finite integer $n_i\ge-1$, a full replete class $\M_i$ of arrows of $\C$, and a family $p_i:E_i\to V_i$. These data have $n_{\mathrm{prop}}=-1$ and satisfy $\M_i\subseteq\M_j$ whenever $i\le j$, together with the following axioms. Every assertion involving $i$ is required for all $i\in I$.
\begin{enumerate}
\item\label{ax:basic} $\M_i$ contains identities, is pullback-stable, and consists of $n_i$-truncated maps. The classes are exhaustive: $\bigcup_i\M_i=\operatorname{Mor}(\C)$.
\item\label{ax:closures} $\M_i$ is closed under composition, relative diagonals, and image inclusions: if $g\in\M_i$, then $\Delta_g$ and $\im(g)\into\operatorname{cod}(g)$ belong to $\M_i$.
\item\label{ax:quotients} If $R_\bullet$ is a groupoid in $\C_{/X}$, every $R_k\to X$ is in $\M_i$, and $|R_\bullet|\to X$ is $n_i$-truncated, then $|R_\bullet|\to X$ belongs to $\M_i$.
\item\label{ax:fullness} Every diagram $Z\xrightarrow{g}Y\xrightarrow{f}X$ with $f,g\in\M_i$ has a weak section family whose structure map $W\to X$ belongs to $\M_i$.
\item\label{ax:generic} $\Loc(p_i)=\M_i$, where $\Loc(p_i)$ is the class of maps becoming pullbacks of $p_i$ after a cover of their codomain. Finally, $\M_{\mathrm{prop}}$ is exactly the class of monomorphisms.
\end{enumerate}
\end{definition}
The finite truncation bound is assigned level by level. Here small means membership in the chosen class $\M_i$ (or $\M$ when a single class is fixed). The quotient condition tracks the truncation level through delooping. Genericity gives local codes~\cite{Sterling23}. Each class is local for effective epimorphisms. 

We first record the closure and covering properties used in the reconstruction.

\begin{lemma}[Finite-limit closure and covers induced by weak section families]\label{lem:small-limits}
Let $\M$ contain identities and be closed under pullback, composition, and relative diagonals. For every $X$, the full $\infty$-subcategory of $\C_{/X}$ on maps in $\M$ is closed under finite limits. If $W$ is a weak section family for $f,g$, its ambient classifying map $jW\to\Pi_{jf}(jg)$ is a cover.
\end{lemma}
\begin{proof}
The terminal slice object is an identity and is small. A binary product of small objects is small by pullback and composition. Given two maps $A\rightrightarrows B$ over $X$ with $A\to X$ and $B\to X$ small, their homotopy equalizer is the pullback of $\Delta_{B/X}$ along the pair $A\to B\times_XB$. It is small over $A$ and hence over $X$. This gives finite limits. The second assertion follows from the locally-surjective criterion at the end of Appendix~\ref{app:descent}, applied in the slice over $jX$.
\qed\end{proof}

\Needspace{12\baselineskip}
We now reconstruct dependent products within a fixed universe level.

\begin{theorem}[Dependent-product reconstruction within a universe level]\label{thm:fullness}
Let $\C$ be an exact $\infty$-category and let a class $\M$ satisfy the closure conditions, including bounded groupoid quotients, and the weak section family condition of Definition~\ref{def:classes}, with bound $N$. If $f,g\in\M$ are composable and $g$ is $n$-truncated, $-1\le n\le N$, then the ambient dependent product $\Pi_{jf}(jg)$ is represented by a map $\Pi_fg\in\M$. Its construction uses at most $n+1$ nested effective groupoid quotients. These products satisfy Beck--Chevalley.
\end{theorem}
\begin{proof}
For $n=-1$, the image of the structure map of the code family represents the dependent product. For $n=0$, one effective groupoid quotient identifies codes with the same evaluation. For $n\ge0$, identities between two represented sections are themselves sections of their pointwise $(n-1)$-truncated identity family. The induction reconstructs this identity space; an effective groupoid quotient then assembles the codes into the desired dependent product.

Work in $\HH=\HH_\C$. Fix a class with bound $N$. We prove, simultaneously for all composable small pairs and all bases, that the ambient product of an $n$-truncated second map is represented by a small map, for $-1\le n\le N$. The construction applies to each specified truncation bound $n$. Let
\[
 Z\xrightarrow{g}Y\xrightarrow{f}X,\qquad S=\Pi_{jf}(jg)
\]
be the pair and its ambient product. The map $S\to jX$ is $n$-truncated, since the right adjoint on slices preserves limits and therefore finite truncation.

\emph{The image base case.}
Choose a small weak section family $h:W\to X$ with evaluation, and write $a:jW\onto S$ for its ambient code cover. For $n=-1$, $S\to jX$ is monic. The factorization
\[
 jW\onto S\into jX
\]
is the ambient image factorization of $jh$. Regularity of $j$ identifies it with the image factorization in $\C$. Thus $S\simeq j\im(h)$ over $jX$, and the image inclusion belongs to $\M$ by hypothesis. Its ambient adjunction restricts to $\C$ by fullness, proving the base case.

\emph{The identity reduction.}
Suppose the assertion holds at level $n-1$, where $n\ge0$. Choose $h,W,a$ as above and put $B=W\times_XW$. Write $Y_B=Y\times_XB$ and $b:Y_B\to B$. The two projections $B\to W$ pull evaluation back to sections $s_0,s_1$ of the pulled-back family $g$ over $Y_B$. Their identity family is the pullback
\begin{equation}\label{eq:identity-pullback}
\begin{tikzcd}
 I \arrow[r] \arrow[d,"i"'] & Z\arrow[d,"\Delta_g"]\\
 Y_B \arrow[r,"{(s_0,s_1)}"'] & Z\times_Y Z.
\end{tikzcd}
\end{equation}
Thus $i\in\M$, it is $(n-1)$-truncated, and $b\in\M$ because it is a pullback of $f$.

Consider a test object $T\to jB$ in the ambient $\infty$-category. The two composites $T\to jW\to S$ are the parameterized sections obtained from $s_0,s_1$. A map over $jB$ from $T$ to $jW\times_SjW$ is a homotopy between these two maps to $S$. Under the section adjunction for $S$, it is a homotopy between the corresponding sections over $T\times_{jB}jY_B$. By~\eqref{eq:identity-pullback}, this is a section of the pulled-back identity family. These identifications are natural equivalences of spaces. In particular, they preserve paths between homotopies and all subsequent coherence. Yoneda gives
\begin{equation}\label{eq:kernel-full}
 jW\times_SjW\simeq\Pi_{jb}(ji)\quad\text{over }jB.
\end{equation}
The induction hypothesis, applied to $(b,i)$, represents the right side by a small $K\to B$.

\emph{The effective quotient.}
The entire \v{C}ech nerve of $a:jW\to S$ now lies in the full image of $j$. Its degree-zero object is $W$ and its degree-one object is the just-constructed $K$. By the Segal identities, all higher terms are finite fibre products of $W$ and $K$. Since the image is full, the structure maps and all coherences lift as a diagram, with a contractible space of compatible choices. We obtain an internal groupoid $R_\bullet$ in $\C_{/X}$.

The maps $W\to X$ and $B\to X$ are small. Therefore $K\to X$ is small by composition, and all higher levels are small by Lemma~\ref{lem:small-limits}. Exactness provides $Q=|R_\bullet|$. The regular embedding preserves this quotient, so
\[
 jQ\simeq|jR_\bullet|\simeq S.
\]
Finite truncation is reflected by $j$, hence $Q\to X$ is $n$-truncated, in particular $N$-truncated. The bounded quotient axiom places $Q\to X$ in $\M$.

For $T\to X$ in $\C$, fullness now restricts the ambient adjunction:
\begin{align*}
 \Map_{\C/X}(T,Q)
 &\simeq\Map_{\HH/jX}(jT,S)\\
 &\simeq\Map_{\HH/jY}(j(T\times_XY),jZ)\\
 &\simeq\Map_{\C/Y}(T\times_XY,Z).
\end{align*}
This proves that $Q$ is the dependent product. Its representing equivalence, including evaluation, has a contractible space of choices. Consequently the products and comparison maps assemble coherently. Beck--Chevalley is inherited from the ambient products.

The base case is a single image. A step at level $n$ invokes a construction at level $n-1$ and then one effective quotient. Hence the maximum nesting depth is $n+1$. The bound concerns the nesting depth of effective quotients for the given weak section families.
\qed\end{proof}

\subsection{Source Coding and Finite Diagrams}\label{app:projective-transfer}
We now derive the higher section codes from weak section families in $\C$ and a projective cover. Throughout this and the next two subsections, except in Corollary~\ref{cor:classical-lifting}, assume the hypotheses on weak section families and a projective cover in Definition~\ref{def:pt-source}, and use
$\HH=\HH_\C$ and $j:\C\to\E\subseteq\HH$ from
\eqref{eq:ambient}. We identify original objects with their images
when this simplifies notation. All diagrams in the ambient
$\infty$-categories retain their coherent homotopies.

The first step is to code finite diagrams in $\C$.

\begin{lemma}[Universal quantification of monomorphisms]\label{pt:mono}
For $m:A\into B$ and $f:B\to X$ in $\C$, $\Pi_fm$ exists in $\C$ and
is the image of the structure map $W\to X$ of any weak section family for $(f,m)$.
It agrees with the ambient product. If a regular $F$ sends all source weak
section families to weak section families in a locally cartesian closed target, it preserves
these monic products.
\end{lemma}
\begin{proof}
The ambient product $S=\Pi_{jf}(jm)$ is monic over $jX$. Weak adequacy says
that $jW\to S$ is a cover: local surjectivity over representables is the
cover criterion in the regular sheaf topology. Thus $S$ is the ambient image
of $jW\to jX$, which is $j\im(W\to X)$ by regularity. Fullness restricts
its adjunction to $\C$.

For the preservation assertion, choose a source weak section family $W$. Its image
under $F$ is weakly adequate by hypothesis. Therefore the target monic
product is the image of $FW\to FX$. Regularity identifies this image with
$F\im(W\to X)=F(\Pi_fm)$. The identification commutes with the
evaluation maps, so it is the canonical dependent-product comparison.
\qed\end{proof}

\begin{lemma}[Finite diagram codes]\label{pt:diagram-codes}
Let $X\in\C$, and let $V_\bullet\to U_\bullet$ be a map of finite
semisimplicial diagrams in $\C_{/X}$, with degrees $0,\ldots,m$.
There is $W\to X$ in $\C$ with a strict diagram map
\[
 W\times_XU_\bullet\longrightarrow V_\bullet
 \quad\text{over }U_\bullet
\]
which is adequate for all such diagram maps with a projective parameter
$T\to X$, with a code on $T$ itself. A regular functor preserving source weak
section families transports this construction to the analogous target code
construction, provided monic products are compared as in Lemma~\ref{pt:mono}.
\end{lemma}
\begin{proof}
For each $k$ choose a weak section family $W_k\to X$ for
$V_k\to U_k\to X$, and put $W^{\rm pre}=\prod_{k=0}^mW_k$ over $X$.
It has a prospective section $a_k$ in each degree. For every face
$d:U_k\to U_l$, the equation
\[
 d_Va_k=a_l(\id\times d_U)
\]
is an equality subobject of $W^{\rm pre}\times_XU_k$. Apply the monic
product along $W^{\rm pre}\times_XU_k\to W^{\rm pre}$ and intersect the
finitely many resulting subobjects. The intersection $W\into W^{\rm pre}$
has all the required face equations. The finite indexing category includes
all face maps; alternatively adjacent faces and their identities suffice.

Given a diagram map with projective parameter $T$, weak adequacy in
each degree gives a cover of $T$ and a code. Projectivity splits each
cover. The resulting finitely many maps $T\to W_k$ reproduce the prescribed
degreewise maps, so all face equations hold and the product map factors
through $W$. Everything used to define $W$ is a finite limit, a chosen weak
section family, or a monic product, proving the functorial assertion.
\qed\end{proof}

\subsection{Finite Projective Resolutions and Section Codes}\label{app:projective-codes}

\begin{lemma}[Projectives and descent]\label{pt:projectives}
The objects of $\Pj$ remain projective in $\E$ and cover every object of
$\E$. If $X_0\onto X$ is a cover in $\E$ and an ambient family over $X$
is represented after pullback to $X_0$, it is represented in $\E$.
\end{lemma}
\begin{proof}
For a cover $A\onto B$ and $jP\to B$, cover $A\times_BjP$ by $jC_0$.
The composite $jC_0\onto jP$ comes from an original cover: the regular full
embedding reflects the image factorization of an original map. A source
section $P\to C_0$ gives the desired lift. Composing an original-object
atlas with a projective cover proves the covering assertion.

The descent assertion follows from Lemma~\ref{lem:exact-descent},
applied to the full exact embedding $\E\subseteq\HH_\C$.
\qed\end{proof}

Finite semisimplicial resolutions encode the descent data through face
maps and their coherent identities. A partial semisimplicial
hypercover $U_{\le m}\to Y$ has augmented matching maps
that are effective epimorphisms through degree $m$. Matching objects are finite homotopy limits.
Here $|\sk_mU|$ denotes the colimit of the finite diagram in degrees
$0,\ldots,m$, equivalently the realization of its left Kan extension to
all semisimplicial degrees.

\begin{lemma}[Finite semisimplicial descent]\label{pt:finite-descent}
In an $\infty$-topos, let $m\ge0$ and let
$U_{\le m}\to Y$ be a partial semisimplicial hypercover. Its partial
realization $A=|\sk_mU|$ satisfies
\begin{equation}\label{pt:finite-connectivity}
 \tau_{\le m-1}^{Y}(A)\simeq Y.
\end{equation}
Here $\tau_{\le r}^{Y}$ denotes the $r$-truncation functor on the ambient slice over $Y$.
Equivalently $A\to Y$ is left orthogonal to every $(m-1)$-truncated map,
also after every base change. In particular, restriction induces an
equivalence between the full spaces of sections of an $(m-1)$-truncated
family over $Y$ and over $A$.
\end{lemma}
\begin{proof}
Work in the ambient slice over $Y$. Extend the partial diagram by its
relative right Kan extension $\widehat U=\cosk_m^Y U$. Its augmented
matching maps are the given covers through degree $m$ and equivalences
above degree $m$. It is therefore a full semisimplicial hypercover.
By the semisimplicial hypercover theorem
\cite{Stefanich23}, the augmentation
$|\widehat U|\to Y$ is $\infty$-connective.

The skeletal filtration from $|\sk_m\widehat U|$ to
$|\widehat U|$ successively attaches the cells
\[
 \widehat U_k\otimes\partial\Delta^k
       \longrightarrow \widehat U_k\otimes\Delta^k,
 \qquad k>m,
\]
by pushouts in the slice; $\otimes$ denotes tensoring an object with a
space. If $Z\to Y$ is $(m-1)$-truncated, then
$\Map_{/Y}(\widehat U_k,Z)$ is an $(m-1)$-type. Restriction from
$\Delta^k$ to $\partial\Delta^k\simeq S^{k-1}$ induces an equivalence
on maps to every such type, since $k-1\ge m$. Mapping the skeletal
filtration into $Z$ therefore yields an equivalence
\[
 \Map_{/Y}(|\widehat U|,Z)
       \simeq \Map_{/Y}(|\sk_mU|,Z).
\]
Mapping the $\infty$-connective augmentation into this truncated object
also yields an equivalence. Their composite proves
\eqref{pt:finite-connectivity}, including its assertion about full mapping
spaces. Matching covers and these colimits are stable under pullback in
an $\infty$-topos, so the same argument applies after every base change.
Orthogonality to maps with the fixed finite truncation bound proves the result
in every $\infty$-topos.
\qed\end{proof}

\begin{lemma}[Transfer of section codes]\label{pt:transfer}
Let $Z\xrightarrow gY\xrightarrow fX$ lie in $\E$, with $X\in\Pj$ and
$g$ $n$-truncated, $n\ge-1$. There is an original $W\to X$ and an
evaluation $W\times_XY\to Z$ which is a weak section family in $\E$.
The construction uses projective resolutions through degree $n+1$.
\end{lemma}
\begin{proof}
Set $m=n+1$. We first resolve $Y$ by original projectives, then resolve
the pullback of $Z$ over this diagram. Source diagram codes encode sections
between these resolutions; finite semisimplicial descent turns their
evaluations into sections of $g$.

Construct a partial projective semisimplicial hypercover
$U_{\le m}\to Y$ in the slice over $X$, choosing at each degree a projective
cover of the augmented matching object. All matching objects lie in $\E$
by finite limits. The unaugmented $U$ is a diagram of original projective
objects: their mutual mapping spaces are discrete, so its coherent face
identities are strict identities in $\C$.

Let $D_\bullet=Z\times_YU_\bullet$. Construct a relative Reedy-covering
resolution $V_{\le m}\to D_{\le m}$ by choosing
\begin{equation}\label{pt:relative-resolution}
 V_k\onto D_k\times_{M_kD}M_kV,\qquad V_k\in\Pj,
\end{equation}
where $M_0D=M_0V=X$ and subsequent matching objects use the lower face
diagrams over $X$. The composite $V_\bullet\to U_\bullet$ is again a
finite diagram in $\C_{/X}$.

Lemma~\ref{pt:diagram-codes} produces $W$ and a universal diagram section.
Composition with $V_\bullet\to D_\bullet\to Z$ gives an ambient map
\[
 W\times_X|\sk_mU|\longrightarrow Z\quad\text{over }Y.
\]
Products preserve these colimits in the ambient topos. By
\eqref{pt:finite-connectivity} and $n$-truncation of $g$, this map extends
uniquely, in the sense of a contractible space of extensions, to
$e:W\times_XY\to Z$.

Consider a section $s:T\times_XY\to Z$ with $T\in\Pj$, and put
$A_k=T\times_XU_k$. It determines a coherent diagram map
$A_\bullet\to D_\bullet$. Suppose compatible lifts have been chosen below
degree $k$. The lower lifts and the prescribed map $A_k\to D_k$ determine
$b_k:A_k\to D_k\times_{M_kD}M_kV$. The object $A_k$ is projective because
$T,X,U_k\in\Pj$ and $\Pj$ is finite-limit closed. Since the relative
matching map is a cover, choose a lift in
\[
\begin{tikzcd}
A_k \arrow[r,dashed,"\lambda_k"] \arrow[dr,"b_k"'] &
V_k \arrow[d,two heads] \\
& D_k\times_{M_kD}M_kV.
\end{tikzcd}
\]
The comparison given by this lift extends the preceding coherent
diagram map to degree $k$, by the matching-object universal property.
Induction therefore gives $A_\bullet\to V_\bullet$ over $U_\bullet$,
together with a coherent comparison to the map induced by $s$.
All $A_k,V_k$ are original objects, so their face equations are equalities
in $\C$, by full faithfulness and discreteness of the relevant mapping
spaces. Lemma~\ref{pt:diagram-codes} gives a code $T\to W$.
Finite semisimplicial descent identifies its evaluation with $s$,
including the coherent comparison.

For an arbitrary parameter $T\in\E_{/X}$, first take a projective cover
$T'\onto T$ and apply the preceding argument. Thus $W$ is weakly adequate.
\qed\end{proof}

\begin{corollary}[Source-local section codes over every original base]
\label{pt:source-sc}
For every original $X\in\C$ and every pair
$Z\to Y\to jX$ in $\E$, there is an original $W\to X$ with an
evaluation $jW\times_{jX}Y\to Z$ which represents every section with an
original parameter $jS\to jX$ after an original regular-epimorphism
cover $S'\onto S$. In particular, $(\C,\Pj)$ satisfies coherent fullness
$\mathrm{CF}_{<\infty}$ of Definition~\ref{def:lc-codes};
coding for quotient inputs is a conclusion of the source hypotheses.
\end{corollary}
\begin{proof}
Choose an original projective cover $P\onto X$. Pull back the pair to
$jP$ and apply Lemma~\ref{pt:transfer}, obtaining an original $W\to P$.
The composite $W\to X$ has the required evaluation, since
$jW\times_{jX}Y\simeq jW\times_{jP}Y_P$.

Given a section with parameter $jS$, choose an original projective cover
$S'\onto S\times_XP$. Its composite $S'\onto S$ is an original regular
epimorphism. The projective-parameter part of the transfer proof represents
the pulled-back section by a map $S'\to W$ over $P$, on that same parameter. Composing over $X$ proves the assertion. Projectivity
gives the lifts on every finite resolution level after this single
cover of the parameter $S$.
\qed\end{proof}

We now apply these section codes to prove the logical completion theorem for sources with projective covers.

\begin{proof}[of Theorem~\ref{thm:pt-completion}]
Corollary~\ref{pt:source-sc} proves $\mathrm{CF}_{<\infty}$. Hence
Theorem~\ref{thm:lc-products} gives local cartesian closure, agreement
with ambient products, Beck--Chevalley, full faithfulness, the natural
subobject isomorphisms, and preservation of existing products.
Exactness and objectwise boundedness are the recalled properties of
higher exact completion of an ordinary regular source.

Finite limits and images preserve equality, finite conjunction, and
existential quantification. Lemma~\ref{pt:mono} constructs every
universal quantifier on source predicates and identifies it with its
ambient product. Implication between $A,B\into X$ is
$\Pi_{A\into X}(A\times_X B\into A)$ and is therefore preserved
as well. The subobject isomorphisms reflect these operations.
\qed\end{proof}

\begin{corollary}[Classical and higher local cartesian closure]\label{cor:classical-lifting}
Let $\C$ be an ordinary regular category with a finite-limit-closed
covering class of regular projectives. The following are equivalent:
\begin{enumerate}
\item $\Exone(\C)$ is locally cartesian closed;
\item every composable pair in $\C$ has a weak section family;
\item $\Ex(\C)$ is locally cartesian closed.
\end{enumerate}
\end{corollary}

\begin{proof}
Suppose first that $\E_0=\Exone(\C)$ is locally cartesian closed,
and write $\iota:\C\into\E_0$ for its regular unit. For a source
pair $Z\xrightarrow gY\xrightarrow fX$, form
$S=\Pi_{\iota f}(\iota g)$ in $\E_0$ and choose an original atlas
$\iota W\onto S$. Its evaluation has source domain
$\iota(W\times_XY)$ and source codomain $\iota Z$, so comes from a
source evaluation by full faithfulness. Given a source-parameterized
section, transpose it to a map $\iota T\to S$ and pull the atlas back.
Cover the resulting object by a source object $\iota T'$. The
composite $\iota T'\onto\iota T$ is a regular epimorphism, hence
comes from a source cover $T'\onto T$, since the regular embedding
reflects image factorizations. Its map to $\iota W$ gives the desired
code. Thus every source pair has a weak section family.

The implication from weak section families in the source to local cartesian closure
of $\Ex(\C)$ is Theorem~\ref{thm:pt-completion}, using the stated
finite-limit-closed projective cover.

Finally, suppose $\E=\Ex(\C)$ is locally cartesian closed. Its
zero-type subcategory is $\Exone(\C)$ by
\eqref{eq:import-shadow}. If $Z\to Y\to X$ consists of zero-types,
the map $Z\to Y$ is $0$-truncated. A right adjoint on slices preserves
limits, hence preserves $0$-truncated objects of that slice by the
diagonal criterion. Its dependent product is consequently a
$0$-truncated family over the zero-type $X$, and its total object is
a zero-type. The product adjunction restricts to zero-types, proving
local cartesian closure of $\Exone(\C)$.
\qed\end{proof}

\subsection{Transporting Codes and Interpretations}\label{app:projective-morphisms}

We next transport section codes using projective parameters and finite descent to obtain the characterization of dependent-product-preserving extensions of source functors in Theorem~\ref{thm:pt-morphisms}.

\begin{theorem}[Extension criterion from weak section families]\label{thm:pt-morphisms}
Assume the hypotheses on weak section families and a projective cover in Definition~\ref{def:pt-source}. Let $\D$ be an exact locally
cartesian closed $\infty$-category with a full replete, finite-limit-closed
covering subcategory $\Qj$ of $0$-truncated projectives. If
$F:\C\to\D$ is regular and $F(\Pj)\subseteq\Qj$, its regular
extension $\bar F:\E\to\D$ preserves dependent products if and
only if $F$ sends every weak section family in $\C$ to a
weak section family in $\D$. Restriction gives an equivalence of spaces
\begin{equation}\label{pt:logical-up}
 \Fun_{\mathrm{reg},\Pi,\Pj\to\Qj}(\E,\D)\core
 \simeq
 \Fun_{\mathrm{reg},\mathrm{weak},\Pj\to\Qj}(\C,\D)\core.
\end{equation}
Weak adequacy in the target is tested against all objects of $\D$.
For each eligible $F$, the space of dependent-product-preserving extensions with an
identification of their restriction with $F$ is contractible.
With chosen universe presentations, adding
\eqref{eq:lc-marked-classes} restricts this equivalence to
functors preserving dependent products and universes.
\end{theorem}

\begin{proof}
Use the regular Yoneda embedding
\[
 h_\D:\D\into\HH_\D
     :=\operatorname{Sh}_{\mathrm{reg}}(\D,\Spc),
\]
after an external universe enlargement if necessary. It is fully
faithful and preserves finite limits, covers, and existing dependent
products, by Lemma~\ref{lem:old-pi}. Thus $\HH_\D$ gives the finite semisimplicial realizations used
for descent, while fullness identifies the resulting maps in $\D$.

Suppose $F$ preserves every weak section family in $\C$.
Lemma~\ref{pt:mono} shows that it preserves the monic products used to
impose face equations. It consequently transports the finite diagram
code construction of Lemma~\ref{pt:diagram-codes}. Consider first the
weak section family $W\to P$ obtained by Lemma~\ref{pt:transfer} for a pair
$Z\xrightarrow gY\to P$, with $P\in\Pj$ and $g$ $n$-truncated.
The regular extension $\bar F$ preserves the finite matching limits
and their covers. Given a target projective parameter $T\in\Qj$
over $FP$, every $T\times_{FP}F(U_k)$ belongs to $\Qj$ and is
therefore projective. The same successive lifts through the relative
matching covers exist. Target weak adequacy of the images of the
weak section families in $\C$ gives their diagram codes on the same projective parameter.

The transported evaluation agrees with the desired target evaluation.
Indeed, $\bar F(e)$ has the transported finite-level evaluations, and
$\bar F(U_{\le n+1})\to\bar F(Y)$ is still a partial hypercover.
Finite-limit functors preserve finite truncation bounds by the diagonal
criterion, so $\bar Fg$ is $n$-truncated. Applying
Lemma~\ref{pt:finite-descent} in $\HH_\D$ identifies extensions of
these finite-level evaluations by a contractible space of comparisons.
Thus every target section with parameter in $\Qj$ is represented by
$F(W)$. Since $\Qj$ covers every object of $\D$, this proves weak
adequacy for all target parameters.

Now let the original base be arbitrary, and fix the source codes
constructed in Corollary~\ref{pt:source-sc}: choose $P\onto X$ in
$\Pj$ and regard the preceding $W\to P$ as a code over $X$. For a
target section with parameter $T\to FX$, pull back the target cover
$FP\onto FX$ to obtain $T\times_{FX}FP\onto T$. The preceding
argument represents its restriction after a further target cover.
Their composite is a cover of $T$, so $F$ preserves this
chosen source code.

Theorem~\ref{thm:lc-morphisms} now proves preservation of
all dependent products.

Conversely, let $\bar F$ preserve products, and take any
weak section family $W$ in $\C$ for $Z\to Y\to X$. Its ambient classifying
map covers the ambient product: the regular-site local-surjectivity criterion applies to every such pair. By Theorem~\ref{thm:pt-completion} this product
belongs to $\E$. Regularity and product preservation transport the
code cover to the target product, proving weak adequacy of $F(W)$ in
$\D$.

The regular universal property of exact completion restricts to
\eqref{pt:logical-up}. Both preservation conditions and the condition
$F(\Pj)\subseteq\Qj$ are invariant under coherent equivalence;
here membership in the selected projective subcategories is understood up to equivalence.
They therefore specify unions of components of the regular mapping
spaces. The restricted equivalence retains natural equivalences and
all higher homotopies. Its homotopy fibre over an eligible $F$,
including the specified identification of its restriction, is
contractible.

Finally, for chosen universe presentations, Corollary~\ref{cor:lc-universes}
restricts~\eqref{pt:logical-up} to interpretations preserving the chosen
universes precisely under condition~\eqref{eq:lc-marked-classes}. The
compatible identification data have a contractible space of choices.
\qed\end{proof}

\subsection{Universe Presentations, Closed Levels, and Doctrines}\label{app:lc-universes}
We treat individual universe presentations under $\mathrm{CF}_{<\infty}$ and show
that conditions (H1)--(H3) imply this coding condition.

Under (H1)--(H3), the completed category has the structure specified in
Definition~\ref{def:classes}. The proof reconstructs dependent products while keeping
the entire section-code tower within one universe level.

\begin{corollary}[Univalent hierarchy and logical extensions preserving the chosen universes]\label{cor:lc-universes}
Let $\C$ be bounded regular, and let $p_i$ be finite-atlas weak universe
presentations over original bases. If $\mathrm{CF}_{<\infty}(\C)$ holds, each has
a univalent completion satisfying
\begin{equation}\label{eq:lc-universes}
 \Map_\E(X,U_i)\simeq(\E/X)_{\M_i}\core,
 \qquad \M_i=\Loc_\E(p_i).
\end{equation}
The map $u_i$ is $n_i$-truncated and $U_i$ is $(n_i+1)$-truncated.
Conditions (H1)--(H3) imply $\mathrm{CF}_{<\infty}(\C)$, and then the hierarchy
satisfies the five conditions for universe closure, is exhaustive, closed, and quotient-stable,
and has propositional classifier $1\into\Omega$.
Under (H1)--(H3), $\Cod(\C)$, equipped with these presentations, is an $\infty$-tripos.
For a regular $F:\C\to\D$ preserving the selected section codes, with
$\D$ exact and locally cartesian closed and with chosen univalent families classifying
$\M_i^{\D}$, the additional condition
\begin{equation}\label{eq:lc-marked-classes}
 \Loc_\D(\bar Fp_i)=\M_i^{\D}\qquad(i\in I)
\end{equation}
restricts~\eqref{eq:lc-logical-up} to the full spaces of interpretations preserving
the chosen universal families.
\end{corollary}

\begin{proof}
Lemma~\ref{lem:local-recoding} identifies the source-defined classes with
the corresponding local classes:
\begin{equation}\label{eq:lc-app-local-class}
 \M_i=\Loc_{\E}(p_i).
\end{equation}
They are pullback-stable and local for covers of arbitrary completed bases.

\emph{Section codes from (H1)--(H3).}
Suppose (H1)--(H3). Given $Z\xrightarrow gY\to hT$, condition (H1) and
directedness put $Y\to hT$ and $Z\to hT$ in a common $\M_i$.
The graph $Z\to Z\times_{hT}Y$ of $g$ is a pullback of
$\Delta_{Y/hT}$, so it lies in $\M_i$. The projection
$Z\times_{hT}Y\to Y$ is a pullback of $Z\to hT$, and composition
closure places $g$ in that same class. Thus (H3) gives an original
code for this arbitrary original-base pair, with adequacy after original
covers. These source-local codes establish $\mathrm{CF}_{<\infty}$.

\emph{Univalent classification.}
Theorem~\ref{thm:lc-products} gives dependent products in $\E$ agreeing with the ambient ones. Theorem~\ref{thm:univalent-completion}, applied to the finite truncation bound of $p_i$, gives its univalent completion $u_i:\widetilde U_i\to U_i$. By~\eqref{eq:lc-app-local-class}, its classifying equivalence is
\begin{equation}\label{eq:lc-app-classification}
 \Map_{\E}(X,U_i)\simeq
 (\E_{/X})_{\M_i}\core.
\end{equation}
This construction uses the products already established globally. We next use (H1)--(H3) to prove same-level closure.

Theorem~\ref{thm:univalent-completion} also shows that $u_i$ is $n_i$-truncated and $U_i$ is $(n_i+1)$-truncated.

\emph{Closure of the hierarchy.}
Every arrow $q:A\to B$ admits an original atlas $hT\onto B$. By (H1), its pullback belongs to $\M_i$ for some $i\in I$, and~\eqref{eq:lc-app-local-class} puts $q$ in that same level. This proves exhaustivity. Truncation is local and follows from $p_i$. The same argument extends the propositional class from monomorphisms over original bases to all monomorphisms of $\E$.

For a pair $f,g\in\M_i$ over $hT$, repeat the identity induction of Theorem~\ref{thm:relative-fullness} over original bases with a witness whose structure map lies in $\M_i$, as required by (H3). At the propositional base case, the product is the image inclusion of that structure map, so (H2) makes it small. At a higher step, $B=W\times_TW$ is small over $T$. The indexing map $Y_B\to hB$ and identity-family map $I\to Y_B$ are small by pullback and relative-diagonal closure. The induction hypothesis shows that the kernel of the code cover, identified with that identity-family section object, is small over $hB$, hence over $hT$ by composition. Its \v{C}ech terms are small over $hT$: pullbacks, composition, and diagonals give finite limits of small objects in that slice. The quotient product is $n$-truncated, with $n\le n_i$, so the bounded quotient clause of (H2) applies. This proves same-level product closure over original bases. Beck--Chevalley and locality descend that closure to all bases.

At the propositional level,~\eqref{eq:lc-app-classification} identifies maps to $U_{\mathrm{prop}}$ with subobjects. The propositional case of Theorem~\ref{thm:univalent-completion} identifies the universal family with $1\into\Omega$. We have obtained the stated exhaustive closed bounded hierarchy. Theorem~\ref{thm:generic-tests} now verifies all five conditions for universe closure. The bounded quotient clause of (H2) gives quotient stability, so the source codomain doctrine, with these presentations and its coherent-fullness witnesses, is an $\infty$-tripos.

\emph{Preservation of the chosen universes in the extension equivalence.}
By Theorem~\ref{thm:lc-morphisms}, $G=\bar F$ preserves dependent
products. Equation~\eqref{eq:univalent-preservation} gives
\[
 G(u_i)\simeq\operatorname{Uni}(\bar Fp_i).
\]
This family classifies $\Loc_{\D}(\bar Fp_i)$, so
Lemma~\ref{lem:classifier-comparisons} identifies it with the chosen target
universal family precisely under condition~\eqref{eq:lc-marked-classes},
with a contractible space of compatible identifications. The condition is
invariant under equivalence of regular functors. Restricting the
mapping-space equivalence~\eqref{eq:lc-logical-up} to these components
therefore proves the assertion.
\qed\end{proof}

The preceding corollary gives the quotient-stable univalent hierarchy of Theorem~\ref{thm:controlled-source}.

\begin{proof}[of Theorem~\ref{thm:controlled-source}]
Apply the assertion under (H1)--(H3) in
Corollary~\ref{cor:lc-universes}. Its proof derives section coding
from (H1)--(H3), obtains univalent classification, and reconstructs
same-level products before applying Theorem~\ref{thm:generic-tests}.
The bounded quotient clause of (H2) gives quotient stability.
Same-level product closure gives a weak section family over
any completed base: take the product itself with its evaluation.
Together with the stated class closures and weak universes, this
verifies the axioms of Definition~\ref{def:classes}.
\qed\end{proof}

\begin{corollary}[Conservative completion of a type doctrine]\label{cor:lc-doctrine}
Let $\PP$ be a regular type doctrine whose closed-type $\infty$-category
is bounded and satisfies $\mathrm{CF}_{<\infty}$. Its completed doctrine
$\mathsf Q\PP$ has effective higher quotients and dependent products. The unit has fully faithful fibre functors, identifies the
subobject lattices of original fibre objects with those of their images, and
preserves all existing dependent products. A hierarchy satisfying the five conditions for universe
closure gives closed univalent universes; presentations satisfying (H1)--(H3) additionally
provide quotient-stable levels and section codes in the source category belonging to the same universe level. For an exact
locally cartesian closed $\infty$-category $\D$, there is an equivalence
\begin{equation}\label{eq:lc-doctrine-logical}
 \Fun_{\mathrm{reg},\Pi}(\Ex(\PP(1)),\D)\core
       \simeq\operatorname{Int}_{\mathrm{codes}}(\PP,\Cod(\D)).
\end{equation}
The right side is the space of regular doctrine interpretations whose
closed-type functors preserve the selected section codes.
When universal families are chosen in the completion and in $\D$,
restrict the left side to functors preserving them and the right side to
interpretations satisfying~\eqref{eq:lc-marked-classes}.
\end{corollary}

\begin{proof}
Comprehension identifies the original fibre $\PP(X)$ with $\C_{/\delta_\PP X}$. The free exact completion has fibre $\E_{/h\delta_\PP X}$ by Theorem~\ref{thm:regular-completion}; its closed-type category is exact. Theorem~\ref{thm:lc-products} gives its dependent products, and a hierarchy satisfying the five conditions for universe closure gives closed univalent universes by Theorem~\ref{thm:tripos-construction}. For presentations satisfying (H1)--(H3), Corollary~\ref{cor:lc-universes} also gives quotient stability and section codes in the source category belonging to the same universe level. The unit is the slice embedding induced by $h$. It is fully faithful and preserves every original dependent product. A subobject of an original slice object is a subobject of its original domain with the induced structural map, so~\eqref{eq:lc-app-subobjects} proves fibrewise predicate conservativity.

By Lemma~\ref{lem:doctrine-maps}, an interpretation of $\PP$ in $\Cod(\D)$ consists of a regular closed-type functor $F:\C\to\D$, a base functor $v$, and an equivalence $F\delta_\PP\simeq v$, since the codomain doctrine's $\delta$ is the identity. For fixed $F$, the space of such $v$ with its comparison is contractible. The fibre functors of this interpretation are the induced slice functors. The chosen-code mapping-space equivalence of Theorem~\ref{thm:lc-morphisms}, with the refinement to interpretations preserving the chosen universes when required, thus becomes the asserted equivalence of interpretation spaces. All higher comparisons are retained by the regular universal property and the comprehension equivalence.
\qed\end{proof}

\begin{corollary}[Quotient-stable hierarchies]\label{cor:quotient-stable-adjunction}
The adjunction~\eqref{eq:adjunction} restricts to $\infty$-triposes whose local
classes are quotient-stable and exact elementary $\infty$-toposes with
quotient-stable closed bounded hierarchies. The right adjoint remains fully
faithful. Presentations satisfying (H1)--(H3) give a sufficient source criterion for this
restricted adjunction.
\end{corollary}

\begin{proof}
Quotient stability is a property of the classified local classes. The
completion of a source in the restricted subcategory has exactly those
classes, and the codomain doctrine of a target with quotient-stable classes
has the same classes. The unit induces the identity on the completed
category and its classes; the counit is the exact-completion equivalence.
Hence both preserve quotient stability. Restricting the mapping-space
adjunction of Theorem~\ref{thm:fundamental} gives the claimed adjunction,
with the counit still an equivalence. Corollary~\ref{cor:lc-universes} gives
the sufficient source criterion.
\qed\end{proof}

\subsection{Classical Recovery and the Bounded-Space Model}\label{app:shadow}\label{app:model}
\begin{proof}[of Theorem~\ref{thm:shadow}]
Write $\E=\Ex(\C_\PP)$ and $dI=j\delta_\PP I$. Since every set is
$0$-truncated and $d$ preserves finite limits, each $dI$ is $0$-truncated.
Subobject conservativity identifies $H(I)$ with $\Sub_\E(dI)$.
For an $H$-valued partial equivalence relation $\rho$ on $I$, let
$S_\rho\into dI$ be its support. The relation is an internal equivalence
relation on this support; exactness gives its quotient $Q_\rho$ in $\E$.
Pulling back $\Delta_{Q_\rho}$ along the product of the support cover with itself gives the
monomorphism defined by $\rho$. Descent of monomorphisms makes
$\Delta_{Q_\rho}$ monic, so $Q_\rho$ is a zero-type.

A map $Q_\rho\to Q_\sigma$ pulls back along the supports to the subobject
$S_\rho\times_{Q_\sigma}S_\sigma\into S_\rho\times S_\sigma$.
It is an $H$-predicate, total on the first support, saturated for the two
relations, and single-valued modulo $\sigma$. Conversely such a predicate
first descends to a map from the source support, by totality and
single-valuedness, and then through its equivalence relation by saturation.
These constructions are inverse. Identities and composition agree with the
PER identities and existential relational composition of the classical
tripos-to-topos construction~\cite{HJP80}. Mapping spaces into zero-types
are discrete. Thus $\rho\mapsto Q_\rho$ is fully faithful.

For any zero-type $X\in\E$, the assumed cover $S\onto X$ with
$S\into dI$ has a monic kernel pair supported on $S\times S$.
Subobject conservativity makes it an $H$-valued PER, and its quotient is $X$.
This proves essential surjectivity.
\qed\end{proof}

From a propositional tripos one may form its regular presentation by supported
predicates and total single-valued functional relations, whose ordinary exact
completion is the classical topos. Equation~\eqref{eq:import-shadow} then gives
the same classical topos as the category of zero-types after higher exact completion. Coherent fullness and
weak universe presentations satisfying the five conditions for universe closure equip
this completion with higher logical structure, as in Theorem~\ref{thm:tripos-construction}.

\begin{proposition}\label{prop:space-model}
Let $\C_{\mathrm{set}}$ be a skeleton of sets of cardinality below some $\kappa_i$, for an increasing sequence of strongly inaccessible cardinals.
$\C_{\mathrm{set}}$ satisfies $\mathrm{CF}_{<\infty}$ and admits the weak universe presentations satisfying (H1)--(H3) of Definition~\ref{def:lc-universes}. Its higher exact completion is equivalent to the bounded-space category in Example~\ref{ex:spaces}. The resulting hierarchy is closed and quotient-stable, and both the source and completed codomain doctrines are $\infty$-triposes. The higher completion strictly enlarges the ordinary exact source $\C_{\mathrm{set}}$ by effective groupoid quotients.
\end{proposition}
\begin{proof}
Work in an external universe containing $\kappa_\omega=\sup_i\kappa_i$. Every cover in $\C_{\mathrm{set}}$ splits, so every covering sieve is maximal and its regular sheaves are all presheaves of spaces. Define
\[
 \Phi(X)(T)=\Map_{\Spc}(T_{\mathrm{disc}},X).
\]
This is right Kan extension from the full subcategory on the singleton: the relevant comma category is the discrete set of points of $T$. Hence $\Phi$ is fully faithful and sends a discrete set $A$ to $hA$.

Let $\mathcal S_b\subseteq\Spc$ be the full category of finitely truncated spaces with a simplicial representative of cardinality below some $\kappa_i$. It is closed under finite limits after taking one common cardinal level. For an internal groupoid, the object and arrow spaces lie at one common cardinal level and have one common finite truncation bound. All Segal terms are finite fibre products of these two spaces, so all admit representatives below that cardinal. The realization uses a countable collection of such data; uncountable regularity retains the size bound. Its diagonal is locally the source--target map of the groupoid and therefore has a finite truncation bound. Thus the realization belongs to $\mathcal S_b$. Exactness and pullback stability are inherited from spaces.

For each discrete $T$, the functor $(-)^T$ preserves finite limits and covers; cover preservation follows from choice on components in the
working external universe. It consequently preserves effective groupoid quotients. Therefore $\Phi(\mathcal S_b)$ is a full exact subcategory of the presheaf category containing $h\C_{\mathrm{set}}$, and it contains $\Ex(\C_{\mathrm{set}})$. Conversely, a $0$-truncated object of $\mathcal S_b$ is equivalent to an original discrete set. For an $n$-truncated object with $n\ge1$, choose one point in each component. This gives a small discrete atlas whose kernel is $(n-1)$-truncated as an absolute space. Induction constructs its kernel and then its groupoid quotient from original sets. We obtain
\begin{equation}\label{eq:lc-app-space-completion}
 \Ex(\C_{\mathrm{set}})\simeq\mathcal S_b.
\end{equation}

For a finite-atlas pair over an original set $T$, compute its fibrewise section spaces in $\Spc$. Increase the inaccessible level to contain the input presentations. Mapping complexes and the countably many limit conditions defining dependent sections then have cardinality below that inaccessible cardinal. The section spaces are finitely truncated, with a bound given by the target family. For each $t\in T$, choose one representative of each component of its section space and form their disjoint union $W$ over $T$. This is an original set, with an evaluation into the given family. A section parametrized by an original $S\to T$ determines a component at each point of $S$. Choice in the working external universe selects its representative
and a homotopy to that section. This proves $\mathrm{CF}_{<\infty}$ with the identity source cover.

For each $(i,n)$ choose a set $V_{i,n}$ of representatives of the $n$-truncated spaces admitting simplicial presentations below $\kappa_i$. There are at most $\kappa_i$ such structures on standardized underlying sets, so this is an original set at the next cardinal level. The tautological family $P_{i,n}\to hV_{i,n}$ has total presentation size at most $\kappa_i$ and finite truncation, hence is a finite-atlas family by~\eqref{eq:lc-app-space-completion} and Lemma~\ref{lem:lc-atlas}. Its generated local class consists exactly of the $n$-truncated maps with $\kappa_i$-small fibres: an original atlas is a discrete set, so the required fibrewise choices give the source-local presentation, and the converse follows from pullback stability and locality of size and truncation.

The classes increase with both bounds and form a directed, exhaustive family; each class is stable under pullback. They contain identities. Composition preserves them because a $\kappa_i$-small homotopy sum of $\kappa_i$-small fibres is $\kappa_i$-small and $n$-truncated maps compose. Relative diagonals and image inclusions have identity-space and propositional fibres, respectively, which remain small. For a groupoid over a base whose terms lie in this class, compute its quotient fibrewise. The same countable-size estimate gives $\kappa_i$-small fibres; if the quotient is $n$-truncated, it lies in the class. This proves the bounded quotient clause of (H2).

For two maps at level $(i,n)$, each fibrewise section space is $n$-truncated and $\kappa_i$-small: it is computed by mapping complexes in a slice of $\kappa_i$-small simplicial sets, and strong inaccessibility bounds all their simplices and countable limit conditions. Choosing representatives of its components gives an original code family whose structure map has $\kappa_i$-small discrete fibres when $n\ge0$. For $n=-1$, take the subset of the original base on which the section proposition is inhabited and choose its section, whose space of choices is contractible. This structure map is monic. Thus (H3) holds in every level. All propositional levels consist of all monomorphisms, and one is chosen as the distinguished level, proving (H1).

The completion theorem now gives the products and closed univalent classifiers for these classes. Alternatively, the classifier in $\Spc$ of $\kappa_i$-small $n$-types is an $(n+1)$-type with a presentation of size at most $\kappa_i$, so it and its universal total family belong to $\mathcal S_b$ at the next cardinal level. Products of two level-$(i,n)$ maps have the small, $n$-truncated fibres just calculated; increasing the cardinal level to contain their arbitrary base keeps the total space in $\mathcal S_b$. They provide weak section families at that same level. The universal families give weak universe presentations for the completed codomain doctrine. Codomain sums satisfy the pullback and cartesianness conditions of Definition~\ref{def:doctrine}, as checked in Appendix~\ref{app:regular-completion}. Theorem~\ref{thm:controlled-source} therefore verifies the five conditions for universe closure for the source codomain doctrine. Its completed hierarchy is closed and quotient-stable, and both the source and completed codomain doctrines are $\infty$-triposes.

Finally, the nerve of $\mathbb Z/2$ is an internal groupoid in $\C_{\mathrm{set}}$. Its effective quotient in the completion is $B(\mathbb Z/2)$, with loop group $\mathbb Z/2$. This higher type lies in $\Ex(\C_{\mathrm{set}})$ and witnesses the strict enlargement of the ordinary exact source.
\qed\end{proof}
Inaccessible cardinals provide the size hierarchy that realizes the abstract coding and closure hypotheses in this example. For the assembly construction, \mbox{Appendix}~\ref{app:realizability-hierarchy} uses tracked presentations to verify $\mathrm{CF}_{<\infty}$ and (H1)--(H3), and hence to construct the closed bounded hierarchy.

\subsection{A Higher Source Enlarged by Exact Completion}\label{app:higher-nonexact}
The following product construction uses actual universal families in the category
of bounded spaces and original-parameter presentations in the category of bounded
sets. The proof verifies the cover-refinement hypothesis used to compare the exact
completions.

\begin{proposition}[A higher source enlarged by exact completion]\label{prop:higher-nonexact}
Let $\C_{\mathrm{set}}$ and $\mathcal S_b=\Ex(\C_{\mathrm{set}})$
be the bounded-set source and bounded-space category of
Example~\ref{ex:spaces}. The category
\[
 \C=\mathcal S_b\times\C_{\mathrm{set}}
\]
is a bounded regular higher source whose exact completion strictly enlarges it:
\[
 \Ex(\C)\simeq\mathcal S_b\times\mathcal S_b.
\]
It admits universe presentations satisfying (H1)--(H3),
and its codomain doctrine is an $\infty$-tripos.
\end{proposition}

\begin{proof}
Put $\D=\mathcal S_b\times\mathcal S_b$. The inclusion
$\C=\mathcal S_b\times\C_{\mathrm{set}}\into\D$ is full and regular.
It satisfies the cover-refinement hypothesis of
\cite{Stefanich23}. Indeed, given a cover
$(X_1,X_2)\onto(Y_1,Y_2)$ in $\D$ with codomain in $\C$, choose a bounded
discrete atlas $A_2\onto X_2$. Then $(X_1,A_2)$ belongs to $\C$, covers
$(X_1,X_2)$, and its composite to $(Y_1,Y_2)$ is a cover. This exact-completion comparison criterion
therefore makes the regular extension $\Ex(\C)\to\D$ fully faithful.
Its image is all of $\D$: fix $A\in\mathcal S_b$ and form a groupoid-quotient
presentation of any $B\in\mathcal S_b$ from bounded discrete sets. Pair
every stage with the constant identity groupoid on $A$. Its realization is
$A$, so induction through the presentation constructs $(A,B)$ from objects
of $\C$. Hence $\Ex(\C)\simeq\D$.

Use the common size-and-height index set of Example~\ref{ex:spaces}.
Let $u_i:\widetilde U_i\to U_i$ be the actual universal family in the
first factor, and let $p_i:P_i\to V_i$ be its original-parameter weak
presentation in the second factor, where $V_i\in\C_{\mathrm{set}}$.
The pair $(u_i,p_i)$ has original parameter object $(U_i,V_i)\in\C$.
Its local class is the product of the two level classes: covers and pullbacks
are componentwise, and local classifying maps can be chosen separately and
combined by the product cover. Thus monotonicity, exhaustiveness, and the
propositional clause (H1), and all the closure operations in (H2), follow
componentwise from Proposition~\ref{prop:space-model}.
For a pair over an original $(T_1,T_2)$ in this level, take the code
\[
 (\Pi_{f_1}g_1,W_2)\longrightarrow(T_1,T_2),
\]
where the first component is the actual dependent product in $\mathcal S_b$
and the second is the original code given by (H3) for bounded sets. This
object is original and its map lies in the same level. A section is
represented exactly in the first factor and locally in the second. The
product of the identity cover in the first factor with the second source
cover proves source-local adequacy. This verifies (H3). In particular the
section-coding condition holds: (H1) places the two maps of any pair in a
common level, and (H3) gives its code. For $Z\to Y$, whose codomain
is presented, one may first put $Y\to hT$ and $Z\to hT$ in a common level;
its graph is a pullback of $\Delta_{Y/hT}$, and composition gives
$Z\to Y$ in that level. Corollary~\ref{cor:lc-universes} now gives the
claimed $\infty$-tripos structure on $\Cod(\C)$.

The first factor has higher mapping spaces; for example,
\[
 \Map_{\mathcal S_b}(1,BC_2)\simeq BC_2.
\]
The internal group $(1,C_2)$ has bar quotient $(1,BC_2)$ in $\D$.
Its second component has loop group $C_2$, whereas second components
of source objects are discrete. Thus $(1,BC_2)$ witnesses the strict
enlargement of this higher source under exact completion.
\qed\end{proof}

\subsection{Finite Groupoids and Logical Closure}\label{app:finite-groups}
\begin{theorem}[Finite-group criterion for quotient closure]\label{thm:finite-group-separation}
Let $\mathfrak F$ be an isomorphism-invariant class of finite groups containing
the trivial group and closed under subgroups and extensions. In the bounded
space category of Example~\ref{ex:spaces}, let $\M_{\mathfrak F}$ consist of
maps with fibres $X$ satisfying
\[
 |\pi_0X|<\infty,\qquad \pi_1(X,x)\in\mathfrak F,\qquad
 \pi_k(X,x)=0\quad(k\ge2).
\]
The class has a univalent classifier, contains identities, and is closed under image inclusions,
relative diagonals, dependent sums, and same-level dependent products.
It is quotient-stable at bound $1$ if and only if $\mathfrak F$ contains every
finite group.
\end{theorem}

\begin{proof}
Every fibre in the class is represented by a finite groupoid: choose one object
in each component and use its finite isotropy group. The collection of such
finite groupoid presentations is a set. Its tautological family has
$1$-truncated fibres and has a presentation within one of the size bounds of
Example~\ref{ex:spaces}. After pulling back along an original discrete atlas of the base, a family
in the class is a collection of these groupoids, and external choice selects
its codes and fibre equivalences. Thus the tautological family has local class
$\M_{\mathfrak F}$. Theorem~\ref{thm:univalent-completion} gives its univalent
classifier. The empty groupoid and terminal groupoid are included.

It suffices to check the operations fibrewise in spaces. The class contains
identities; image fibres are empty or contractible. Path spaces of a finite
groupoid are finite sets, proving relative diagonal closure.
For dependent sums, let $Z\to Y$ have fibres in the class and let $Y$ itself
belong to the class. Its total space is a $1$-type. Its components are finite
in number. At a chosen point, the homotopy sequence gives an extension
\[
 1\longrightarrow\pi_1(Z_y,z)\longrightarrow\pi_1(Z,z)
       \longrightarrow H\longrightarrow1,
 \qquad H\le\pi_1(Y,y).
\]
The first injection uses $\pi_2Y=0$; the image $H$ is the stabilizer of the
chosen fibre component. Subgroup and extension closure put the total
fundamental group in $\mathfrak F$. It is finite, giving the required finite
groupoid type.

For dependent products, use a finite groupoid model for $Y$. The family is
represented by a finite isofibration of finite groupoids: its Grothendieck
construction has finitely many objects and arrows, after choosing finite
models for each fibre. Its section groupoid therefore has finitely many
objects and arrows. For a section $s$, evaluation on the finitely many objects
of $Y$ gives an injective homomorphism
\[
 \Aut(s)\into\prod_{y\in\operatorname{Ob}Y}\Aut_{Z_y}(s(y)).
\]
A finite product of groups in $\mathfrak F$ is in $\mathfrak F$, since such a
product is obtained by successive extensions. Subgroup closure therefore puts
$\Aut(s)$ in $\mathfrak F$. The section groupoid belongs to the same class,
proving same-level dependent-product closure. Fibrewise calculation and
Beck--Chevalley give these assertions over arbitrary bounded bases.

If the class is quotient-stable at bound $1$, take the bar groupoid of an
arbitrary finite group $G$. Its terms $G^k$ are finite discrete sets and hence
belong to the class. The quotient $BG$ is $1$-truncated, so quotient stability
implies $G=\pi_1(BG)\in\mathfrak F$. Thus $\mathfrak F$ contains every finite group.

Conversely, suppose $\mathfrak F$ contains all finite groups. Then the class
is exactly that of $\pi$-finite $1$-types. Here $\pi$-finite means finitely truncated with finite $\pi_0$ and finite homotopy groups at every basepoint.
Let $R_\bullet$ be a groupoid of such
objects and suppose $Q=|R_\bullet|$ is $1$-truncated. Its component set is a
quotient of the finite set $\pi_0R_0$, hence is finite. Choose a lift in $R_0$
of a point of $Q$. Its loop space is a homotopy fibre of
$R_1\to R_0\times R_0$. Homotopy fibres of maps between $\pi$-finite $1$-types are
$\pi$-finite, by the homotopy exact sequence. Here the loop space is also
$0$-truncated because $Q$ is a $1$-type, so it is a finite set. Consequently
all isotropy groups of $Q$ are finite and $Q$ is in the class. The same
argument in each fibre proves the relative quotient-stability assertion.
\qed\end{proof}

\Needspace{6\baselineskip}
We apply the finite-group criterion to finite $2$-groups to obtain the example in Corollary~\ref{cor:quotient-separation}.

\begin{proof}[of Corollary~\ref{cor:quotient-separation}]
Finite $2$-groups contain the trivial group and are closed under
subgroups and extensions. Theorem~\ref{thm:finite-group-separation}
therefore gives the closed univalent level. Insert this level in the
bounded-space hierarchy below any fixed sufficiently large level of small
$1$-types. The resulting index set remains directed, and the existing
levels remain exhaustive. Theorem~\ref{thm:generic-tests} therefore gives
an $\infty$-tripos with this additional level.

Each $C_3^k$ is a finite discrete space with trivial fundamental groups,
so every term of the bar groupoid of $C_3$ belongs to the finite-$2$-group
level. Its quotient $BC_3$ is $1$-truncated and has fundamental group
$C_3$, of order $3$. Thus this quotient lies strictly beyond that level
while retaining its truncation bound, and belongs to a larger level of
the exhaustive hierarchy.
\qed\end{proof}

\section{The Conservative Realizability Hierarchy}\label{app:realizability-hierarchy}
This appendix proves the assembly-source lemma (Lemma~\ref{real:ordinary-source}) and uses it to prove the conservative realizability hierarchy (Theorem~\ref{thm:real-hierarchy}). We use the source, size convention, and classes of Section~\ref{sec:realization}. Put
\(\HH=\Sh(\Asm,\Spc)\); the embedding \(\E\into\HH\) is full,
regular, and preserves the effective groupoid quotients used below.
Carrier representatives and projective lifts are selected in the
working external universe. Every assembly map used in the
presentations admits a numerical tracker.

\subsection{The Assembly Source and Its Projective Cover}\label{app:assembly-source}
\begin{proof}[of Lemma~\ref{real:ordinary-source}]
Write $\rho_A(a)\subseteq\NN$ for the nonempty set of realizers of
an element of an assembly $A$.  Fix recursive coding of pairs and
Kleene application $e\cdot r$. For partial computations, $a\downarrow$ means defined, and $a\simeq b$ means both are undefined or have the same value (Kleene equality).
A morphism $u:A\to B$ is a carrier
function for which some index $d$ satisfies
$d\cdot r\mathord\downarrow\in\rho_B(u(a))$ whenever
$r\in\rho_A(a)$.  The image of such a morphism has underlying image
set and realizer set
\[
 \rho_{\operatorname{im}(u)}(b)
       =\bigcup_{u(a)=b}\rho_A(a).
\]
The inclusion in $B$ is tracked by $d$ and the map from $A$ to its
image is tracked by the identity.  These are the usual stable image
factorizations of assemblies.  More explicitly, a morphism $e:E\to B$
is a regular epimorphism precisely when there is an index $c$ such that
for every $b\in B$ and $r\in\rho_B(b)$,
\[
 c\cdot r\mathord\downarrow,\qquad
 \exists a\in E\;\bigl(e(a)=b\ \text{and}\
                          c\cdot r\in\rho_E(a)\bigr).
\]
This criterion is stable under pullback, by pairing an unchanged name
of the other pullback component with the name produced by $c$.
It also gives the coequalizer of the kernel pair: a map constant on
its carrier fibres factors through the underlying surjection, and
its tracker is composed with $c$.  All finite limits and images have
carriers within one of the stated bounds, so the same arguments apply
to this size-stratified category.

For a partitioned assembly $P$, let $\nu_P(p)$ denote the unique
realizer of $p$.  A terminal partitioned assembly is a singleton with name
$0$.  Products use paired names; equalizers restrict the carrier and
retain the original names.  These constructions are also the corresponding
limits in $\Asm$, so $\PAsm$ is closed under finite limits.
Every assembly $A$ has the regular-epimorphism cover
\[
 \widehat A=\{(a,r):r\in\rho_A(a)\}\longrightarrow A,
 \qquad \nu_{\widehat A}(a,r)=r.
\]
Its carrier remains below the same inaccessible bound as $A$.
Given a regular epimorphism $e:E\to B$, a map $u:P\to B$
tracked by $d$, and a cover tracker $c$ as above, fix a
well-ordering of the carrier of $E$, which exists by its stated
cardinal bound. For each $p\in P$, let $a_p$ be the least element
satisfying
\[
 e(a_p)=u(p),\qquad
 c\cdot(d\cdot\nu_P(p))\in\rho_E(a_p).
\]
The cover condition ensures that such an element exists.
The resulting carrier function is tracked by the single composite
program $c\circ d$. It is therefore a lift in $\Asm$, proving
projectivity.  This is the
standard projective presentation of realizability; see, for example,
\cite{FreyPA18} and the assembly constructions
in~\cite{BergMoerdijk08}.

Now fix $Z\xrightarrow gY\xrightarrow fX$ in $\Asm$.
For $x\in X$ put $Y_x=\{y:f(y)=x\}$.  Let the carrier of $W$
consist of pairs $(x,s)$, where $s:Y_x\to Z$ is a carrier section
of $g$ and admits an index $e$ such that
\[
 \forall y\in Y_x\;\forall b\in\rho_Y(y),\qquad
       e\cdot b\mathord\downarrow\in\rho_Z(s(y)).
\]
Give $(x,s)$ all realizers $\langle r,e\rangle$ with
$r\in\rho_X(x)$ and $e$ satisfying the displayed condition.
Thus every tracker of a section contributes a valid name.  Projection to $X$ is tracked by the first projection, and
the evaluation $W\times_XY\to Z$ is tracked by
\[
       (\langle r,e\rangle,b)\longmapsto e\cdot b.
\]
If $h:T\to X$ has tracker $a$, and a section
$s:T\times_XY\to Z$ has tracker $d$, the parameter theorem for
partial recursive functions gives a total recursive operation
$\sigma$ satisfying
\[
       \sigma(d,t)\cdot b\simeq d\cdot\langle t,b\rangle.
\]
For $t\in\rho_T(t_0)$, the index $\sigma(d,t)$ therefore tracks
the fibre section $y\mapsto s(t_0,y)$.  The carrier function
\[
       t_0\longmapsto\bigl(h(t_0),s(t_0,-)\bigr)
\]
is tracked by the single program
$t\mapsto\langle a\cdot t,\sigma(d,t)\rangle$.
Its evaluation is the given section, and the carrier function is
unique.  This proves the dependent-product adjunction, hence also
weak adequacy over the original parameter.

Finally choose one $\kappa_i$ strictly larger than the cardinalities of the carriers of
$X,Y,Z$.  The carrier of $W$ is a subset of the disjoint union, over
$x\in X$, of the sets of functions $Y_x\to Z$.  Strong
inaccessibility makes this union smaller than $\kappa_i$.
Realizer sets are subsets of $\NN$, so the carrier bound is retained.  Thus the construction remains in $\Asm$.
\qed\end{proof}

\begin{lemma}[Comparison of realizability completions]\label{real:completion-comparison}
With the same external size bounds on both sides, the regular embedding $\Asm(K_1)\into\Eff$ induces an equivalence
\[
 \Ex(\Asm(K_1))\simeq\Ex(\Eff).
\]
\end{lemma}
\begin{proof}
To see the displayed equivalence of higher completions directly, a finite-limit-preserving functor from an ordinary category to an exact $\infty$-category lands in its $0$-truncated part. That part is an ordinary exact category. The ordinary exact/regular completion universal property and~\eqref{eq:import-universal} then identify the two higher completions. The composite embedding of assemblies is fully faithful.
\qed\end{proof}

\begin{proposition}[A higher realizability type]\label{real:bc2}
Let \(C_2\) be the numerically named two-element group assembly. Its group nerve has an effective quotient \(BC_2\) in \(\E\), and its canonical point satisfies \(1\times_{BC_2}1\simeq C_2\). Its fundamental group is therefore \(\pi_1(BC_2,*)\cong C_2\).
\end{proposition}
\begin{proof}
For the numerically named two-element group assembly \(C_2\),
effectiveness gives \(1\times_{BC_2}1\simeq C_2\).
Thus the based loop object is the discrete group \(C_2\), so
\(\pi_1(BC_2,*)\cong C_2\), with its two distinct global elements.
\qed\end{proof}

Lemma~\ref{pt:projectives} carries this projective cover into \(\E\),
and every object of \(\E\) has an original partitioned cover.
Theorem~\ref{thm:pt-completion} and Corollary~\ref{pt:source-sc}
therefore give the general section-code transfer for the assembly source.
The argument below equips this completion with uniform size bounds, weak
universe families, and same-level closure.

\subsection{Locality and Closure of Small Projective Atlases}
\begin{lemma}[Locality, monotonicity, and lowering the bound]
\label{real:locality}
Each \(\mathcal S_{\kappa,n}\) is pullback-stable and local for covers
of the codomain.  The classes increase with \(\kappa\) and with \(n\).
If \(f\in\mathcal S_{\kappa,n}\) is \(m\)-truncated, where
\(-1\le m\le n\), then \(f\in\mathcal S_{\kappa,m}\).
\end{lemma}
\begin{proof}
For pullback stability, a projective test of the new codomain is also a
test of the original codomain; apply induction to the diagonal.
For locality, let \(X'\onto X\) be a cover and suppose the pullback of
\(f\) belongs to the stated class.  Every projective \(P\to X\)
lifts through \(X'\): the pullback cover over \(P\) splits.  Thus its
pullback family has the required atlas.  The assertion for the diagonal
follows by induction, and truncation is local for covers.

Increasing \(\kappa\) is immediate.  To start the induction in \(n\),
pull a monomorphism back to a partitioned assembly \(P\).  A subobject
of an assembly in \(\Eff\) is an assembly, and its carrier embeds in
the carrier of \(P\).  Its partitioned cover has at most countably many
names over any element of \(P\).  Since \(\kappa>\aleph_0\), this is
a small atlas.  Its diagonal is monic.  The remaining inclusions follow
by induction on the diagonal.  Finally, the same atlases meet the requirement at the lower truncation bound; apply induction to the diagonal.
For \(m=-1\) the assertion is the definition.
\qed\end{proof}

\begin{proposition}[Closure properties]
\label{real:closures}
For every \(\kappa\) and \(n\ge-1\), the class
\(\mathcal S_{\kappa,n}\) contains identities and is closed under
composition, relative diagonals, and image inclusions.  Its objects in
each slice are closed under finite limits.  If \(R_\bullet\) is an
internal groupoid over \(X\), all \(R_k\to X\) belong to
\(\mathcal S_{\kappa,n}\), and \(|R_\bullet|\to X\) is
\(n\)-truncated, then this quotient map also belongs to
\(\mathcal S_{\kappa,n}\).
\end{proposition}
\begin{proof}
Identities and image inclusions are monomorphisms.  Diagonal closure is
part of the definition and Lemma~\ref{real:locality}.  We prove
composition closure by induction on \(n\); monomorphisms give the base
case.  For \(Z\xrightarrow{g}Y\xrightarrow{f}X\), choose, after a
projective test \(P\to X\), an atlas \(A\onto Y_P\), and then an
atlas \(B\onto Z\times_YA\).  Their composite covers \(Z_P\).
The carrier fibres of \(B\to P\) are unions of fewer than \(\kappa\)
sets of cardinality less than \(\kappa\), hence are small by regularity
of \(\kappa\).  The diagonal of \(fg\) factors as
\[
 Z\xrightarrow{\Delta_g}Z\times_YZ
       \longrightarrow Z\times_XZ,
\]
where the second arrow is a pullback of \(\Delta_f\).  Both arrows
belong to \(\mathcal S_{\kappa,n-1}\), so the induction applies.
Pullbacks, products, and equalizers in a slice are now obtained from
composition, pullback, and diagonal closure.  This gives all finite
limits, including their homotopy data.

We first note that any map \(A\to B\) between small objects over
\(X\) is small: factor it through its graph in \(A\times_XB\).
The graph is a pullback of \(\Delta_{B/X}\), and the projection is a
pullback of \(A\to X\).
Put \(Q=|R_\bullet|\).  The case \(n=-1\) follows from its assumed
truncation.  For \(n\ge0\), an atlas of \(R_0\) after a projective
base change provides an atlas of \(Q\).  On the cover
\(R_0\times_XR_0\onto Q\times_XQ\), the diagonal of \(Q/X\)
pulls back to
\(R_1\to R_0\times_XR_0\).  The graph observation makes this map
\(\mathcal S_{\kappa,n}\)-small.  It is \((n-1)\)-truncated because
\(Q/X\) is \(n\)-truncated.  Lower its bound using
Lemma~\ref{real:locality}, and descend along the cover.  This proves
the diagonal requirement for \(Q\), and hence the assertion.
\qed\end{proof}

\begin{lemma}[Exhaustiveness]
\label{real:exhaustive}
Every arrow of \(\E\) belongs to \(\mathcal S_{\kappa_i,n}\) for
some \(i\) and some finite \(n\ge-1\).
\end{lemma}
\begin{proof}
Every arrow is finitely truncated because \(\E\) is bounded.
Induct on its truncation degree.  For \(f:Y\to X\), choose a
projective cover \(P_0\onto X\) and a projective cover
\(A\onto Y\times_XP_0\).  Their carriers lie below some cardinal
\(\kappa_i\), after increasing \(i\) if necessary.  Every projective
test \(P\to X\) lifts through \(P_0\), so pulling back \(A\to P_0\)
gives an atlas with the same fibre bound, independently of the size of
\(P\).  Apply the induction hypothesis to \(\Delta_f\), and increase
\(i\) once more to include its bound.  The finite truncation degree bounds the number of cardinal increases.
\qed\end{proof}

\subsection{Finite-height Tracked Kan Presentations and Weak Universes}
To construct weak universe families, we now develop tracked codes for bounded Kan presentations. We first establish the realization and projective-resolution results needed for this encoding.

All simplicial realizations in the following argument are first taken
in \(\mathcal H=\Sh(\Asm,\Spc)\), or in its indicated
slice.  The next lemma shows that effective groupoid quotients realize these bounded presentations in \(\E\).  A Kan simplicial assembly with chosen horn fillers is a simplicial assembly
in which each horn map has a specified section that is a morphism of assemblies.  A split simplicial object means
that its degeneracies arise from the usual decomposition into
nondegenerate summands.
All simplicial objects and their finite cotensors in this subsection
are relative to the assembly base under consideration.  Thus a Kan presentation
over \(B\) carries its horn fillers as tracked maps in
\(\Asm_{/B}\).
For a finite simplicial set \(L\), the cotensor of a simplicial assembly \(K_\bullet\) over \(B\) is given by \((K^L)_m=K(\Delta^m\times L)\), where \(K(T)\) is the object of compatible \(T\)-shaped simplices, formed using finite limits in \(\Asm_{/B}\).

\begin{lemma}[Realization of bounded Kan presentations in \(\E\)]
\label{real:kan-admission}
Let \(B\) be an assembly and \(K_\bullet\to B\) a Kan
simplicial assembly with chosen horn fillers.  If its ambient realization \(Q\to B\) is
\(r\)-truncated, for finite \(r\ge-1\), then \(Q\to B\) belongs
to \(\E_{/B}\).
If \(B\) and all \(K_m\) are partitioned and every
\(K_m\to B\) has carrier fibres of size less than \(\kappa\), then
\(Q\to B\) belongs to \(\mathcal S_{\kappa,r}\).
\end{lemma}
\begin{proof}
We induct on the relative truncation degree, allowing an arbitrary
assembly as base at every step.  For \(r=-1\), \(Q\to B\) is the
image of \(K_0\to B\), so is represented in \(\Asm\).
For the induction step, \(K_0\onto Q\) is a cover.  The simplicial
path object, with endpoints restricted to the constant vertex objects,
is
\[
 L_\bullet=K_\bullet^{\Delta^1}
   \times_{K_\bullet^{\partial\Delta^1}}
                  (K_0\times_BK_0).
\]
Its levels are finite limits of levels of \(K\): in degree \(m\)
only the finitely many nondegenerate simplices of
\(\Delta^m\times\Delta^1\) occur.  It is Kan over
\(K_0\times_BK_0\).  Chosen horn fillers for it can be obtained by
successively attaching the finitely many simplices in the relevant
prism and using the chosen fillers of \(K\).  Yoneda sends these
sections to actual horn-filler functions, so the associated simplicial
presheaf is objectwise Kan.  The Kan path-object calculation
identifies the presheaf realization of \(L\) with the homotopy
pullback of the two vertex maps into the presheaf realization of \(K\).
Left exact sheafification therefore identifies its realization in
\(\mathcal H\) with \(K_0\times_QK_0\).  Its relative truncation
degree is \(r-1\).
By induction this identity object belongs to \(\E\).  The entire
\v{C}ech groupoid of \(K_0\onto Q\) consequently belongs to \(\E\),
and exactness gives its quotient there.  The regular fully faithful
embedding identifies that quotient with \(Q\).

For the smallness assertion, finite limits of the indicated carriers
remain smaller than \(\kappa\) in each fibre.  The induction hypothesis gives
\(|L_\bullet|\to K_0\times_BK_0\) in \(\mathcal S_{\kappa,r-1}\).
Lemma~\ref{real:locality} descends this assertion to \(\Delta_{Q/B}\).
The cover \(K_0\onto Q\), and its pullbacks to projective tests of
\(B\), provide the required small atlases.  This proves the recursive
smallness condition.
\qed\end{proof}

\begin{lemma}[Finite projective coproducts in the completion]
\label{real:projective-coproducts}
Every original partitioned assembly is projective also in
\(\mathcal H\).  If \(P=\coprod_{a\in A}P_a\) is a finite
coproduct of partitioned assemblies in \(\Asm\), then
\[
 \Map_{\E}(P,Q)\simeq
       \prod_{a\in A}\Map_{\E}(P_a,Q)
 \qquad(Q\in\E).
\]
Thus this coproduct, including the empty one, is a
coproduct in \(\E\).
\end{lemma}
\begin{proof}
Write \(h\) for the regular Yoneda embedding into \(\mathcal H\).
For a cover \(A\onto B\) in \(\mathcal H\) and a map
\(hP\to B\), pull back to obtain a cover over \(hP\).
The local-surjectivity criterion for an effective epimorphism of
sheaves gives a covering sieve of \(P\) on which the identity
section of \(hP\) lifts.  In the regular topology that sieve
contains an original regular epimorphism \(C\onto P\).
Its lift is a map \(hC\to A\) over \(B\).  Source projectivity
splits \(C\onto P\), giving the required map \(hP\to A\)
and its coherent comparison to \(hP\to B\).
Consequently \(\Map_{\mathcal H}(hP,-)\) preserves covers,
as well as all limits.

The finite coproduct of partitioned assemblies in \(\Asm\) is
partitioned, using disjoint carrier unions and tagged names; its
universal property in assemblies is tracked by case analysis on
the finite tag.  Hence both functors
\[
 \Map_{\mathcal H}(hP,-),\qquad
 \prod_{a\in A}\Map_{\mathcal H}(hP_a,-)
\]
preserve finite limits and covers, with the empty product interpreted
as the terminal-space functor.  They therefore preserve effective
groupoid quotients: the image of a quotient cover is a cover, and
finite-limit preservation identifies its entire \v{C}ech nerve.
The comparison between these functors is an equivalence on original
representables by the coproduct universal property in \(\Asm\) and
full faithfulness.  The full subcategory on which it is an
equivalence is closed under finite limits and effective groupoid
quotients.  It contains \(\E\), which is the closure of the original representables
under finite limits and effective groupoid quotients
in the higher exact-completion construction.  Full faithfulness
now gives the asserted mapping-space equivalence in \(\E\).
\qed\end{proof}

\begin{lemma}[Projective horn fillers]
Let $K_\bullet\to Y$ be an augmented simplicial object in a regular
$\infty$-category, with every augmented matching map
$K_r\to M_r^Y K$ a cover. Then every horn map
$K_m\to K(\Lambda_i^m)$, for $m\ge1$, is a cover. If the levels lie in
a full finite-limit-closed class of projectives, these horn maps admit
sections.
\end{lemma}
\begin{proof}
Work in the slice over $Y$. For a simplicial subcomplex $L$ of a simplex,
write $K_Y(L)$ for the finite limit over its nondegenerate faces in this
slice, and put $K_Y(\varnothing)=Y$. Set $H=K_Y(\Lambda_i^m)$.
Attaching the missing $(m-1)$-face gives a pullback square
\[
\begin{tikzcd}
K_Y(\partial\Delta^m) \arrow[r] \arrow[d] &
K_{m-1} \arrow[d,two heads] \\
H \arrow[r] & M_{m-1}^Y K.
\end{tikzcd}
\]
Thus $K_Y(\partial\Delta^m)\to H$ is a cover. Attaching the
$m$-simplex gives the augmented matching cover
$K_m\to M_m^Y K=K_Y(\partial\Delta^m)$. Their composite is a cover.
For $m=1$, the first attachment adds the missing vertex along the empty
boundary and uses the cover $K_0\to Y$.

The finite face category of the horn is weakly contractible. The slice
forgetful functor creates limits of this shape, so the underlying object
of $H$ is the finite horn limit $K(\Lambda_i^m)$ in the ambient category. This identifies
the composite above with the horn map. Finite-limit closure puts its
codomain in the selected projective class; projectivity therefore
gives a section.
\qed\end{proof}

\begin{lemma}[Projective resolutions with fibrewise cardinal bounds]
\label{real:resolution}
Let \(P\in\PAsm\) and \(Y\to P\) belong to
\(\mathcal S_{\kappa,n}\), \(n\ge0\).  There is a Kan
simplicial object of partitioned assemblies over \(P\) with chosen horn fillers,
whose levels have \(\kappa\)-small carrier fibres and whose \((n+1)\)-coskeleton has
realization equivalent to \(Y\) over \(P\).
Levels through \(n+1\) and chosen horn fillers through dimension
\(n+2\) specify this final presentation.
\end{lemma}
\begin{proof}
Construct a split projective hypercover of \(Y\) in the slice over
\(P\).  At each stage its relative matching object is a finite
homotopy limit of \(Y\) and the already constructed levels.  It is
therefore small by Proposition~\ref{real:closures}, and the atlas
condition gives a projective cover with small carrier fibres over
\(P\).  Add the finite coproduct in \(\Asm\) of the prescribed
degenerate summands.  Lemma~\ref{real:projective-coproducts} makes
this also a coproduct in \(\E\), so their maps to the higher
matching object, together with their coherent augmentations, combine
to one map.  This preserves both projectivity and the fibre bound,
and is the usual inductive construction of a split simplicial
hypercover.
These levels belong to the size-stratified source.
Choose a selected \(\kappa_j\) strictly above both \(\kappa\) and
the cardinality of \(P\).  Each total carrier has size at most
\(\lvert P\rvert\cdot\kappa<\kappa_j\), uniformly at every
stage.  The external countable recursion therefore stays within the single
selected bound \(\kappa_j\).

Apply the projective horn-filler lemma in $\E_{/P}$. Its hypotheses hold
because the augmented matching maps are covers and partitioned assemblies
over $P$ form a finite-limit-closed class of projectives. Every horn map
therefore has a section. Full faithfulness makes this section an original
assembly morphism, hence a tracked map; discreteness of the original
mapping spaces makes its section equation an equality of assembly morphisms. Choose
these sections. External choice in the working universe provides the choices of presentation
levels and fillers required at each stage of the countable recursion. The structure maps of each bounded presentation
admit a finite tuple of numerical trackers.

For a Kan complex, the unit to its \((n+1)\)-coskeleton induces an
isomorphism on homotopy groups through degree \(n\), and that
coskeleton is \(n\)-truncated
\cite{KerodonTruncated}.  Its horn fillers in dimensions at
most \(n+1\) are inherited from the original complex.  In dimension
\(n+2\), a filler is specified by the missing \((n+1)\)-face;
an original filler gives this face.  In every larger dimension
the filler is unique, because
\(\operatorname{sk}_{n+1}\Lambda_i^m
=\operatorname{sk}_{n+1}\Delta^m\) for \(m\ge n+3\).
These operations are composites of finite-limit projections and the
chosen tracked fillers, so they are assembly morphisms.  This proves
that the finite data asserted in the lemma suffice.
The Kan and truncation assertions apply objectwise to our Kan
presheaf with chosen horn fillers and then after left exact sheafification.  A hypercover realization maps
\(\infty\)-connectively to its target.  Consequently its
\(n\)-truncation is \(Y\), and the coskeleton realization is
equivalent to \(Y\).  The finite truncation makes this identification available in the
ambient topos. The final bounded realization belongs to
\(\E\) by Lemma~\ref{real:kan-admission}.
\qed\end{proof}

\begin{theorem}[Tracked weak universes]
\label{real:universes}
For every \(\kappa=\kappa_i\) and finite \(n\ge0\), there is a
family \(p_{\kappa,n}:E_{\kappa,n}\to V_{\kappa,n}\) in \(\E\)
such that
\[
              \Loc(p_{\kappa,n})=\mathcal S_{\kappa,n}.
\]
Its base is a partitioned assembly of cardinality at most \(\kappa\),
and its presentation levels have the same total cardinal bound.
Thus the construction is available at \(\kappa_{i+1}\).
\end{theorem}
\begin{proof}
\emph{Code construction.} A code consists of the following finite data: ordinal carriers
\(\alpha_0,\ldots,\alpha_{n+1}<\kappa\), singleton name functions
\(\alpha_k\to\NN\), face and degeneracy carrier functions satisfying
the truncated simplicial identities, and horn-filler sections through
dimension \(n+2\) for their \((n+1)\)-coskeleton.  The sections in dimension
\(n+2\) can be specified by the missing-face operations with values
in \(\alpha_{n+1}\).  The code also includes a fixed finite
tuple \(q\) of numerical trackers for all these operations.  The carrier consists
of the data equipped with valid trackers, giving a set-theoretic
definition of an assembly whose operations are uniformly tracked.

\emph{Universal operations.} Give the code its singleton name \(q\).  Give its universal
\(k\)-simplex \((v,a)\) the singleton name
\(\langle q,\nu_{v,k}(a)\rangle\).  Universal face, degeneracy,
and filler maps are tracked uniformly by extracting the appropriate
component of \(q\) and applying it to the input name.  Distinct codes
may share the name \(q\); this allows the assembly to retain all
carrier presentations compatible with the same tracked operations.  We obtain a simplicial family of partitioned
assemblies after taking the coskeleton.  Its levels over a fixed code
are formed by finite limits of its finitely many carriers, so all
fibre cardinalities remain below \(\kappa\).  Its realization is
\(n\)-truncated and belongs to \(\mathcal S_{\kappa,n}\), by
Lemma~\ref{real:kan-admission}.  Call it \(p_{\kappa,n}\).

\emph{Local presentations.} Conversely, pull \(f\in\mathcal S_{\kappa,n}\) back to a projective
cover \(P\) of its base and use Lemma~\ref{real:resolution}.
In each carrier fibre relabel the finitely many sets by ordinals
below \(\kappa\), retaining their original numerical names.  All
structure operations of the family have globally chosen trackers.

\emph{Tracker normalization.} Normalize the finite-limit encodings as follows.  If an operation
expects a base name, recover it using the tracker of the projection
of one of its vertex inputs to \(P\); every horn under consideration
has a vertex.  Include these finitely many projection trackers in
the data used to compile the operation trackers.  The resulting
fibre operations read the names of their simplex inputs and
have a common tracker tuple for all \(p\in P\).
Consequently, the classifying map \(P\to V_{\kappa,n}\) is tracked.

\emph{Recovery of the family.} Because relabelling the carriers retains the original numerical names,
the resulting isomorphisms of assemblies are tracked.  More explicitly,
the forward isomorphism sends an original simplex name to its base
name, obtained by the tracker of its projection to \(P\), together
with the fixed tuple and the unchanged simplex name.  The inverse
extracts the latter name.  These algorithms track inverse carrier
bijections.  Pulling back the
universal simplicial family recovers the given presentation, with
these isomorphisms compatible with all its structure maps.  It
therefore recovers \(f_P\) on realization.  Locality proves the
claimed equality of classes.

\emph{Size bounds.} Strong inaccessibility bounds the collection of all finite choices
of carriers, functions, names, and tracker tuples by \(\kappa\).
The universal total carriers also have size at most \(\kappa\).
They therefore define assemblies at the next selected
cardinal, while their classified fibre presentations retain the
strict bound \(<\kappa\).
\qed\end{proof}

\subsection{Small Section Codes and Dependent Products}

\begin{lemma}[Small section-code families over projective bases]
\label{real:section-codes}
Suppose \(P\in\PAsm\), \(Y\to P\) belongs to
\(\mathcal S_{\kappa,N}\), and \(Z\to Y\) belongs to
\(\mathcal S_{\kappa,n}\), where \(-1\le n\le N\) and
\(N\ge0\).  There is a partitioned assembly \(W\to P\), with
carrier fibres of cardinality less than \(\kappa\), and an evaluation
\(W\times_PY\to Z\) over \(Y\), which represents every section
locally on its parameter object.  Equivalently, in \(\mathcal H\)
the induced map \(W\to\Pi_{Y/P}Z\) is a cover.
\end{lemma}
\begin{proof}
Set \(m=n+1\).  Construct a partial projective semisimplicial
hypercover \(U_{\le m}\to Y\) over \(P\), and a relative
Reedy-covering resolution
\[
 V_k\onto D_k\times_{M_kD}M_kV,
 \qquad D_k=Z\times_YU_k,
 \qquad 0\le k\le m,
\]
as in Lemma~\ref{pt:transfer}.  Here \(M_0D=M_0V=P\).
All matching objects are finite homotopy limits over \(P\).
By Proposition~\ref{real:closures} they belong to
\(\mathcal S_{\kappa,N}\).  The defining atlas condition, applied
to the identity map \(\id_P:P\to P\), therefore permits every
\(U_k\) and \(V_k\) to be partitioned with carrier fibres smaller
than \(\kappa\).  Indeed, a map between partitioned assemblies
with such fibres belongs to \(\mathcal S_{\kappa,0}\): its
pullbacks to projective tests are their own small atlases, and its
diagonal is monic.  Thus the induction retains the cardinal bound \(\kappa\) throughout the indicated finite range.

For \(p\in P\), let the carrier \(W_p\) consist of strict
semisimplicial sections
\[
 a_{\le m}:U_{\le m,p}\longrightarrow V_{\le m,p}
 \quad\text{over }U_{\le m,p},
\]
together with a tuple \((e_0,\ldots,e_m)\) of numerical trackers
for the component carrier functions.  Give this code the singleton
name \(\langle\nu_P(p),e_0,\ldots,e_m\rangle\).
The face equations define the carrier set externally, and the tracker
tuple gives its evaluations uniformly.
Evaluation in degree \(k\) applies \(e_k\) to the simplex name,
so defines a morphism of assemblies.  The face equations make these
evaluations a strict semisimplicial map.

For each \(p\), the set \(W_p\) is a subset of
\[
 \NN^{\,m+1}\times
       \prod_{0\le k\le m}
              \operatorname{Set}(U_{k,p},V_{k,p}).
\]
Every base and exponent in this finite product has cardinality less
than \(\kappa\).  Strong inaccessibility gives
\(\lvert W_p\rvert<\kappa\).  The total carrier of \(W\) is
also a permitted carrier: if the carrier of \(P\) lies below
\(\kappa_j\), choose a selected cardinal strictly above both
\(\kappa\) and \(\kappa_j\), and use regularity of that cardinal.
Projection of names tracks \(W\to P\).

Composing the degreewise evaluation with \(V_k\to D_k\to Z\)
gives a coherent map
\(W\times_P|\operatorname{sk}_mU|\to Z\) over \(Y\).
By Lemma~\ref{pt:finite-descent}, the partial hypercover map is
left orthogonal, after every base change, to \(n\)-truncated maps.
There is consequently a contractible space of extensions to
\[
                   W\times_PY\longrightarrow Z.
\]
Finite descent thus constructs the required evaluation directly from
the finite diagram codes.

To prove adequacy, cover an arbitrary parameter \(T\to P\) in
\(\E\) by a partitioned assembly \(T_0\).  A given section over \(T_0\)
determines a map \(T_0\times_PU_{\le m}\to D_{\le m}\).
Lift it through the relative matching covers one degree at a time.
The lifting domains \(T_0\times_PU_k\) are partitioned and hence
projective.  The chosen lifts include their coherent comparison to
the given map; their face equations in \(V\) are strict because
mapping spaces between embedded assemblies are discrete.
Choose a tracker \(c_k\) for each of these finitely many lifts.
Using the fixed computable pairing conventions, the parameter theorem
for partial recursive functions gives an index
\(e_k(t)=s(c_k,\nu_{T_0}(t))\) for its restriction to the fibre
\(U_{k,p}\).  Here \(s\) is one fixed computable specialization
function acting on the numerical trackers of the carrier functions.  Pairing the output of the tracker of \(T_0\to P\) with the results
of these finitely many specializations gives a single tracker of \(T_0\to W\).
The two realized evaluations agree, with all homotopies, because
Lemma~\ref{pt:finite-descent} makes restriction to the partial
hypercover an equivalence of section spaces.  Thus every section is
represented after the parameter cover \(T_0\onto T\), as required.
In particular this holds for all original representable tests, which
is the local-surjectivity criterion for \(W\to\Pi_{Y/P}Z\) to
be a cover in \(\mathcal H\).
\qed\end{proof}

\begin{theorem}[Same-level dependent products]
\label{real:pi}
If \(f:Y\to X\) and \(g:Z\to Y\) belong to
\(\mathcal S_{\kappa,N}\), then the ambient dependent product
\(\Pi_f g\) is represented in \(\E\) and belongs to
\(\mathcal S_{\kappa,N}\).  More precisely, if \(g\) is
\(n\)-truncated, the resulting map belongs to
\(\mathcal S_{\kappa,n}\).  These dependent products satisfy
Beck--Chevalley.
\end{theorem}
\begin{proof}
If \(N=-1\), replace it by \(0\) for the argument below.  By
Lemma~\ref{real:locality}, regard \(g\) as a member of
\(\mathcal S_{\kappa,n}\).  First suppose \(X=P\) is partitioned.  Write
\(S=\Pi_f g\) in the ambient sheaf category, and use
Lemma~\ref{real:section-codes} to obtain a small projective cover
\(W\onto S\).  Induct on \(n\).  For \(n=-1\), \(S\to P\)
is the represented image of \(W\to P\), and is monic.

For \(n\ge0\), put \(B=W\times_PW\).  The two evaluations
determine an identity family \(I\to Y\times_PB\), a pullback of
\(\Delta_g\).  It belongs to \(\mathcal S_{\kappa,n-1}\).
The induction hypothesis with the fixed source bound \(N\) represents
\[
 W\times_SW\simeq
       \Pi_{Y\times_PB\,/\,B}I
\]
by a map in \(\mathcal S_{\kappa,n-1}\) over the partitioned
assembly \(B\).  Since \(B\to P\) has small carrier fibres, every
term of the \v{C}ech groupoid of \(W\onto S\) is small over \(P\).
Its effective quotient lies in \(\E\), and the ambient embedding
identifies it with \(S\).  It is \(n\)-truncated, so bounded quotient
closure in Proposition~\ref{real:closures} makes it
\(\mathcal S_{\kappa,n}\)-small.

For general \(X\), choose a projective cover \(P\onto X\).
The preceding argument represents \(S_P\to P\).  Its descent
groupoid is already represented: its degree-one object is
\[
 S_P\times_SS_P\simeq S_P\times_XP,
\]
and all higher terms are analogous finite limits.  Exactness gives
its quotient in \(\E\), which is \(S\) in the ambient sheaf
category.  Locality then proves that \(S\to X\) belongs to
\(\mathcal S_{\kappa,n}\).
Full faithfulness restricts its ambient adjunction to \(\E\), and
the ambient Beck--Chevalley equivalences restrict with it.
\qed\end{proof}

\subsection{The Propositional Universe and Original-Base Codes}
We now give the propositional weak universe and express same-level section coding over original assembly bases. Together with the preceding constructions, these verify conditions (H1)--(H3).

\begin{lemma}[A propositional weak universe over a partitioned assembly]
\label{real:propositional}
The subobject classifier \(1\into\Omega\) of the size-stratified
effective topos is a subobject classifier in \(\E\).
There is an original partitioned assembly \(V_{\mathrm{prop}}\)
and an original monomorphism
\(p_{\mathrm{prop}}:P_{\mathrm{prop}}\into V_{\mathrm{prop}}\)
whose local class is all monomorphisms of \(\E\).
\end{lemma}
\begin{proof}
For a subobject \(A\into X\) in \(\E\), choose an original
partitioned cover \(P\onto X\).  Its pullback is a subobject of
an original assembly, hence is original by
Theorem~\ref{thm:pt-completion}(ii).  It has a characteristic map
\(P\to\Omega\) in \(\Eff\).  On the overlap \(P\times_XP\),
the two pulled-back characteristic maps agree: cover that overlap by
original partitioned assemblies and use the subobject-classifier
property of \(\Eff\) there.  Since \(\Omega\) is \(0\)-truncated, these
agreements give the uniquely coherent descent data on the entire \v{C}ech
nerve.  Effectiveness descends the map to \(X\to\Omega\), and
pullback of truth recovers \(A\into X\), as can be checked after
\(P\onto X\).  The same local argument proves uniqueness and
identifies the full mapping space with the discrete space of subobjects.
Thus the truth family of \(\Eff\) is a subobject classifier in \(\E\).

Choose an original partitioned cover
\(V_{\mathrm{prop}}\onto\Omega\) and pull back truth.
The total object of the resulting monomorphism is original by subobject
conservativity.  Conversely, for any characteristic map
\(X\to\Omega\), the cover
\(X\times_\Omega V_{\mathrm{prop}}\onto X\) presents the
corresponding subobject as a pullback of \(p_{\mathrm{prop}}\).
Its local class is therefore exactly all monomorphisms.
\qed\end{proof}

\begin{corollary}[Source section codes within each universe level]
\label{real:source-hierarchy}
Let \(X\) be an original assembly and let
\(Z\xrightarrow gY\xrightarrow fX\) be a pair in
\(\mathcal S_{\kappa,N}\).  There is an original assembly map
\(W\to X\) in \(\mathcal S_{\kappa,N}\), with an evaluation
\(W\times_XY\to Z\), which represents every section with an
original parameter after an original regular-epimorphism cover of
that parameter.
Consequently the presentations \(p_{\kappa_i,n}\), together with
\(p_{\mathrm{prop}}\), verify all three conditions
(H1)--(H3) of Definition~\ref{def:lc-universes}, as required by
Theorem~\ref{thm:controlled-source}: exhaustive increasing local classes,
bounded quotient closure, and section codes in the source category belonging to the same universe level.
\end{corollary}
\begin{proof}
First let \(N\ge0\).  Use the standard partitioned cover
\[
 P=\{(x,r):r\in\rho_X(x)\}\onto X.
\]
Its carrier fibres are at most countable.  Pull the pair back to
\(P\) and apply Lemma~\ref{real:section-codes}, taking \(n=N\).
It gives a partitioned assembly \(W\to P\), with each carrier
fibre smaller than \(\kappa\), and the required evaluation over
\(P\).  The carrier fibre of \(W\to X\) at \(x\) is a union,
indexed by \(\rho_X(x)\subseteq\NN\), of sets smaller than
\(\kappa\).  It is therefore smaller than \(\kappa\).
For any partitioned test \(Q\to X\), the pullback
\(W\times_XQ\) in \(\Asm\) is partitioned: its carrier is the subset of
\(W\times Q\) on which the two maps to \(X\) agree, with paired
names.  This pullback is its own small atlas.  Since an original map
is \(0\)-truncated and has monic relative diagonal,
\(W\to X\) belongs to \(\mathcal S_{\kappa,0}\), hence to
\(\mathcal S_{\kappa,N}\).

For a section with original parameter \(T\to X\), choose an
original partitioned cover \(T'\onto T\times_XP\).
The composite \(T'\onto T\) is an original regular epimorphism.
The projective-parameter argument in
Lemma~\ref{real:section-codes} gives a code \(T'\to W\)
over \(P\), on the same parameter \(T'\).  Its evaluation
is the given section restricted to \(T'\), including its coherent
comparison.  This proves the assertion for \(N\ge0\).

If \(N=-1\), the dependent product \(\Pi_f g\to X\) is monic.
Theorem~\ref{thm:pt-completion}(ii) identifies it with an original subobject
\(W\into X\).  Its actual evaluation gives a section code over each parameter itself, and \(W\to X\) belongs to the propositional
level.

The weak universes of Theorem~\ref{real:universes} and
Lemma~\ref{real:propositional} have original bases. By
Lemma~\ref{lem:local-recoding}, their classes defined using original test
contexts and original covers agree with their local classes. Thus these are
\(\Loc(p_{\kappa_i,n})=\mathcal S_{\kappa_i,n}\) and the class of all
monomorphisms, respectively. The generic families have
finite-atlas presentations: their total objects lie in \(\E\),
and finite truncation and original atlases give the required presentations
by induction on truncation degree.
Lemma~\ref{real:exhaustive} and monotonicity give (H1),
Proposition~\ref{real:closures} gives (H2), and the same-level codes
just constructed give (H3).
\qed\end{proof}

\subsection{Univalent Classification, Size, and Logical Interpretations}\label{app:realization-classification}
The preceding corollary verifies the hypotheses of
Corollary~\ref{cor:lc-universes}.  In particular, the weak universe
families have univalent completions classifying their full
local classes.  The next proposition extends the size bounds from the original code
objects to their completed classifiers.
\begin{proposition}[Size of completed universes and cumulativity]
\label{real:cumulative}
Write
\(\widetilde U_{\kappa_i,n}\to U_{\kappa_i,n}\)
for the univalent completion of \(p_{\kappa_i,n}\), where
\(n\ge0\).  Its class is \(\mathcal S_{\kappa_i,n}\), and both
\(U_{\kappa_i,n}\) and \(\widetilde U_{\kappa_i,n}\), considered
over the terminal object, belong to
\(\mathcal S_{\kappa_{i+1},n+1}\).
For \(i\le j\) and \(n\le m\), the inclusions of classes induce
classifying maps
\[
              U_{\kappa_i,n}\longrightarrow U_{\kappa_j,m}
\]
and pullback squares of their universal families.  These maps and
squares admit coherent identity and composition equivalences, with
a contractible space of choices.
\end{proposition}
\begin{proof}
Put \(\kappa=\kappa_i\), \(\lambda=\kappa_{i+1}\), and
\(V=V_{\kappa,n}\).  By Theorem~\ref{real:universes} the carrier
of \(V\) has size at most \(\kappa<\lambda\).  Thus
\(V\to1\) belongs to \(\mathcal S_{\lambda,0}\).
Over \(V\times V\), both pulled-back weak-universe families belong
to \(\mathcal S_{\kappa,n}\).  Theorem~\ref{real:pi} places their
fibrewise function object in the same class.  The equivalence object
is its subobject cut out by the proposition of being an equivalence,
so Proposition~\ref{real:closures} places that object in
\(\mathcal S_{\kappa,n}\) as well.

After composition with \(V\times V\to1\), all terms of the
equivalence groupoid therefore belong to
\(\mathcal S_{\lambda,n}\).  Its quotient
\(U_{\kappa,n}\) is \((n+1)\)-truncated by
Corollary~\ref{cor:lc-universes}.  Monotonicity followed by bounded
quotient closure puts it in \(\mathcal S_{\lambda,n+1}\).
The completed universal family pulls back to \(p_{\kappa,n}\)
along the cover \(V\onto U_{\kappa,n}\); locality puts this
family itself in \(\mathcal S_{\kappa,n}\).  Composition with the
base map whose size bound was established above gives
\(\widetilde U_{\kappa,n}\to1\) in
\(\mathcal S_{\lambda,n+1}\).

For cumulativity, the universal family at \((i,n)\) lies in the
class classified at \((j,m)\), by
Lemma~\ref{real:locality}.  Its classifying map and the equivalence
with the corresponding pullback form a contractible space, since
the target family is univalent and classifies the entire class.
The same uniqueness identifies the identity map and all composites
with the chosen classifying maps, including every higher coherence.
\qed\end{proof}

We can now assemble the preceding constructions into the cumulative realizability hierarchy of Theorem~\ref{thm:real-hierarchy}.

\begin{proof}[of Theorem~\ref{thm:real-hierarchy}]
Lemma~\ref{real:locality}, Proposition~\ref{real:closures}, and
Lemma~\ref{real:exhaustive} give locality, the required closure
properties, and exhaustiveness.  Theorem~\ref{real:pi} gives same-level
dependent products with the sharper target truncation bound and
Beck--Chevalley.  Their actual evaluation maps give the weak section
families of Definition~\ref{def:classes} over every base.
Theorem~\ref{real:universes} and Lemma~\ref{real:propositional}
give the generic families, with original code bases, and
Corollary~\ref{real:source-hierarchy} verifies all three
conditions (H1)--(H3), including section codes in the source category belonging to the same universe level.
Corollary~\ref{cor:lc-universes} and
Proposition~\ref{real:cumulative} give the classifying equivalences,
cumulative maps, and next-level bounds.

The source codomain doctrine \(\Cod(\Asm)\) is regular: its sums
are composition in slices, with the cartesian Beck--Chevalley structure
of Proposition~\ref{prop:moens}. Theorem~\ref{thm:pt-completion} gives
\(\mathrm{CF}_{<\infty}(\Asm)\). With (H1)--(H3) already established,
Corollary~\ref{cor:lc-universes} verifies the
$\infty$-tripos axioms of Definition~\ref{def:tripos} for this source,
and its construction is
\(\TT\Cod(\Asm)=\Ex(\Asm)\simeq\E\), with precisely the local
classes and classifiers already constructed. Since \(\E\) is exact,
its codomain doctrine is also an $\infty$-tripos; its actual products give
the section codes and its univalent families give the weak universes.
The exact-completion comparison gives
\(\E_{\le0}\simeq\Eff\), and full faithfulness identifies all
maps between original objects.  Theorem~\ref{thm:pt-completion}(ii)
shows that the unit identifies the subobject lattices of original objects
with those of their images and preserves existing dependent products.
For dependent-product-preserving interpretations, Theorem~\ref{thm:pt-morphisms} applies
with the replete closure of \(\PAsm\) as its source projective
class and with the specified target projective covering class.
For interpretations satisfying
\(\Loc(\overline Fp_i)=\M_i^{\D}\) for each explicit weak universe presentation,
including the propositional one, restriction gives the stated equivalence
of logical interpretation spaces.
\qed\end{proof}

\section{Classifier Separation and Computability}\label{app:rice}
This appendix gives the complete proofs of Theorem~\ref{thm:rice} and Corollary~\ref{cor:two-classifiers}, including the construction, representation, and uniformity arguments. Throughout,
\[
 \E=\Ex(\Asm(K_1))\simeq\Ex(\Eff),
 \qquad j:\Eff\into\E.
\]
The equivalence of completions is Lemma~\ref{real:completion-comparison}. The identification of the $0$-truncated part and the assembly covers follow from~\eqref{eq:import-shadow}, the classical exact/regular completion of assemblies, and the atlas property of higher exact completion~\cite{vanOosten08,Stefanich23}.

\subsection{Reset Families and Univalent Classification}\label{app:reset}\label{app:em}
An assembly consists of a carrier set and a nonempty set of numerical realizers for each element. A map is a carrier function admitting a partial recursive tracker that sends every realizer of an input to a realizer of its output. Equality of maps is determined by their carrier functions; each map may have several trackers. A modest assembly has disjoint sets of realizers for distinct elements~\cite{vanOosten08}.

For part~(\ref{rice:fibres}) of Theorem~\ref{thm:rice}, we construct the tracked reset family and its iterated deloopings, then identify their fibres at total and non-total indices.

\begin{proof}[of Theorem~\ref{thm:rice}(\ref{rice:fibres})]
For the fixed acceptable numbering $(\varphi_e)$, choose a computable bounded-halting relation and define
\begin{equation}\label{eq:progress}
 \ell_e(s)=\max\{m\le s:\text{every }k<m\text{ halts under }\varphi_e
                    \text{ within }s\text{ steps}\}.
\end{equation}
The number $\ell_e(s)$ can be computed by running finitely many bounded simulations. It is nondecreasing: if the first $m$ inputs have halted within $s$ steps, the same is true within $s+1$ steps and $m\le s+1$.

Set $a_{e,s}:\Ftwo\to\Ftwo$ to zero when $\ell_e(s+1)>\ell_e(s)$ and to the identity otherwise. Write $a_e[s,t)$ for the composite of the transitions $a_{e,r}$ with $s\le r<t$. If $e\in\Tot$, each finite initial segment eventually halts, so $\ell_e(s)$ is unbounded. A nondecreasing integer sequence is unbounded exactly when it has infinitely many strict increases. Thus the transitions $a_{e,s}$ contain infinitely many zero maps. If $e\notin\Tot$, let $k_0$ be the least input which does not halt. For sufficiently large $s$ all inputs below $k_0$ have halted, and $\ell_e(s)=k_0$ thereafter. Hence there are only finitely many zero transitions. All individual transitions and all finite composites $a_e[s,t)$ are recursive uniformly in their arguments.

Define
\begin{equation}\label{eq:reset-relation}
 (s,b)\sim_e(t,d)\ \Longleftrightarrow\
 \exists v\ge s,t:\ a_e[s,v)(b)=a_e[t,v)(d).
\end{equation}
This relation is reflexive and symmetric. For transitivity, transport two witnessing equalities to a common later stage and compose the transition maps. It is computably enumerable by searching over the witness $v$. Let $M_e=(\NN\times\Ftwo)/{\sim_e}$. Define addition by
\begin{equation}\label{eq:addition}
 [s,b]+[t,d]
 =\bigl[m,a_e[s,m)(b)+a_e[t,m)(d)\bigr],
 \qquad m=\max(s,t).
\end{equation}
Transporting representatives to a common stage shows that this operation is well defined. At that stage it is addition in $\Ftwo$, so the abelian-group axioms hold in the quotient. The zero is $[0,0]$ and each element is its own inverse.

Give a class $x$ the realizer set
\[
 \rho_{M_e}(x)=\{\langle s,b\rangle:x=[s,b],\ b\in\{0,1\}\}.
\]
This is nonempty, and distinct classes have disjoint realizer sets. A tracker for addition decodes the two input pairs, computes the maximum stage and the finite transitions, and returns the code in~\eqref{eq:addition}. Zero is realized by a constant program and inversion by the identity. These trackers are uniform in $e$.

For the total family over $N$, use carrier $\{(e,x):x\in M_e\}$ and names $\langle e,s,b\rangle$ with $x=[s,b]$. Projection to $N$ is recursively tracked. Pairing and projection operations turn the just-described algorithms into trackers for the relative group operations. This proves that $\mathcal M\to N$ is an abelian group object in the assembly slice, and hence in $\Eff_{/N}$.

If $e\in\Tot$, every representative $[s,b]$ is transported through a later zero map and equals $[0,0]$. The group is the zero group, and its unique maps to and from the terminal zero-group assembly are tracked. If $e\notin\Tot$, choose $S$ after the last zero transition. The maps
\[
 \Ftwo\longrightarrow M_e,\quad b\longmapsto[S,b],
 \qquad
 M_e\longrightarrow\Ftwo,\quad
 [s,b]\longmapsto
 \begin{cases}a_e[s,S)(b),&s<S,\\b,&s\ge S\end{cases}
\]
are well-defined homomorphisms inverse to one another. The second formula is invariant under~\eqref{eq:reset-relation}, because all transport after $S$ is the identity. For fixed $e,S$, both formulas are recursive on names. This proves the fibre isomorphisms stated in Section~\ref{sec:application} as tracked isomorphisms in realizability for each fixed $e$ and its chosen stage $S$.

With the tracked coefficient family in place, we construct its deloopings and the equivalence objects required for univalent classification.
For an abelian group object $A$ in a slice of $\Eff$, its commutative group structure in that slice determines a coherent commutative group object in the corresponding slice of $\E$. Its bar nerve is an internal Segal groupoid. Exactness gives its quotient $BA$, with $\Omega BA\simeq A$ by effectiveness. Finite products commute with this realization in the ambient sheaf $\infty$-topos and the relevant products and quotients are preserved by the embedding of $\E$. Thus $BA$ again has a coherent commutative group structure, and the process iterates.

Define $K(A,0)=A$ and $K(A,n+1)=BK(A,n)$. In every slice these objects are pointed, $(n-1)$-connected, and $n$-truncated for $n\ge1$, with $\pi_nK(A,n)=A$ and all other positive homotopy groups zero. The loop calculation proves this inductively; truncation can be checked on the diagonal after the effective cover by the basepoint. All constructions commute with pullback because the relevant quotients are pullback-stable. Applying them to $\mathcal M\to N$ constructs $p_n:\codes_n\to N$ and proves Theorem~\ref{thm:rice}(\ref{rice:fibres}).
\qed\end{proof}

\begin{lemma}[Parameterized Eilenberg-Mac Lane equivalences]\label{lem:em-eq}
For $n\ge2$, a context $S\in\Eff$, and abelian groups $A,B\in\Ab(\Eff_{/S})$, the object $K_S(B,n)\times_S\underline{\Iso}_{\Ab/S}(A,B)$ represents fibrewise equivalences between $K_S(A,n)$ and $K_S(B,n)$ on all test objects of $\E_{/S}$:
\begin{equation}\label{eq:em-equivalence}
 \underline{\Eq}_S(K_S(A,n),K_S(B,n))
 \simeq K_S(B,n)\times_S
 \underline{\Iso}_{\Ab/S}(A,B),\qquad n\ge2.
\end{equation}
\end{lemma}
\begin{proof}
Use the ambient $\HH=\Sh(\Eff,\Spc)$. In each slice of an $\infty$-topos, the homotopy-group functor identifies pointed degree-$n$ Eilenberg-Mac Lane objects with abelian group objects in the full subcategory of $0$-truncated objects for $n\ge2$~\cite{Lurie09}. This is an equivalence of $\infty$-categories with an ordinary category: in particular, its pointed mapping spaces are discrete. Applied after an arbitrary base change $T\to jS$, it identifies pointed equivalences with isomorphisms of the pulled-back abelian groups.

Every unpointed equivalence $f$ has a value $f(0)$ in the target group type. Translation by $-f(0)$ produces a pointed equivalence. Conversely, a target point $b$ and a pointed equivalence produce an unpointed equivalence by translation by $b$. The coherent commutative group laws make these operations mutually inverse on the corresponding spaces of data, naturally in $T$. Thus ambient fibrewise equivalences are the product of $K_S(B,n)$ and the fibrewise group-isomorphism object.

The latter object is constructed in $\Eff_{/S}$ using fibrewise function objects, equations for homomorphisms, and pairs of inverse homomorphisms. Here the fibres are sets, so inverse data are uniquely determined. Lemma~\ref{lem:old-pi} identifies this object of $\Eff_{/S}$ with its ambient fibrewise isomorphism object, including after arbitrary higher base change. The first factor is the finite iterated bar quotient already constructed in $\E_{/S}$. Hence their product lies in $\E$. Fullness of its embedding restricts the ambient representing property to all tests in $\E_{/S}$, proving the lemma.
\qed\end{proof}

We now use these equivalence objects to construct the univalent completion of $p_n$.

\begin{proof}[of the univalent-completion assertion in Theorem~\ref{thm:rice}]
Apply the lemma with $S=N\times N$ and the two pullbacks of $\mathcal M$. It represents the degree-one equivalence object of $p_n$. The groupoid of fibre equivalences is coherent by composition of equivalences in the ambient slices; its levels are finite fibre products of the represented objects, and hence lie in $\E$. Its tautological action on $\codes_n$ is represented for the same reason. Form the two effective quotients as in~\eqref{eq:equiv-groupoid}. Remark~\ref{rem:fragment} applies and yields
\[
 \begin{tikzcd}
 \codes_n\arrow[r]\arrow[d,"p_n"']&\widetilde U_n\arrow[d,"u_n"]\\
 N\arrow[r,"c_n"']&U_n,
 \end{tikzcd}
\]
with $c_n$ a cover, the square cartesian, and $u_n$ univalent. The class it classifies consists exactly of local pullbacks of $p_n$. Since the fibres of $u_n$ are $n$-truncated, the description of the diagonal of $U_n$ by fibrewise equivalences makes $U_n$ $(n+1)$-truncated.
\qed\end{proof}

\subsection{Identity Complexity and Global Homotopy}\label{app:projective}
We next establish the arithmetic and homotopical properties of the classifier. The following projectivity result gives the global-section calculation and the local separation argument.
\begin{lemma}\label{lem:projectives}
The objects $1$ and $N$ are projective with respect to covers in $\E$. The functor $\Gamma:\E\to\Spc$ is regular and preserves internal groupoid quotients.
\end{lemma}
\begin{proof}
The objects $1$ and $N$ are partitioned assemblies. They are regular projectives in $\Asm$ by Lemma~\ref{real:ordinary-source} and remain projective in $\E$ by Lemma~\ref{pt:projectives}.

Now let $Y\onto X$ be a cover in $\E$. Every global point of $X$ pulls it back to a cover over $1$, which splits by the first assertion. Therefore $\Gamma Y\to\Gamma X$ is surjective on connected components, equivalently an effective epimorphism of spaces. Global sections preserve finite limits, so $\Gamma$ is regular. An internal groupoid is the \v{C}ech nerve of its quotient cover. Applying $\Gamma$ preserves both that nerve and the cover; therefore it preserves its realization. This proves preservation of the effective groupoid quotients used in the classifier calculation.
\qed\end{proof}

We turn next to parts~(\ref{rice:complete}) and~(\ref{rice:global}) of Theorem~\ref{thm:rice}, relating numerical identity to totality and computing the classifier's global homotopy type and internal truncation height.

\begin{proof}[of Theorem~\ref{thm:rice}(\ref{rice:complete}) and~(\ref{rice:global})]
Univalence identifies the numerical identity problem with totality, as follows.
Choose an index $e_*$ of a total function. Its fibre is contractible by the reset calculation. Set $u_*=c_n(e_*)$. Univalence gives
\[
 \operatorname{Id}_{U_n}(c_n(e),u_*)
 \simeq\underline{\Eq}(X_{n,e},1).
\]
A global point on the right exists exactly when $X_{n,e}$ is contractible. If $e$ is total, the fibre is $1$ and such a point exists. Otherwise its $n$th homotopy-group object is the nonzero group $\Ftwo$ in $\Eff$, so it is not equivalent to $1$. Thus the set in~\eqref{eq:pi2} is precisely $\Tot$.

The set $\Tot$ is in $\Pi^0_2$ because membership has the form $\forall k\,\exists s\,T(e,k,s)$ for a recursive bounded-halting predicate. Let any $\Pi^0_2$ set be written $A=\{x:\forall k\,\exists s\,R(x,k,s)\}$ with $R$ recursive. Uniformly compile a program $e(x)$ that on input $k$ searches for a witness $s$ and halts when it finds one. The acceptable numbering allows this compilation effectively. Then
\[
 x\in A\ \Longleftrightarrow\ e(x)\in\Tot
 \ \Longleftrightarrow\
 \Gamma\operatorname{Id}_{U_n}(c_n(e(x)),u_*)\ne\varnothing.
\]
This proves many-one completeness and Theorem~\ref{thm:rice}(\ref{rice:complete}).

In particular $\Tot$ is not recursive. One direct verification reduces the halting problem to it: given a program and an input, form a program that first simulates that computation and, if it halts, returns zero on every argument. It is total exactly when the original computation halts. The halting problem is undecidable by the diagonal argument. This nonrecursiveness yields the semantic separation proved below.

We now compute the global homotopy type and the internal truncation height of the classifier.
Let $R_\bullet$ be the equivalence groupoid just used. By Lemma~\ref{lem:projectives},
\[
 \Gamma U_n\simeq|\Gamma R_\bullet|.
\]
Its object space is $\Gamma N=\NN$. Its arrows between $e,d$ are the global equivalences between $X_{n,e}$ and $X_{n,d}$. Lemma~\ref{lem:em-eq} computes these spaces explicitly. The calculation uses the explicit equivalence-object representation and the preservation of groupoid quotients.

For every individual $e$, the reset construction gives a tracked equivalence of its group with either $0$ or $\Ftwo$ and hence of its fibre with the corresponding standard Eilenberg-Mac Lane type. Both cases occur. Thus the full global equivalence groupoid has two components, represented by $1$ and $K(\Ftwo,n)$. Their automorphism spaces are contractible and $K(\Ftwo,n)$ respectively, since~\eqref{eq:em-equivalence} and $\Aut_{\Ab}(\Ftwo)=1$ identify the latter with translations. Global sections preserve the iterated groupoid quotients, so the global Eilenberg-Mac Lane objects here are the Eilenberg-Mac Lane spaces in $\Spc$ with the stated coefficients. It follows that
\begin{equation}\label{eq:global-calculation}
 \Gamma U_n\simeq B\Aut(1)\amalg B\Aut(K(\Ftwo,n))
 \simeq1\amalg K(\Ftwo,n+1).
\end{equation}
This is a $\pi$-finite space~\cite{Anel25}: it has two components and only one nonzero positive homotopy group, namely $\Ftwo$ in degree $n+1$ on the second component.

At a code $u_{\mathrm{nt}}$ of a fixed non-total index, the internal version of the same calculation is
\[
 \Omega_{u_{\mathrm{nt}}}U_n
 \simeq\underline{\Aut}(K(\Ftwo,n))
 \simeq K(\Ftwo,n).
\]
Consequently $\pi_{n+1}(U_n,u_{\mathrm{nt}})\cong\Ftwo$ as a group object of $\Eff$. Together with the upper bound, this gives $U_n$ exact truncation height $n+1$. This proves Theorem~\ref{thm:rice}(\ref{rice:global}).
\qed\end{proof}

\subsection{Numerical Observations and Classifier Separation}\label{app:numerical-separation}
\begin{proof}[of Theorem~\ref{thm:rice}(\ref{rice:maps}) and~(\ref{rice:local})]
Let $F:U_n\to N$. Since $\Gamma N$ is discrete, $\Gamma F$ takes a constant numerical value $a$ on the contractible component of $\Gamma U_n$ and a value $b$ on its other component. The composite $Fc_n:N\to N$ lies between objects of the embedded $\Eff$. Full faithfulness identifies it with an assembly map, whose carrier function is total recursive because $N$ has the singleton realization relation $\rho_N(k)=\{k\}$~\cite{Hyland82,vanOosten08}. Its values are
\[
 (Fc_n)(e)=
 \begin{cases}a,&e\in\Tot,\\b,&e\notin\Tot.\end{cases}
\]
If $a\ne b$, comparison with the fixed number $a$ decides $\Tot$, contradicting the nonrecursiveness of $\Tot$ established above. Thus $a=b$. An assembly map $N\to N$ is determined by its carrier values, so $Fc_n$ equals the constant-$a$ map as an arrow. Since $c_n$ is an effective epimorphism, equality of maps into the $0$-type $N$ is detected after its pullback, by exact descent. Hence $F$ is the constant-$a$ map.

Every numeral gives such a map and different numerals give different maps because $U_n$ has a global code. The mapping space into $N$ is discrete. Therefore the constant-map inclusion $\NN\to\Map_\E(U_n,N)$ is an equivalence. The same proof, or composition with the inclusion $2\into N$ in $\Eff$, gives $\Map_\E(U_n,2)\simeq2$. This proves Theorem~\ref{thm:rice}(\ref{rice:maps}). The argument uses the two global components of $U_n$ and shows that every internal numerical observation takes the same value on both.

The constancy result distinguishes the uniform classifier from the classifier generated by the two fixed representatives. The totality problem detects this distinction after arbitrary covers, as follows.
Let $q_n:\mathcal Y_n\to2$ be any family whose fibres over $0$ and $1$ are equivalent to $1$ and $K(\Ftwo,n)$, respectively. Suppose there were a cover $a:V\onto N$, a map $r:V\to2$, and an equivalence $a^*p_n\simeq r^*q_n$. By Lemma~\ref{lem:projectives}, choose a section $s:N\to V$. Pulling back gives
\[
 p_n\simeq(rs)^*q_n.
\]
Thus $rs:N\to2$ takes value zero exactly at total indices, since the two fixed fibre types are not equivalent. It is a recursive Boolean-valued function, a contradiction. This proves Theorem~\ref{thm:rice}(\ref{rice:local}).
\qed\end{proof}

This completes the proof of Theorem~\ref{thm:rice}.

We next compare the uniform classifier with the classifier of the fixed two-fibre family.

\begin{proof}[of Corollary~\ref{cor:two-classifiers}]
Form the univalent completion of $q_n$ with base $W_n$ using Lemma~\ref{lem:em-eq}. Its global equivalence groupoid is the groupoid on the same two representatives, so $\Gamma W_n\simeq1\amalg K(\Ftwo,n+1)$. The label map of its equivalence groupoid factors through the constant groupoid on $2$: an equivalence preserves the coefficient-group isomorphism class, and the isomorphism object between $0$ and $\Ftwo$ in $\Eff$ is empty. Equivalently, the source and target labels agree internally. Realization therefore gives $W_n\to2$, taking different values at its two global codes. An equivalence $W_n\simeq U_n$ would transport it to a nonconstant $U_n\to2$, contradicting Theorem~\ref{thm:rice}(\ref{rice:maps}). Hence the two classifiers are inequivalent.
\qed\end{proof}

The disjoint union in the global formula is computed in the external $\infty$-category of spaces. Internally, the uniform tracked family and its equivalence groupoid retain the realizability information detected by this separation.

\section{Dependent Products from Tracked Cochains}\label{app:cochains}
This appendix proves the general reconstruction theorem for dependent products from tracked cochains (Theorem~\ref{thm:cochain-products}) and determines the exact quotient depth of products built from the reset family (Corollary~\ref{cor:cochain-depth}). We work in the higher exact completion \(\E=\Ex(\Asm(K_1))\). Its fully faithful embedding in
\(\HH=\Sh(\Asm(K_1),\Spc)\) preserves finite limits, images, and
effective groupoid quotients, and every object admits an assembly cover.
Choices of carrier elements are made in the working external universe.
Tracked cochains give an explicit calculation of the dependent products given by Theorem~\ref{thm:real-hierarchy}. Starting from the coefficient operations, we construct section and identity codes and reconstruct the product from this finite tower of section codes. The quotient-depth calculation measures this reconstruction using finite limits, images, and effective groupoid quotients, as in Definition~\ref{def:cochain-depth}.

Let $B$ be an assembly and let $A\to B$ be an abelian group object in $\Asm(K_1)_{/B}$ with $2\id_A=0$. Its group operations have uniform partial recursive trackers. Put $G_B=B\times C_2$, where $C_2=\{1,t\}$ has $t^2=1$ and the standard finite decidable assembly presentation. The coefficient action is trivial. Iterated bar quotients define $K_B(A,0)=A$ and $K_B(A,r+1)=B_BK_B(A,r)$; write $B_BG_B$ for the bar quotient of $G_B$. Exactness gives these constructions. Products of objects over $B$ below denote finite fibre products by default.

\begin{theorem}[Dependent products from tracked cochains]\label{thm:cochain-products}
For every $n\ge1$, consider the structure map $f:B_BG_B\to B$ and the second projection $g_n:K_B(A,n)\times_B B_BG_B\to B_BG_B$.
The map $A\to B$ is a weak section family for $(f,g_n)$, with evaluation induced by the normalized cocycle
\begin{equation}\label{eq:cochain-code}
 c_a(h_1,\ldots,h_n)=
 \begin{cases}a,&h_1=\cdots=h_n=t,\\0,&\text{otherwise}.\end{cases}
\end{equation}
Finite fibre powers of $A$ provide independently tracked codes for every successive identity problem. Coherent fullness reconstructs a dependent product
\begin{equation}\label{eq:cochain-product}
 Q_n(A):=\Pi_fg_n\simeq\prod_{j=0}^{n}K_B(A,j)
 \quad\text{in }\E_{/B}.
\end{equation}
For every $T\to B$ in $\E$, evaluation induces a natural equivalence of mapping spaces
\begin{equation}\label{eq:cochain-adjunction}
 \Map_{\E/B}(T,Q_n(A))\simeq
 \Map_{\E/B}(T\times_B B_BG_B,K_B(A,n)).
\end{equation}
The construction commutes with arbitrary base change in $\E$ and uses at most $n$ nested effective groupoid-quotient levels.
\end{theorem}
For $A=C_2$ over the terminal assembly, the code object has two discrete elements, while each component of $Q_2(A)$ has nonzero $\pi_1$ and $\pi_2$. The quotient steps reconstruct identities between the represented functions and their higher identities. The classical cochain calculation~\cite{StacksDK,KerodonDK} gives uniform tracked codes directly from the coefficient operations.

\begin{definition}[Quotient depth]\label{def:cochain-depth}
Construction expressions in $\E_{/B}$ are generated from objects of the embedded slice $\Eff_{/B}$ by finite limits, images, and effective Segal groupoid quotients. Their depths are defined recursively. An initial expression has depth zero. A finite-limit or image expression has depth equal to the maximum depth of its inputs, with the empty maximum taken to be zero. The quotient formation rule takes an effective Segal groupoid of expressions whose level depths $d_k$ have finite $\sup_{k\ge0}d_k$, and assigns its quotient expression depth $1+\sup_{k\ge0}d_k$. For an object $X$ admitting such an expression, its quotient depth is the least depth of a construction expression whose value is equivalent to $X$.
\end{definition}

\subsection{Reconstruction from a Finite Tower of Section Codes}

Recall that a weak section family for $Z\xrightarrow gY\xrightarrow fX$
is an object $W\to X$ and an evaluation $W\times_XY\to Z$ over $Y$
such that every parameterized section becomes one of these evaluations
after a cover of its parameter object. The local representation includes a witnessing homotopy over $Y$
and therefore retains the identity data used in the induction.

\begin{lemma}[Reconstruction from a finite tower of section codes]\label{lem:cc-finite}
Let $\mathcal C$ be a small regular category and
$\mathcal H=\operatorname{Sh}_{\mathrm{reg}}(\mathcal C,\Spc)$.
Let $\mathcal C\subseteq\mathcal D\subseteq\mathcal H$ be full regular
embeddings, with $\mathcal D$ exact and the embedding $\mathcal D\hookrightarrow\mathcal H$ preserving its groupoid quotients.
For $n\ge-1$, consider $Z\xrightarrow gY\to X$ in
$\mathcal D$, with $g$ $n$-truncated. Suppose one has independently constructed weak section families
for this problem and, recursively, for the identity problem between every
pair of universal codes, down to the $(-1)$-truncated stage. Then
$\Pi^{\mathcal H}_{Y\to X}g$ belongs to $\mathcal D$ and is its dependent
product there. An image and at most $n+1$ nested groupoid quotients suffice.
\end{lemma}
\begin{proof}
Let $S=\Pi^{\mathcal H}_{Y\to X}g$. The code evaluation induces a cover
$a:W\onto S$: weak adequacy is local surjectivity over the
representable tests of the regular site. The same conclusion holds if
weak adequacy has been checked on a class covering those tests.

When $n=-1$, $S\to X$ is monic and $S=\im(W\to X)$, which is in
$\mathcal D$. For $n\ge0$ set $D=W\times_XW$ and $Y_D=Y\times_XD$.
Let $I\to Y_D$ be the identity family between the two universal
evaluations. It is a pullback of $\Delta_g$ and is $(n-1)$-truncated.
In the ambient slice,
\begin{equation}\label{eq:cc-kernel}
 W\times_SW\simeq\Pi^{\mathcal H}_{Y_D\to D}(I\to Y_D).
\end{equation}
Both sides classify homotopies between parameterized sections, together
with all higher homotopies; this is the ambient dependent-product
adjunction. The recursive code hypothesis and induction represent the
right side. Consequently the full \v{C}ech nerve of $W\onto S$ lies in
$\mathcal D$, since its terms are finite limits of $W$ and this kernel.
Its effective quotient lies in $\mathcal D$ and agrees with $S$ in
$\mathcal H$. Full faithfulness restricts the ambient adjunction.
Each identity lowers the truncation degree by one and each inductive step
uses one quotient, proving the bound.
\qed\end{proof}

The indicated tower of section problems reconstructs the product by the
section-identity construction of Theorem~\ref{thm:relative-fullness}.

\subsection{Bar Constructions over Projective Parameters}

Every object of $\E$ has a partitioned assembly cover by
Lemmas~\ref{real:ordinary-source} and~\ref{pt:projectives}.
We use these projective parameters to evaluate the bar constructions
from Appendix~\ref{app:em}.

\begin{lemma}[Projective evaluation]\label{lem:cc-projective}
A partitioned assembly is projective against covers in $\E$. If
$P\to B$ is partitioned, then
\[
 \Gamma_P:\E_{/P}\longrightarrow\Spc,
 \qquad X\longmapsto\Map_{\E/P}(P,X)
\]
preserves finite limits, covers, and effective groupoid quotients.
In particular, if $A_P=A\times_BP$ and
$H_P=\operatorname{Hom}_{\Asm/B}(P,A)$, then
\begin{equation}\label{eq:cc-projective-em}
 \Map_{\E/B}(P,K_B(A,r))\simeq K(H_P,r),\qquad r\ge0.
\end{equation}
\end{lemma}
\begin{proof}
Projectivity of $P$ in $\E$ follows from
Lemmas~\ref{real:ordinary-source} and~\ref{pt:projectives}.

Finite limits are preserved by a represented mapping-space functor.
Every point of the target of a cover pulls it back to a cover over $P$,
which splits. Hence $\Gamma_P$ sends covers to effective epimorphisms of
spaces. An effective groupoid is the \v{C}ech nerve of its quotient cover.
Preserving that nerve and cover preserves its realization. Iterating this
argument on the bar constructions gives~\eqref{eq:cc-projective-em}.
At $r=0$, full faithfulness identifies the group of maps with $H_P$.
\qed\end{proof}

We can now compute section spaces over these parameters by mapping out
of the bar presentation of $B_BG_B$.

\begin{lemma}[Sections on a projective parameter]\label{lem:cc-section-test}
For a partitioned $P\to B$ there is a natural equivalence
\begin{equation}\label{eq:cc-test}
 \Map_{\E/B}(P\times_B B_BG_B,K_B(A,n))
 \simeq\Map_{\Spc}(BC_2,K(H_P,n)).
\end{equation}
\end{lemma}
\begin{proof}
Mapping out of the bar quotient gives the limit of the cosimplicial spaces
\[
 [k]\longmapsto\Map_{\E/B}(P\times C_2^k,K_B(A,n)).
\]
Each $P\times C_2^k$ is partitioned. By
Lemma~\ref{lem:cc-projective}, its displayed mapping space is
$K(\operatorname{Hom}_{\Asm/B}(P\times C_2^k,A),n)$. Finite case distinction on the
decidable $C_2$ coordinates gives, naturally,
\[
 \operatorname{Hom}_{\Asm/B}(P\times C_2^k,A)=H_P^{C_2^k}.
\]
These are precisely the levels and cosimplicial maps obtained by mapping
the nerve of $C_2$ into $K(H_P,n)$. Taking limits proves the claim.
\qed\end{proof}

\subsection{Cocycle and Identity Codes from Coefficient Operations}

\paragraph*{Normalized cochain calculation.}
We recall the elementary cochain calculation, including the identity data
needed below. Let $H$ be any abelian group with $2H=0$, with trivial
$C_2$-action. For $q>0$, $C^q(C_2;H)$ consists of normalized cochains, namely functions $C_2^q\to H$ vanishing whenever an argument is $1$.
Evaluation at $(t,\ldots,t)$ identifies $C^q(C_2;H)=H$; put $C^0(C_2;H)=H$ and $Z^q(C_2;H)=\ker(\delta:C^q(C_2;H)\to C^{q+1}(C_2;H))$.
For $q>0$ the coboundary is
\begin{align*}
 (\delta u)(g_1,\ldots,g_{q+1})={}&u(g_2,\ldots,g_{q+1})\\
 &+\sum_{i=1}^{q}(-1)^i
 u(g_1,\ldots,g_i g_{i+1},\ldots,g_{q+1})\\
 &+(-1)^{q+1}u(g_1,\ldots,g_q).
\end{align*}
At $(t,\ldots,t)$ all middle terms vanish. The remaining coefficient is
$1+(-1)^{q+1}$, which is either zero or two. Moreover
$(\delta h)(t)=h-h=0$ in degree zero. Therefore the full normalized complex is
\begin{equation}\label{eq:cc-zero-complex}
 H\xrightarrow{0}H\xrightarrow{0}H\xrightarrow{0}\cdots.
\end{equation}

\begin{lemma}[Cochain codes as a coherent tower]\label{lem:cc-cochains}
For $n\ge1$, the mapping space $\Map(BC_2,K(H,n))$ has the following natural
coherent model. Its objects are $a_n\in H$. An arrow consists of a pair
of objects with equal $a_n$ coordinates and a choice of $a_{n-1}\in H$.
When $n\ge2$, a $2$-arrow consists of a pair of parallel arrows with equal
$a_{n-1}$ coordinates and a choice of $a_{n-2}\in H$. This continues
down to $a_0\in H$; the next level is equality. Compositions are addition. In particular every
component of the mapping space has a representative~\eqref{eq:cochain-code},
and this same assertion holds for every successive identity space.
\end{lemma}
\begin{proof}
Use the simplicial abelian model $\operatorname{DK}(H[n])$ for $K(H,n)$, where $H[n]$
is the homological complex concentrated in degree $n$. A simplicial map
from a simplicial set $L$ to this model is a chain map from the normalized
free chains $N\mathbb Z[L]$ to $H[n]$, hence an $n$-cocycle. Homotopies
are modeled by degree-one chain homotopies; iterated homotopies use the
subsequent degrees. More precisely the derived mapping space is modeled
by Dold-Kan applied to the connective truncation of the chain Hom
complex. This follows by the free simplicial abelian adjunction,
normalization, and the chain homotopy equivalence between chains on a
product and the tensor product of chains. The free chains of $L$ compute the derived mapping space directly.
These are the standard normalization and homotopy constructions of
Dold-Kan \cite{StacksDK,KerodonDK}.

For $L=BC_2$, the degree-zero term of this connective complex is
$Z^n(C_2;H)$. Its degree-$i$ term is $C^{n-i}(C_2;H)$ for
$1\le i\le n$.
Thus $a_n$ occupies homological degree zero, while $a_{n-i}$ occupies
degree $i$ and records an $i$-fold identity, for $1\le i\le n$.
All its differentials vanish by~\eqref{eq:cc-zero-complex}, and
$Z^n(C_2;H)=H$. Thus the complex is $H$ in each degree $0,\ldots,n$
with zero differential. Its strict abelian higher groupoid has exactly
the listed coordinates and compositions. This Dold-Kan model is equivalent to $\Map(BC_2,K(H,n))$. The equivalence is natural in $H$, since each chain construction above is functorial in $H$.
\qed\end{proof}

\begin{proposition}[Uniformly tracked section and identity codes]\label{prop:cc-codes}
The normalized cocycle
\begin{equation}\label{eq:cc-cocycle}
 c_a(h_1,\ldots,h_n)=
 \begin{cases}a,&h_1=\cdots=h_n=t,\\0,&\text{otherwise}\end{cases}
\end{equation}
defines an evaluation
\[
 e_n:A\times_B B_BG_B\longrightarrow K_B(A,n)
\]
over $B$. With evaluation $(e_n,\operatorname{pr}_2)$ over $B_BG_B$, $A\to B$ is a weak section family for Theorem~\ref{thm:cochain-products}.
The coordinate descriptions in Lemma~\ref{lem:cc-cochains} give independent
weak section families for every recursively induced identity problem, all the
way to the monic stage.
\end{proposition}
\begin{proof}
Formula~\eqref{eq:cc-cocycle} is a cocycle by the direct calculation above.
Apply its simplicial cocycle map into $\operatorname{DK}(A[n])$ fibrewise. In each
simplicial degree the formulas use finitely many additions, zero,
and tests of the finitely many $C_2$ coordinates. They are therefore
assembly morphisms, uniformly tracked from the group trackers of $A$.
The realization of $\operatorname{DK}(A[n])$ in the ambient topos is the iterated bar
object $K_B(A,n)$: normalization identifies both with the same
Eilenberg-Mac Lane object. This also constructs $e_n$ by fullness of
$\E\into\HH$. The same finite formulas for lower cochains construct
the homotopies and all their higher identity evaluations.

For compatibility with projective evaluation, choose the standard
$n$-fold bar model for this comparison. Normalization and the
Alexander-Whitney and shuffle chain maps compare its diagonal with
$A[n]$, by finite sums in each simplicial degree. They induce mutually
inverse equivalences on realizations, and the $n$-fold bar realization
is computed by $n$ successive groupoid quotients. Consequently the
cocycle and lower-cochain evaluations transported to this bar model
commute with projective evaluation: finite products and the indicated
groupoid quotients are preserved by Lemma~\ref{lem:cc-projective}.
Under~\eqref{eq:cc-projective-em} they are precisely the cocycle
and cochain maps of Lemma~\ref{lem:cc-cochains}.

For adequacy, take any $T\to B$ in $\E$ and a section over
$T\times_B B_BG_B$. Choose a partitioned assembly cover $P\onto T$.
Lemmas~\ref{lem:cc-section-test} and~\ref{lem:cc-cochains} say that the
restricted section is homotopic to $e_n(a,-)$ for some
$a\in H_P=\operatorname{Hom}_{\Asm/B}(P,A)$. This is exactly a map $P\to A$ selecting a code
and the required homotopy. For an identity problem, cover its parameter
object by a partitioned assembly in the same way. The full natural
cochain model identifies its identity space; the next cochain coordinate
represents each of its components. This gives the map selecting the next code and its
specified homotopy. Repeating proves the assertion at all $n+1$ levels.

Diagonals implement the equality conditions on earlier coordinates. All
code objects are finite fibre powers of $A$ with repeated-coordinate
maps, so the construction applies to general tracked coefficient
assemblies, including the reset family of Appendix~\ref{app:reset}.
\qed\end{proof}

\subsection{Reconstruction and the Height-Two Calculation}\label{app:cochain-reconstruction}

\begin{proof}[of Theorem~\ref{thm:cochain-products}]
Proposition~\ref{prop:cc-codes} establishes the hypotheses on the finite tower of section codes in
Lemma~\ref{lem:cc-finite}. That lemma represents the ambient section object
by successive images and effective groupoid quotients in $\E$. This
proves existence and, for every $T\to B$ in $\E$, the full adjunction
\begin{equation}\label{eq:cc-adjunction}
 \Map_{\E/B}(T,Q_n(A))\simeq
 \Map_{\E/B}(T\times_B B_BG_B,K_B(A,n)).
\end{equation}
The tracked coefficient operations thus generate both the codes and their reconstructed dependent product.

The quotient description is explicit. The code coordinates
are $a_n,a_{n-1},\ldots,a_0$, with zero boundary maps and
additive compositions. At the bottom, the $a_0$ section problem is already
represented by its code object: fibrewise, maps from the connected type
$BC_2$ to a discrete group are constant, as is also immediate from
Lemma~\ref{lem:cc-section-test}. Its equality stage and set quotient leave the represented object
unchanged up to equivalence. The next quotient makes the $a_0$ automorphisms into $K_B(A,1)$.
Continuing upward deloops each lower coordinate once at each subsequent
level. Finally $a_n$ gives $A$, $a_{n-1}$ gives $K_B(A,1)$, and $a_0$ gives
$K_B(A,n)$. Since the boundary maps vanish and the coordinate actions are trivial, the
result is their finite product~\eqref{eq:cochain-product}.

Equivalently, the zero-differential complex in
Lemma~\ref{lem:cc-cochains} is the direct sum of $A[j]$, $0\le j\le n$.
The finite product of their iterated bar realizations is the same
coherent code quotient. The naturality in projective tests ensures
compatibility with the code evaluations and groupoid composition;
descent from the projective covers extends the identification to every
higher test object. This identifies the objects together with their full mapping spaces
and higher identities.

This construction uses at most $n$ quotient levels. Trivial coefficient groups yield contractible factors.
All finite limits, images, and effective groupoid quotients used are
pullback-stable and preserved by the ambient embedding. Ambient
Beck--Chevalley then identifies the pullback of the representing object
and its evaluation with the corresponding section problem over $B'$.
This proves the base-change assertion for every $B'$ in $\E$.
\qed\end{proof}

\begin{example}[$n=2$ case of Theorem~\ref{thm:cochain-products}]
Here is the complete code tower for $n=2$. Write $A_a,A_b,A_c$ for copies of $A$ with coordinates $a,b,c$, respectively. Let $S$ denote the ambient section object. The section codes are
$W_0=A_a$, with cocycle $c_a(t,t)=a$. Over
$X_1=A_a\times_BA_{a'}$, take
\[
 W_1=A_a\times_BA_b\longrightarrow X_1,
 \qquad(a,b)\longmapsto(a,a).
\]
The evaluation is the normalized $1$-cochain $b(t)=b$, a self-homotopy
of $c_a$, since $(\delta b)(t,t)=2b=0$. The next code base is
$X_2=W_1\times_{X_1}W_1\simeq A_a\times_BA_b\times_BA_{b'}$.
Take
\[
 W_2=A_a\times_BA_b\times_BA_c\longrightarrow X_2,
 \qquad(a,b,c)\longmapsto(a,b,b).
\]
The $0$-cochain $c$ specifies a modification because $\delta c=0$.
The final identity condition is $c=c'$, the relative diagonal of this
$0$-truncated code family. The image base case handles precisely this
proposition.

The $c$-codes already represent the second identity section object.
Their quotient one level up gives
\begin{equation}\label{eq:cc-kernel-two}
 W_0\times_SW_0\simeq
 A_a\times_B A_b\times_B K_B(A,1)
 \longrightarrow A\times_BA,
\end{equation}
with structure map $(a,b,z)\mapsto(a,a)$. Its composition adds the
$b$ and $K_B(A,1)$ coordinates. The final step takes the quotient of the groupoid with
object family $A_a$ and automorphism group object $A_b\times_B K_B(A_c,1)$.
It is
\[
 A_a\times_B B_B\!\left(A_b\times_B K_B(A_c,1)\right)
 \simeq A\times_B K_B(A,1)\times_B K_B(A,2).
\]
For $B=1$ and $A=C_2$, the starting code object has two
discrete elements, while each component of the resulting product has
nonzero $\pi_1$ and $\pi_2$. The two quotient levels reconstruct these higher identities from the
discrete codes.
\end{example}

\subsection{The Reset Family and Exact Quotient Depth}\label{app:cochain-depth}
Appendix~\ref{app:reset} constructs the uniformly tracked group assembly
$\mathcal M\to N$, proves $2\mathcal M=0$, and establishes the fibre
isomorphisms stated in Section~\ref{sec:application}. The present section uses those
tracked group operations directly. The following truncation
estimate gives the lower bound for the quotient depth of
Definition~\ref{def:cochain-depth}.

\begin{lemma}[Depth bounds truncation]\label{lem:cc-depth}
An object of $\E_{/B}$ constructed with quotient depth at most $d$
is $d$-truncated over $B$.
\end{lemma}
\begin{proof}
Objects of the embedded slice $\Eff_{/B}$ are $0$-truncated. Finite limits preserve a common
nonnegative truncation bound. An image is a subobject of its codomain,
so has that same bound. Suppose the levels of an internal groupoid
$R_\bullet$ are $d$-truncated over $B$, and set $Q=|R_\bullet|$.
After the cover $R_0\times_BR_0\onto Q\times_BQ$, its relative
diagonal is $R_1\to R_0\times_BR_0$. This map is $d$-truncated:
maps between $d$-truncated objects have $d$-truncated homotopy fibres.
Truncation descends along covers, so $\Delta_{Q/B}$ is $d$-truncated
and $Q/B$ is $(d+1)$-truncated. Structural induction proves the result.
\qed\end{proof}

We now apply the cochain construction to the reset family and use this estimate to determine the exact quotient depth of the resulting function types.

\begin{proof}[of Corollary~\ref{cor:cochain-depth}]
Apply Theorem~\ref{thm:cochain-products} with $B=N$ and
$A=\mathcal M$. The cocycle and identity-code formulas use finite
combinations of the group trackers from Appendix~\ref{app:reset} and
decidable tests of $C_2$ coordinates. Equalities between earlier coordinates
are imposed by diagonal maps. The construction is therefore uniform in $e$.
The fibre isomorphisms from Section~\ref{sec:application}, the product
formula~\eqref{eq:cochain-product}, and pullback stability give
\eqref{eq:cochain-reset} as equivalences of realizability types.

The fibre at a total index is contractible; each fibre at a non-total index has two
components because its degree-zero factor is $\Ftwo$. Thus the
contractible-fibre index set is exactly $\Tot$, which is
$\Pi^0_2$-complete under computable many-one reductions by the argument
in Appendix~\ref{app:projective}.

The product formula constructs $\mathcal Q_n$ with depth at most $n$:
$K_N(\mathcal M,j)$ uses $j$ iterated bars, and a finite product takes
the maximum of their depths. Choose an index $e$ of a nowhere-defined
function. At the zero basepoint of the fibre $(\mathcal Q_n)_e$,
\[
 \pi_n((\mathcal Q_n)_e)=\Ftwo\ne0.
\]
Thus $\mathcal Q_n\to N$ has truncation height $n$.
Lemma~\ref{lem:cc-depth} gives quotient depth at least $n$, and the
product construction attains this bound. Starting with the $0$-truncated
code object $\mathcal M\to N$, the reconstruction produces a section
object with $\pi_n\cong\Ftwo$ in the chosen fibre.
\qed\end{proof}

\end{document}